\documentclass[11pt,letterpaper]{article}
\usepackage[utf8]{inputenc}
\usepackage{subcaption}
\usepackage{amsmath}
\usepackage{amsfonts}
\usepackage{amssymb}
\usepackage{mathtools}
\usepackage{graphicx}
\usepackage{array}
\usepackage{xfrac}
\usepackage{xcolor}
\usepackage[margin=1in]{geometry}
\usepackage[ruled]{algorithm}
\usepackage{algpseudocode}
\usepackage{amsthm}
\usepackage{caption}
\usepackage{tikz}
\usepackage[T1]{fontenc}
\usepackage{microtype}
\usepackage{kpfonts}
\usepackage{booktabs}
\usepackage{soul}
\usepackage{authblk}

\DeclareSymbolFont{yhlargesymbols}{OMX}{yhex}{m}{n}
\DeclareMathAccent{\yhwidehat}{\mathord}{yhlargesymbols}{"62}
\renewcommand{\widehat}[1]{\yhwidehat{#1}}

\makeatletter
\newcommand\fs@booktabsruled{%
  \def\@fs@cfont{\bfseries\strut}\let\@fs@capt\floatc@ruled
  \def\@fs@pre{\hrule height\heavyrulewidth depth0pt \kern\belowrulesep}%
  \def\@fs@mid{\kern\aboverulesep\hrule height\lightrulewidth\kern\belowrulesep}%
  \def\@fs@post{\kern\aboverulesep\hrule height\heavyrulewidth\relax}%
  \let\@fs@iftopcapt\iftrue
}
\makeatother

\floatstyle{booktabsruled}
\restylefloat{algorithm}

\usetikzlibrary{quantikz2}
\tikzset{
  operator/.append style={
    minimum width=8mm,
    minimum height=8mm,
    inner xsep=1pt,
    inner ysep=1pt
  }
}

\makeatletter
\newtheoremstyle{thmwithbreak}
  {3pt}{3pt}
  {\itshape}
  {}
  {\bfseries}
  {.}
  {0.5em}
  {\thmname{#1}\thmnumber{ #2}\thmnote{ (\textbf{#3})}\par}
\makeatother

\theoremstyle{thmwithbreak}

\newtheorem{theorem}{Theorem}
\newtheorem{definition}[theorem]{Definition}
\newtheorem{lemma}[theorem]{Lemma}

\usepackage{hyperref}
\hypersetup{colorlinks=true,citecolor=blue,linkcolor=blue,filecolor=blue,urlcolor=blue,breaklinks=true}

\allowdisplaybreaks

\begin{document}

\title{Quantum Phaselift}

\author[1,2]{Dhrumil Patel\thanks{djp265@vt.edu}}
\author[1]{Laura Clinton}
\author[1,2]{Steven T. Flammia}
\author[1,3]{Raúl García-Patrón}

\affil[1]{Phasecraft}
\affil[2]{Virginia Tech}
\affil[3]{University of Edinburgh}

\maketitle

\begin{abstract}
Estimating quantum time-series such as the Loschmidt amplitude $f(t)=\langle\psi|\mathrm{e}^{-\mathrm{i}Ht}|\psi\rangle$ is central to spectroscopy, Hamiltonian analysis, and many phase-estimation algorithms.
Direct estimation via the Hadamard test requires controlled implementations of $\mathrm{e}^{-\mathrm{i}Ht}$, and the depth of these controlled circuits grows with $t$, making long-time estimation challenging on near-term hardware.
We introduce Quantum Phaselift, a lifting-based framework that estimates the rank-one matrix $Z = f f^\dagger$ rather than estimating $f$ directly.
We propose simple quantum circuits for estimating the entries of $Z$ and show that measuring only a narrow band of this matrix around the diagonal is sufficient to uniquely recover $f$.
Crucially, this reformulation decouples the controlled circuit depth from the maximum evolution time to scale instead with the width of the measured band.
We prove that a $O(1)$ bandwidth suffices for generic signals, leading to substantial savings in controlled operations compared to direct estimation methods.
We develop three recovery algorithms with provable exact recovery in the noiseless setting and stability under measurement noise.
Finally, we numerically demonstrate that high-quality recovery is possible for the 2D Fermi-Hubbard and 2D transverse-field Ising model signals of size exceeding 100 time points using only a few million measurement shots and reasonable post-processing time, making our time-series estimation techniques efficient and effective for near-term implementations.
\end{abstract}

\tableofcontents

\section{Introduction}

Materials science and chemistry are among the most important fields where quantum computation has a strong potential to provide practical quantum advantage~\cite{Daley2022}.
Quantum hardware has reached a level of quality that non-trivial time-dynamic simulations can be implemented on accessible quantum hardware~\cite{Alam2025, Alam2025a}.
In parallel, classical approaches to evaluating spectral properties of large many-body Hamiltonians face fundamental bottlenecks, as the computational cost of simulating long-time dynamics grows rapidly with system size~\cite{LeBlanc2015,Innerberger2020}. 
We can therefore hope to see in the next few years quantum simulations of direct spectroscopy experiments that could challenge their classical alternatives. 

Many spectral quantities of physical interest, including spectral densities and response functions, can be directly obtained by applying signal processing techniques to time-series (signal) derived from time evolution. 
A canonical example is the Loschmidt amplitude (or, fidelity amplitude),
\begin{equation}\label{eq:L-amplitude}
    f(t)=\langle\psi|\mathrm{e}^{-\mathrm{i}Ht}|\psi\rangle,
\end{equation}
where $H$ is a Hamiltonian, $|\psi\rangle$ is a known input state, and $t$ denotes time. 
This signal encodes the spectral information of $H$ induced by $|\psi\rangle$ and therefore serves as a building block for various quantum phase estimation algorithms~\cite{Somma2019, OBrien2019, Lin2022, Ding2023, Wang2023, Yi2024, Yang2024, Clinton2024}.
Beyond phase estimation, this signal arises in a broad range of physical contexts, including dynamical quantum phase transitions~\cite{Heyl2013, Zvyagin2016, Budich2016}, quantum chaos and stability analysis~\cite{Peres1984, Prosen2003}, particle production~\cite{Schwinger1951}, lattice gauge theories~\cite{Martinez2016}, and fluctuation theorems~\cite{Talkner2007, Silva2008}.
More general dynamical observables, such as out-of-time-order correlators and multi-time correlation functions, can be viewed as structured generalizations of~\eqref{eq:L-amplitude}, involving multiple time evolutions, intermediate operators, and distinct input and output states~\cite{mottaSubspaceMethodsElectronic2024, yoshiokaKrylovDiagonalizationLarge2025, yuQuantumCentricAlgorithmSampleBased2025, Abanin2025,Flew2025}. 
Across these settings, the central computational task is the reliable estimation of time-dependent complex-valued signals generated by quantum dynamics.
In this paper, we focus on the signal of the form in~\eqref{eq:L-amplitude}.

A standard approach to estimating the above time-series is via the Hadamard test~\cite{Somma2002,Somma2019}. 
While conceptually straightforward, this method requires implementing a controlled version of the full time-evolution operator $\mathrm{e}^{-\mathrm{i} H t}$ using an ancillary qubit. 
This presents a severe hardware limitation: as the evolution time $t$ increases, the depth of this controlled operation grows accordingly. 
This imposes coherence and control requirements that are difficult to satisfy on near-term quantum devices.
This is not something we can hope to remove by clever compilation in general, since the no-fast-forwarding theorem rules out sublinear-in-$t$ simulation for generic Hamiltonians~\cite{Berry2006}.
So if the goal is long-time signals (hence high spectral resolution), the naive pipeline becomes hardware-limited quickly.
Consequently, this challenge has motivated a substantial body of work aimed at reducing the cost of \emph{controlled time evolution} or eliminating it altogether.

\subsection{Related Works and the Trade-Offs}
\label{sec:related-work}

Reducing or eliminating the controlled time-dynamics from time-series estimation is an area of rapidly growing interest. 
While current hardware can arguably implement the time dynamics of a few Trotter steps, controlling this evolution with an extra ancillary qubit, as required by the Hadamard test, often renders the implementation impractical due to decoherence. 
Existing solutions typically address this by trading the control depth for other expensive resources. We categorize these approaches below.

\paragraph{Doubling the Qubits.} 
One folklore approach to removing the control qubit involves doubling the system size. 
In this scheme, the system is duplicated into two registers, and the uncontrolled time-dynamics are applied to only one register, sandwiched between two controlled-SWAP operations acting on both registers~\cite{Huggins2020}.
While this successfully replaces the time-dependent control depth with a fixed-depth operation (the controlled-SWAP), it necessitates a two-fold increase in qubit count.
Furthermore, it requires, as a condition, the knowledge and easy preparation of an eigenstate of the Hamiltonian of interest, which is used as the input to the second register. 
While this assumption is reasonable for certain physically motivated systems, for instance, Hamiltonians for which the vacuum is a known eigenstate, it breaks down in more general settings, limiting the applicability of the method to a restricted class of Hamiltonians.

\paragraph{Complexity of Reference State Preparation.}
A second class of approaches relies on preparing a superposition state of the form $|\psi\rangle+ \mathrm{e}^{\mathrm{i}\theta}|R\rangle$, where $|R\rangle$ is an eigenvector of the Hamiltonian that serves as a phase reference~\cite{Lu2021,Russo2021,Schuckert2023,Hemery2024}. 
Classical post-processing algorithms are then used to extract the time series. 
The primary drawback here is the complexity of state preparation: even if $|R\rangle$ is trivial (e.g., the all-zeros state), creating the superposition typically requires a controlled implementation of the state-preparation unitary $U_{\psi}$ or similar operations.
Thus, the burden of control is often shifted from the time-evolution operator to the state-preparation operator, rather than being eliminated.

\paragraph{Purely Ancillary-Free Methods.}
There has also been a recent proposal without the need for any sort of controlled operations, making ancillary-free quantum phase-estimation closer to the range of applicability of traditional quantum phase-estimation.  
In~\cite{Chan2024}, Chan et al.\ leveraged the power of the classical shadows technique to estimate energy gaps via the computation of expectation values of certain operators as a function of time. 
A downside of the technique and another similar technique presented in~\cite{Lee2025} is that it cannot directly provide the spectrum but only energy gaps, and that it cannot retrieve time-series. 
Alternatively, Yang et al. showed that one can remove the ancillary qubit by combining real-time and imaginary-time evolution \cite{Yang2024}, exploiting ideas from complex function analysis. 
Its main limitation is the need for an additional implementation of imaginary-time evolution on a quantum device itself, which currently works only for input states with finite correlation length and short evolution times, but improvement in this direction could render the technique more widely applicable.

\paragraph{Phase Retrieval Approaches.}
Last year, some of the authors of this paper introduced two new approaches that adapt classical phase-retrieval techniques to remove controlled time evolution in the estimation of quantum time-series \cite{Clinton2024}. 
Their first technique, adapted from vectorial phase retrieval~\cite{Raz2013}, eliminates the controlled time-dynamics at the expense of having to prepare many auxiliary GHZ-like input states. 
While this method shares some of the limitations of the aforementioned reference-state–based techniques, it has an important advantage that the required input states are problem-independent and can be prepared in logarithmic circuit depth when mid-circuit measurements are available.
However, the recovery guarantees rely on additional structural assumptions on the underlying spectrum, such as well-defined spectral support and spectral independence between auxiliary signals, and improved performance is achieved primarily by increasing the number of auxiliary input states and measurements, leading to significant state-preparation and sampling overhead. 

\paragraph{Sequential Hadamard Test.}
More recently, Schiffer et al.~\cite{Schiffer2025} proposed an approach known as the Sequential Hadamard Test, which reduces the required controlled time evolution to a single gate at each step.
Although the quantum circuit they exploit has some similarities to the one presented in our paper, the authors investigate and estimate a different signal: $g(\ell) = \langle \psi | U(\ell) | \psi \rangle = \langle \psi | U_\ell \cdots U_1 | \psi \rangle$, where $U(\ell)$ is a unitary and can be written as a decomposition of $\ell$ simpler unitaries $U_1, \ldots, U_\ell$. 
Their reconstruction algorithm is based on a telescoping product procedure, which shares conceptual similarities with our algebraic estimator when specialized to the case $K = 1$. 
A limitation of this approach is that the total number of shots scales with the number of gates in the circuit, in addition to other problem parameters, which can make practical implementations challenging. 
We believe that the reconstruction techniques developed in this work can be used to improve the robustness and efficiency of the classical post-processing in the Sequential Hadamard Test. 
Conversely, ideas from that framework may also be leveraged to further reduce the cost of controlled time evolution in our own proposal.

\bigskip
\noindent
In a nutshell, there is always a trade-off, and determining which technique works best requires a careful, case-by-case analysis. 
A crucial aspect that is often overlooked in this landscape is the scaling, numerical stability, and recovery guarantees of the classical post-processing required to reconstruct the desired signal from the measurement data.

\subsection{Our Contribution: Quantum Phaselift}

In classical signal processing, the \emph{Phaselift} algorithm is well celebrated for its rigorous recovery guarantees and numerical robustness against noise in phase retrieval problems~\cite{Candes2013}. 
We briefly review this algorithm in Section~\ref{sec:phaselift-review}. 
While the Phaselift framework has been adapted to quantum state tomography~\cite{Lu2015}, its potential for quantum phase estimation has remained unexplored. 
In this work, we bridge this gap. 
We demonstrate that the mathematical principles of Phaselift lead to natural hybrid quantum-classical algorithms for time-series estimation that offer both low control depth and rigorous stability guarantees.

\paragraph{The Lifting Principle.}
Similar to the classical Phaselift algorithm, our approach as well essentially "lifts" the problem from the recovery of the discretized complex time-series vector $f \in \mathbb{C}^T$, with entries $f_i \coloneqq f(i \Delta)$ for $i = 0, 1, \ldots, T-1$ and $\Delta, T > 0 $, to the recovery of the associated rank-one outer-product matrix
\begin{equation}
    Z \coloneqq ff^\dagger.
\end{equation}
This new framework, as we will see later, reduces the required control time-evolution depth to a small constant, or even a single time step, while inheriting the noise resilience characteristic of Phaselift. 
In this paper, we show that the circuit only requires the preparation of the state $\mathrm{e}^{-\mathrm{i}Ht}|\psi\rangle$, a controlled time-dynamics over a single or few time-steps, and a final projection onto the state $|\psi\rangle$ in the system register. 
An additional advantage of the circuit is that it inherits the error-mitigation resilience of 
\textit{verified quantum phase-estimation}~\cite{OBrien2021}. 

\subsubsection{Summary of Results}
\label{sec:summary}

In this paper, we study the recovery of a complex-valued time-correlation signal $f(t) = \langle \psi | \mathrm{e}^{-\mathrm{i} Ht} | \psi\rangle$ associated with a finite-dimensional Hamiltonian $H$ and a known input state $|\psi\rangle$. 
In practice, we only access $f$ on a uniform grid $t_i = i\Delta$ for $i=0, 1, \ldots,T-1$ and $\Delta, T > 0$, giving the discrete-time signal
\begin{equation}
    f = (f_0, f_1, \ldots,f_{T-1})^{\mathsf T}, \qquad f_i = f(t_i),
\end{equation}
with $f_0 = 1$ since $\langle \psi | \psi \rangle = 1$. 
The problem is then to reconstruct $f$ given access to $\mathrm{e}^{-\mathrm{i}H\Delta}$, a unitary $U_{\psi}$ preparing $|\psi\rangle$, and its inverse $U_{\psi}^{\dagger}$. 
We refer to this problem as the \emph{signal recovery problem}, which we will formalize in Section~\ref{sec:signal-recovery-problem}. 
For clarity, we refer the reader to Table~\ref{tab:notation} in Section~\ref{sec:notation} for a comprehensive summary of the notation used throughout this paper.

\paragraph{Lifted Representation and Measurements.} 
While one could, in principle, estimate each entry $f_i$ directly using a standard Hadamard test, this approach necessitates the implementation of controlled-$\mathrm{e}^{-\mathrm{i}H t_i}$ operations, as illustrated below:
\begin{center}
    \begin{quantikz}
    \lstick{\ket{\boldsymbol{0}}}  & \gate{\text{Had}}  & \gate{R_{z}(\theta)} &\ctrl{1}   & \gate{\text{Had}} &\meter{} \\
     \lstick{\ket{\boldsymbol{0}}} &\gate{U_\psi} & \qw & \gate{\mathrm{e}^{-\mathrm{i}Ht_i}} & \qw
    \end{quantikz}
\end{center}
A significant drawback of this method is that as $t_i$ increases from $0$ to $(T-1)\Delta$, these controlled time evolutions become increasingly expensive. 
This overhead is unavoidable in general, as the no-fast-forwarding theorem rules out sublinear-in-$t$ simulation of $\mathrm{e}^{-\mathrm{i}H t}$ for generic Hamiltonians~\cite{Berry2006}.
This results in deeper circuits and imposes stringent coherence-time requirements that are often difficult to satisfy on near-term hardware. 
Instead, our approach is to work with quadratic correlations of the form $f_i\overline{f_j}$, which we encode in the lifted matrix $ Z = f f^\dagger$ in the form of its entries $Z_{ij} = f_i \overline{f_j}$ for convenience.
In Section~\ref{sec:lifted-representation}, we design simple quantum circuits for estimating diagonal entries $Z_{ii}=|f_i|^2$ using only $U_\psi$, $U_\psi^\dagger$, and uncontrolled time evolutions $\mathrm{e}^{-\mathrm{i}Ht_i}$, and off-diagonal entries $Z_{ij}$ using controlled time evolutions over only \textit{time differences} $|t_j-t_i|$, that is, controlled-$\mathrm{e}^{-\mathrm{i}H(t_j-t_i)}$, in addition to the aforementioned unitaries. 

\paragraph{$K$-banded Observation.} Because the circuit depth of controlled-$\mathrm{e}^{-\mathrm{i}H(t_j-t_i)}$ unitary grows with $|t_j-t_i|$, it is simple to see that the entries far from the diagonal of $Z$ are more expensive to estimate than nearby ones.
When the underlying signal has additional structure, estimating the full matrix $Z$ is not only impractical, as mentioned before, but also unnecessary.
We therefore restrict attention to estimating the entries of the $K$-band (diagonal + $K$ first off-diagonals) of $Z$ and study which signals are recoverable from such partial measurement data.  
For this, we introduce a signal class $\mathcal{S}_K$ (Definition~\ref{def:signal-class}) which consists of signals with no run of $K$ consecutive zeros.
For signals in this class, we demonstrate the following:
\begin{lemma}[Identifiability; informal Lemma~\ref{lem:identifiability} statement] The $K$-band of $Z$ provides information-theoretically sufficient data to reconstruct a signal $f$ if and only if  $f \in \mathcal{S}_K$.
\end{lemma}
This turns $K$ into a practical design choice: use the easy-to-estimate diagonal entries $Z_{ii} = |f_i|^2$ to spot zero runs in the underlying signal, then estimate only as many off-diagonals, which will determine the width of the band, as those gaps demand. 
We therefore distinguish between these two notions throughout the paper: the intrinsic width $W$ of a signal, defined as the length of its largest run of (near-)zero entries, and the measurement bandwidth $K$, which is a user-chosen parameter determining how many off-diagonals of the lifted matrix are measured.
According to the above lemma, we need to take $K$ so that $K \geq W+1$ to recover the underlying signal uniquely.
Notably, such a structure is not exploited when estimating each $f_i$ directly via the Hadamard test. 
Once the bandwidth $K$ is chosen, we collect the noisy estimates of the $K$-band entries of $Z$ and record them in a $K$-banded observation matrix $\widehat{Z}^{(K)}_{\vphantom{a}}$ (Definition~\ref{def:K-banded-observation-matrix}), setting all unmeasured entries outside the $K$-band to zero as a placeholder for post-processing.

\paragraph{Phaselift and Estimators.} Given $\widehat{Z}^{(K)}_{\vphantom{a}}$, the signal recovery problem essentially reduces to a \emph{structured phaselift problem}, where the task is to compute a signal estimate $\widehat{f}$ that approximates the true signal $f$ as closely as possible. 
We present three estimators to solve this problem, each built around a specific algebraic or spectral relationship between $f$ and its lifted matrix $Z$.

\bigskip
\noindent
\emph{1. Block-by-block algebraic estimator.} 
This estimator rests on the basic identity $Z_{ji} = f_j \overline{f_i}$, which links the phase of $f_i$ to that of any nearby nonzero $f_j$, while the diagonal entries are related to the magnitudes through the identity $Z_{ii} = |f_i|^2$. 
In the noiseless case, that is, when there is no noise in estimating the $K$-band entries of $Z$, this gives a simple sequential reconstruction rule: start with $f_0 = 1$, and for each $i\ge 1$, read off $|f_i|$ from the diagonal as $|f_i| = \sqrt{Z_{ii}}$, and read off $\operatorname{arg}(f_i)$ using the identity
\begin{equation}
    \arg(f_i)=\arg(f_j)-\arg(Z_{ji}),
\end{equation}
for some previously reconstructed neighbor $j < i$ with $|i-j|\le K$ and $f_j\neq 0$. 
By definition, signals in the class $\mathcal{S}_K$ always have such a neighbor, ensuring that the phase information propagates uniquely across the entire signal. 
In practice, the estimator applies the same idea to the noisy, $K$-banded observation matrix $\widehat{Z}^{(K)}_{\vphantom{a}}$ and outputs a noisy signal estimate $\widehat{f}$. 
In Section~\ref{sec:block-by-block-algebraic-estimator}, we present the detailed procedure together with a pseudocode (Algorithm~\ref{algo:algebraic}). 
Furthermore, we show the following: 
\begin{theorem}[Informal Theorem~\ref{thm:exact-rec-algebraic-est} statement] 
The block-by-block algebraic estimator (Algorithm~\ref{algo:algebraic}) recovers the true signal exactly in the absence of measurement noise.
\end{theorem} 
Moreover, for signals in $\mathcal{S}_1$, we show the following stability result for this estimator in the presence of measurement noise.
\begin{theorem}[Informal Theorem~\ref{thm:stability-base-case} statement]
    For signals with no zero entries ($f \in \mathcal{S}_1$), the block-by-block algebraic estimator is stable under measurement noise. Specifically, if the diagonal entries of $Z$ are estimated with error at most $\varepsilon > 0$, and the first off-diagonal entries with error at most $ 0 < \varepsilon' < \varepsilon$, then the reconstruction error is
$\|\widehat f - f\|_2 = O\!\left(\sqrt{T}\varepsilon/\sqrt{\gamma} + T^{3/2} \varepsilon'/\varepsilon\right)$,
where $\gamma = \min_i |f_i|^2$.
\end{theorem}

Using the above stability result, we derive explicit query-complexity bounds for the unitaries $U_\psi$, $\mathrm{e}^{-\mathrm{i}H\Delta}$, and controlled-$\mathrm{e}^{-\mathrm{i}H\Delta}$ that guarantee $\|\widehat f - f\|_2 \le \eta$ with high probability for $\eta > 0$.

\bigskip
\noindent
\emph{2. Block-by-block eigenvector estimator.} 
This estimator exploits the local low-rank structure of the lifted matrix $Z$. 
Specifically, for every $b \in \{0, \ldots, T-K-1\}$, the $(K+1)\times(K+1)$ diagonal block of $Z$, 
\begin{equation}
    Z_{b:b+K,\,b:b+K} = f_{b:b+K}^{\vphantom{\dagger}} f_{b:b+K}^\dagger,
\end{equation}
is exactly rank one, so the local signal segment $f_{b:b+K}$ is simply the (unnormalized) principal eigenvector of that block, up to a global phase. 
The estimator transfers this idea to the noisy setting: each $(K+1)\times(K+1)$ diagonal block of the $K$-banded observation matrix $\widehat{Z}^{(K)}_{\vphantom{a}}$ is treated as a perturbed rank-one matrix, and its principal eigenvector is used to recover an approximate local signal segment. 
The estimator then aligns the arbitrary global phases of these local segments and stitches them together into a single global estimate $\widehat f$. 
In Section~\ref{sec:block-by-block-eigenvector-estimator}, we present the detailed procedure together with a pseudocode (Algorithm~\ref{algo:eigen-estimator-averaging}). 
As with the algebraic estimator, here as well, we show the following:
\begin{theorem}[Informal Theorem~\ref{thm:exact-rec-eigen-est} statement] 
The block-by-block eigenvector estimator (Algorithm~\ref{algo:eigen-estimator-averaging}) recovers the true signal exactly in the absence of measurement noise.
\end{theorem} 

\bigskip
\noindent
\emph{3. Least-squares estimator.}
Unlike the previous estimators, which reconstruct the signal from local diagonal blocks, this estimator takes a global approach by viewing the task as a positive semidefinite (PSD)-constrained matrix completion problem.  
It seeks the PSD matrix $X$ whose entries agree with the observed $K$-band measurements in a least-squares sense:
\begin{equation}
    X^* = \operatorname*{argmin}_{X \succeq 0}
    \sum_{\substack{0 \le i, j < T \\ |i-j| \le K}}
    \left| \widehat{Z}_{ij}^{(K)} - X_{ij} \right|^2.
\end{equation}
The intuition behind this estimator is simple: in the noiseless case and for any $f \in \mathcal{S}_K$, the $K$-band uniquely determines the lifted matrix $Z$ within the PSD cone, so the minimizer must be $X^* = Z$.  
We prove this statement rigorously in Theorem~\ref{thm:exact-recovery-quantum-phaselift-estimator}.  
Now, once $X^*$ is found, the signal is directly obtained via the relation $Z = f f^\dagger$ by taking the principal eigenvector of $X^*$ and fixing its global phase using the constraint $f_0 = 1$. 
In the noisy setting, the estimator carries over this idea naturally: $X^*$ is now the PSD matrix that best fits the noisy $K$-band measurements, and its leading eigenvector provides a robust signal estimate $\widehat f$.  
Section~\ref{sec:quantum-phaselift-estimator} presents the full procedure and a pseudocode (Algorithm~\ref{algo:least-squares}). 
As before with previous estimators, we show the following:
\begin{theorem}[Informal Theorem~\ref{thm:exact-recovery-quantum-phaselift-estimator} statement] 
The least-squares estimator (Algorithm~\ref{algo:least-squares}) recovers the true signal exactly in the absence of measurement noise.
\end{theorem} 
Moreover, as with the algebraic estimator, we show the stability of this estimator in the presence of measurement noise.
\begin{theorem}[Informal Theorem~\ref{thm:stability-proof-least-squares-estimator} statement]
    The least-squares estimator is stable under measurement noise; that is, the following holds:
    \begin{equation}
        \left \| \widehat{f} - \mathrm{e}^{\mathrm{i} \phi}f \right \|_2 \leq \frac{2(2\sqrt{2} + 1)}{\sigma_{\mathrm{min}} } \cdot \frac{\left \|\varepsilon\right \|_2}{\left \| f\right \|_2},
    \end{equation}
    where $\sigma_{\mathrm{min}}$ is the smallest conic singular value of $K$-banded measurement map restricted to a particular cone (we will define this formally in~\eqref{eq:mini-conic-singular-value}), and $\varepsilon$ is the measurement noise vector defined in~\eqref{eq:b-l-observations}.
\end{theorem}
It is important to note here is that, unlike the algebraic estimator, the above stability guarantee does not directly translate into explicit query-complexity bounds, as it depends on the smallest conic singular value, for which we currently do not have a closed-form expression in terms of $T$ and $K$. We leave the derivation of such an expression, or sharp lower bounds for this quantity, to future work.

\paragraph{Quantum Phaselift High-Level Workflow.}
With the above notions in place, we are now in the position to present a pseudocode (Algorithm~\ref{algo:high-level}) that provides a high-level overview of our entire Quantum Phaselift workflow.

\begin{algorithm}[t]
\caption{Quantum Phaselift: High-Level Workflow}
\label{algo:high-level}
\begin{algorithmic}[1]
\State \textbf{Input:} Access to $\mathrm{e}^{-\mathrm{i}H\Delta}$, $U_\psi$, target signal length $T$, noise threshold $\chi$
\State \textbf{Output:} Reconstructed time-series $\widehat f \in \mathbb{C}^T$

\bigskip
\State \textbf{Step 1: Diagonal estimation (support inspection)}
\For{$i = 0$ to $T-1$}
    \State Estimate $Z_{ii} = |f_i|^2$ using the quantum circuit presented in Section~\ref{sec:measurement-diag}
\EndFor
\State Estimate intrinsic width $W$ as the largest run of indices with $|Z_{ii}| < \chi$
\State Choose bandwidth $K \ge W + 1$ (optionally $K = W + c$ for small constant $c$)

\bigskip
\State \textbf{Step 2: Off-Diagonal estimation}
\For{all $(i,j)$ such that $|i-j| \le K$}
    \State Estimate $Z_{ij} = f_i \overline{f_j}$ using the quantum circuit presented in Section~\ref{sec:measurement-off-diag}
\EndFor
\State Assemble the $K$-banded observation matrix $\widehat Z^{(K)}$ as discussed in Section~\ref{sec:k-banded-observation-matrix}

\bigskip
\State \textbf{Step 3: Classical reconstruction}
\State Apply one of the following estimators:
\State \hspace{1em} (i) block-by-block algebraic estimator (Algorithm~\ref{algo:algebraic}),
\State \hspace{1em} (ii) block-by-block eigenvector estimator (Algorithm~\ref{algo:eigen-estimator-averaging}), or
\State \hspace{1em} (iii) least-squares estimator (Algorithm~\ref{algo:least-squares})
\State Obtain reconstructed signal $\widehat f$

\State \Return $\widehat f$
\end{algorithmic}
\end{algorithm}

\paragraph{From Time-Series Estimation to Spectral Estimation.} Now that we have a robust time-series estimate $\widehat{f}$, we show in Section~\ref{sec:time-series-to-spectrum} how to convert it into estimates for the eigenvalues of $H$, specifically, the part of the spectrum actually visible to $|\psi\rangle$.
In other words, assuming the Nyquist–Shannon condition on the sampling step $\Delta$, a minimum nonzero spectral gap $\gamma$ within the subspace supported by $|\psi\rangle$, and a sufficiently accurate time-series $\|\widehat f-f\|_2\le \varepsilon$, we give conditions on $T$ and $\varepsilon$ that guarantee that the Fourier spectrum of $\widehat{f}$ has detectable peaks near the true eigenvalues, and we show that the resulting eigenvalue estimation error scales $\sim 1/(T\Delta)$.

\begin{figure}[t]
    \centering
\includegraphics[width=\linewidth]{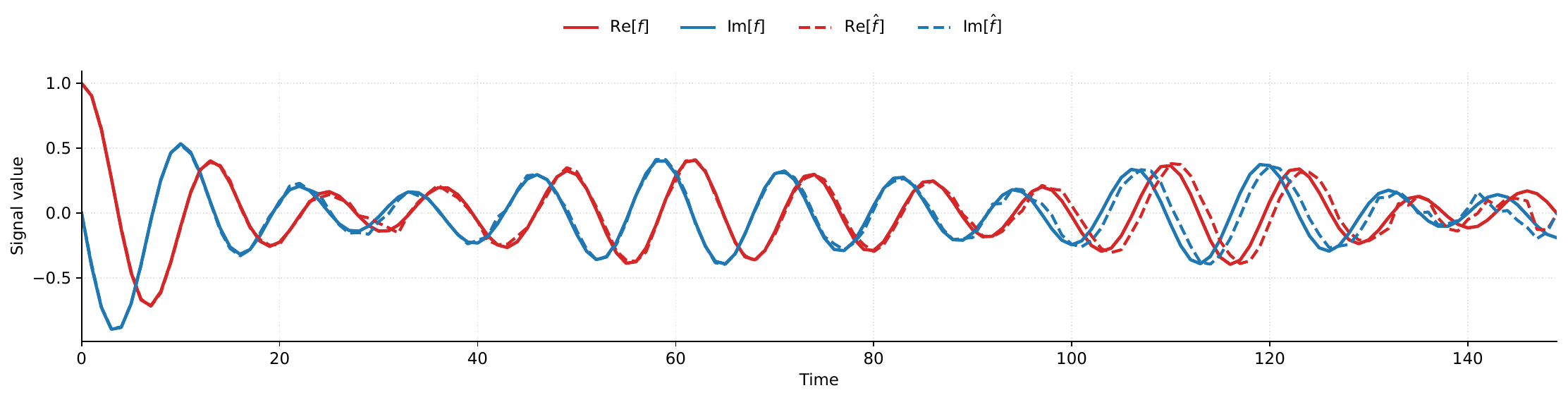}
    \caption{Accurate signal recovery for a representative instance of a $4\times 3$ transverse-field Ising model defined in~\eqref{eq:tf-ising-model}, with randomly chosen coupling parameters and a uniform superposition input state.
    The total number of measurement shots is fixed to approximately one million and distributed uniformly across the measured $K$-band entries of the lifted matrix $Z$.
    The reconstructed signal is obtained using the block-by-block eigenvector estimator (Algorithm~\ref{algo:eigen-estimator-averaging}).}
    \label{fig:constant-shots-summary}
\end{figure}

\paragraph{Numerical Simulations.}
Finally, in Section~\ref{sec:numerical-simulations}, we complement our theoretical analysis with extensive numerical simulations on signals generated by Fermi-Hubbard and transverse-field Ising Hamiltonians. 
We validate the theoretical identifiability conditions established in Lemma~\ref{lem:identifiability}, demonstrating that signal recovery is successful provided the measurement bandwidth $K$ exceeds the length of any zero-runs in the signal.
Provided that the bandwidth is chosen correctly, we demonstrate, across all tested instances, that the proposed estimators reliably recover the underlying signal from noisy banded measurements.
In particular, we observe that the block-by-block algebraic and eigenvector estimators exhibit strong empirical robustness.
We further investigate a practically relevant regime where only a fixed number of measurement shots are available.
Even under a strict restriction of approximately one million total number of shots, we observe accurate signal reconstruction.
Remarkably, as illustrated in Figure~\ref{fig:constant-shots-summary}, this behavior persists for long signals of length $T=150$, where the reconstructed signal, which we get using our eigenvector estimator, closely matches the exact signal over the entire time window.
These results confirm that our estimators enable long-time signal recovery under realistic shot constraints.

\subsection{Structure of the Paper}

The remainder of the paper is organized as follows. Section~\ref{sec:preliminaries} establishes the notation and reviews the classical Phaselift problem. 
Section~\ref{sec:signal-recovery-problem} formalizes the signal recovery problem that we introduced before. 
In Section~\ref{sec:lifted-representation}, we present the quantum circuits required to estimate the entries of the lifted matrix $Z$. 
Section~\ref{sec:K-banded-observation} introduces the signal class $\mathcal{S}_K$ and discusses the adaptive strategy for bandwidth selection, that is, how many off-diagonals of $Z$ (size of the band of $Z$) one needs to estimate for signal recovery. 
Section~\ref{sec:estimators} forms the core of this paper, presenting the block-by-block algebraic, block-by-block eigenvector, and least-squares estimators, along with their theoretical recovery guarantees in the noiseless case and stability analysis for the algebraic and least-squares estimators in the noisy case.
Section~\ref{sec:numerical-simulations} presents numerical simulations that validate the theoretical results and illustrate the practical performance of the proposed estimators.
Finally, in Section~\ref{sec:time-series-to-spectrum}, we show how the error in the estimation of the time-series propagates into the error in the estimation of the spectrum of $H$.

\section{Preliminaries}
\label{sec:preliminaries}

\subsection{Notation}\label{sec:notation}

We adopt the standard linear algebra and quantum computing notation throughout this paper. 
A comprehensive summary of the symbols used, including definitions for signal vectors, matrices, and matrix bands, is provided in Table~\ref{tab:notation}.

\begin{table}[ht]
\centering
\begin{tabular}{ll}
\toprule
Notation & Meaning \\
\midrule
$\mathbb{N}, \mathbb{R}, \mathbb{C}$ & Sets of natural, real, and complex numbers\\
$\operatorname{Re}[z], \operatorname{Im}[z], \overline{z}$ & Real part, imaginary part, and complex conjugate of $z \in \mathbb{C}$\\
$\overline{x}, x^{\mathsf{T}}, x^{\dagger}$ & Conjugate, transpose, and conjugate transpose of a vector $x$\\
$x_i, x_{i:j}$ & $i$-th entry and block $(x_i, \ldots, x_j)^{\mathsf{T}}$ of a vector $x$\\
$\overline{A}, A^{\mathsf{T}}, A^{\dagger}$ & Conjugate, transpose, and conjugate transpose of a matrix $A$\\
$A_{ij}, A_{i:j, k:\ell}$ & $(i,j)$-th entry and submatrix with rows $(i:j)$ and columns $(k:\ell)$ of $A$\\
$K$-band of $A$ & Matrix whose $(i,j)$-entry equals $A_{ij}$ if $|i-j|\le K$ and $0$ otherwise\\
$\operatorname{band}_K(A)$ & Projection of $A$ onto its $K$-band\\
$\boldsymbol{x}$ & Bitstring (e.g., $\boldsymbol{0} \equiv 00\ldots0$)\\
$H$ & Hamiltonian $\in \mathbb{C}^{d \times d}$\\
$|\psi\rangle, U_{\psi}$ & Input state $\in \mathbb{C}^{d}$ and a unitary $\in \mathbb{C}^{d\times d}$ preparing it\\
$f(t)$ & Continuous-time signal evaluated at time $t$ (Definition~\ref{eq:signal-continuous}) \\
$\Delta$ & Sampling interval \\
$T$ & Number of discrete time points sampled \\
$f$ & Discrete signal vector $\in \mathbb{C}^T$ (Definition~\ref{def:discrete-time-signal}) \\
$\widehat f$ & Signal estimate $\in \mathbb{C}^T$\\
$\mathcal{S}_K$ & Class of signals with no run of $K$ consecutive zeros (Definition~\ref{def:signal-class})\\
$Z$ & Lifted representation $\in \mathbb{C}^{T\times T}$, $Z = ff^{\dagger}$\\
$\widehat{Z}^{(K)}_{\vphantom{a}}$ & $K$-banded observation matrix $\in \mathbb{C}^{T\times T}$\\
Had, $R_{z}(\theta)$& Hadamard gate and phase shift gate\\
$\left \| x \right \|_2, \left \| A \right \|_p$ & Euclidean norm of a vector $x$ and Schatten-$p$ norm of a matrix $A$\\
\bottomrule
\end{tabular}
\caption{Summary of notation.}
\label{tab:notation}
\end{table}

\subsection{Classical Phase Retrieval and Phaselift}
\label{sec:phaselift-review}

The classical phase retrieval problem asks to recover an unknown complex vector $x \in \mathbb{C}^n$ of size $n > 0$ given only the magnitudes of $m$ linear measurements of the form
\begin{equation}
    y_k = |\langle a_k, x \rangle|^2, \qquad \text{for } k = 0, 1, \ldots, m,
\end{equation} 
where $\{a_k \in \mathbb{C}^n\}_k$ are known measurement vectors. 
This inverse problem is inherently non-convex, and to overcome this difficulty, the Phaselift approach introduces the concept of "lifting" the problem into a higher-dimensional space~\cite{Candes2013}. 
The key observation is that the above mentioned quadratic measurements on the vector $x$ can be reinterpreted as linear measurements on the lifted rank-one matrix $X = x x^\dagger$:
\begin{equation}
    |\langle a_k, x \rangle|^2 = \operatorname{Tr}\left[a_k^{\vphantom{\dagger}} a_k^\dagger x x^\dagger\right] = \operatorname{Tr}\left [A_k X\right],
\end{equation}
where $A_k \coloneqq a_k^{\vphantom{\dagger}} a_k^\dagger$ are rank-one Hermitian matrices. 
This transformation converts the non-convex quadratic constraints on $x$ into linear constraints on the PSD matrix $X$. 
The problem then becomes finding a matrix $X$ that satisfies these linear constraints, is positive semidefinite, and has rank one.

While rank minimization is generally NP-hard, Phaselift relaxes this condition by minimizing the trace (that is, the nuclear norm proxy for rank) over the PSD cone:
\begin{equation}
    \min_{X \succeq 0} \operatorname{Tr}\left [X\right] \quad \text{subject to} \quad \operatorname{Tr}\left [A_k X\right] = y_k.
\end{equation}
This convex formulation is robust to noise and has been shown to recover the signal with high probability under suitable measurement conditions~\cite{Candes2013}.

This lifting philosophy forms the theoretical foundation of this paper. 
As detailed in Section~\ref{sec:lifted-representation}, we lift the time-series $f$ to the matrix $Z = f f^\dagger$ and design quantum circuits to estimate its entries directly. 
All three estimators that we will present in Section~\ref{sec:estimators} rely on this lifted representation: the block-by-block algebraic and block-by-block eigenvector estimators exploit the rank-one structure of $Z$ locally, while the least-squares estimator adopts the global optimization perspective of the above mentioned classical Phaselift algorithm of~\cite{Candes2013}, solving for the matrix $X$ within the PSD cone that is most consistent with the measurements.

\section{Signal Recovery Problem}
\label{sec:signal-recovery-problem}

We begin by formalizing the signal recovery problem that we introduced before. 
We consider a Hamiltonian $H \in \mathbb{C}^{d \times d}$ acting on a finite-dimensional Hilbert space with dimension $d \in \mathbb{N}$ and a known input quantum state $|\psi\rangle \in \mathbb{C}^d$.  
Our goal is to reconstruct the complex-valued time series
\begin{equation}
    f(t) = \langle \psi | \mathrm{e}^{-\mathrm{i} H t} | \psi\rangle,
    \label{eq:signal-continuous}
\end{equation}
which represents the overlap of the time-evolved state $\mathrm{e}^{-\mathrm{i}Ht}|\psi\rangle$ with the reference state $|\psi\rangle$. 
This function is a continuous function in $t$ and encodes information about the spectrum of $H$ weighted by the overlaps of $|\psi\rangle$ with its eigenstates.

In practice, $f(t)$ can only be sampled at discrete time points since it is a continuous function. 
Therefore, we define a uniform sampling grid
\begin{equation}\label{eq:uniform-sampling-grid}
    t_i = i \Delta, \qquad \text{for all }i = 0, 1, \ldots, T-1,
\end{equation}
where $\Delta > 0$ is the sampling interval (that is, the difference between two consecutive sampled time points) and $T \in \mathbb{N}$ is the total number of discrete time points. 
We denote the corresponding discrete-time version of $f$ by the complex-valued vector
\begin{equation}
    f = (f_0, f_1, \dots, f_{T-1})^{\mathsf{T}} \in \mathbb{C}^{T} \label{def:discrete-time-signal},
\end{equation}
where we define its entries as
\begin{equation}\label{def:discrete-time-signal-elements}
    f_i = f(t_i), \qquad \text{for all } i = 0, 1, \ldots, T-1.
\end{equation}
Note that we use the same notation $``\!f"$ for the continuous signal, as well as for its discretized version, to avoid introducing new notation for representing the signal. 
For most of the paper, we focus on the discretized signal; however, the context will make it clear which one is being discussed at any given point. 

Furthermore, since the system evolution satisfies 
\begin{equation}
    \langle \psi | \mathrm{e}^{-\mathrm{i}Ht_0} | \psi\rangle = \langle \psi | \psi\rangle = 1
\end{equation}
at $t_0=0$, we inherently have $f_0 = 1$. 
This constraint will be useful in removing the global phase ambiguity inherent to the signal recovery problem, and all of our estimators (Section~\ref{sec:estimators}) will exploit this constraint.

With these definitions in place, we formalize the problem that we focus on in this paper:
\begin{definition}[Signal Recovery Problem]\label{def:signal-recovery}
Let $H \in \mathbb{C}^{d\times d}$ be a finite-dimensional Hamiltonian with $d > 0$ and let $|\psi\rangle \in \mathbb{C}^d$ be a known quantum state. The signal recovery problem is to recover/reconstruct the discrete-time signal $f$ as defined in~\eqref{def:discrete-time-signal} and~\eqref{def:discrete-time-signal-elements}, given oracle access to the time-evolution unitary $\mathrm{e}^{-\mathrm{i}H\Delta}$ and to the state-preparation unitary $U_\psi$, along with its inverse $U_\psi^\dagger$.
\end{definition}

\section{Lifted Representation and Measurements}\label{sec:lifted-representation}

Motivated by the hardware constraints, recall from our discussion in previous sections that our focus in this paper is on estimating the quadratic correlations between signal values, that is, $f_i \overline{f_j}$, instead of directly estimating the signal values $f_i$. 
As shown in this section, estimating these correlations requires only controlled time evolutions of the form $\mathrm{e}^{-\mathrm{i}H(t_j - t_i)}$. Crucially, unlike the controlled-$\mathrm{e}^{-\mathrm{i}H t_i}$ operations used in the standard Hadamard test, these evolutions depend only on the time difference $|t_j - t_i|$. This significantly simplifies the implementation when the signal has additional structure, which we make precise later in Section~\ref{sec:K-banded-observation}.

\subsection{Lifted Representation}

To work directly with quadratic correlations, we introduce a lifted representation of the signal. 
Specifically, we collect all correlations $f_i \overline{f_j}$ into the matrix
\begin{equation}\label{eq:def-Z}
    Z = f f^{\dagger}, \qquad Z_{ij} = f_i \overline{f_j}.
\end{equation}
This lifting maps the original signal vector $f \in \mathbb{C}^T$ to a matrix $Z \in \mathbb{C}^{T \times T}$ whose entries correspond exactly to the quantities we aim to estimate. 
By construction, $Z$ is a rank-one positive semidefinite matrix. 

Observe that the diagonal entries of $Z$ encode the signal magnitudes $\{|f_i|^2\}_i$, while its off-diagonal entries encode all pairwise quadratic correlations $\{f_i \overline{f_j}\}_{i,j }$, with $i \neq j$. 
Taken together, these entries determine the signal $f$ up to a global phase. 
Indeed, for any arbitrary phase $\phi \in \mathbb{R}$, we observe that 
\begin{equation}
    (\mathrm{e}^{\mathrm{i}\phi} f)(\mathrm{e}^{\mathrm{i}\phi} f)^{\dagger} = \mathrm{e}^{\mathrm{i}\phi} f  \mathrm{e}^{-\mathrm{i}\phi} f^{\dagger} = f f^{\dagger} = Z.
\end{equation}
Thus every vector in the equivalence class $f \sim \mathrm{e}^{\mathrm{i}\phi} f$ produces the same lifted matrix.
In our setting, this ambiguity is not present because the constraint $f_0 = 1$ fixes the global phase.

With this lifted formulation in place, we now describe how to estimate the entries of $Z$ using quantum circuits.

\subsection{Measuring the Diagonal Elements}
\label{sec:measurement-diag}

The diagonal elements of $Z$ are real and correspond to
\begin{equation}
    Z_{ii} = \left | f_i \right|^2 = \left | f(t_i) \right|^2 = |\langle\psi|\mathrm{e}^{-\mathrm{i}Ht_i}|\psi\rangle|^2 = \operatorname{Tr}\!\left[\rho_\psi \mathrm{e}^{\mathrm{i}Ht_i}\rho_\psi \mathrm{e}^{-\mathrm{i}Ht_i}\right]=
    \operatorname{Tr}\!\left[\rho_\psi\rho_\psi(t_i)\right],
\end{equation}
where the notation $\rho_{\psi} \equiv |\psi\rangle \langle \psi|$ denotes the density matrix of the pure quantum state $|\psi\rangle$, and
\begin{equation}
    \rho_\psi(t) \coloneqq \mathrm{e}^{\mathrm{i}Ht} \rho_\psi \mathrm{e}^{-\mathrm{i}Ht}
\end{equation}
is its time-evolved counterpart. It is simple to see that the quantity $\operatorname{Tr}[\rho_\psi\rho_\psi(t_i)]$ can be estimated using the following quantum circuit:
\begin{center}
    \begin{quantikz}
    \lstick{$\ket{\boldsymbol{0}}$} &\gate{U_\psi} &\gate{\mathrm{e}^{\mathrm{i}Ht_i}} &\gate{U_\psi^\dagger} & \meter{\boldsymbol{0}} 
    \end{quantikz}
\end{center}
To obtain an unbiased estimator $\widehat Z_{ii}$ for $Z_{ii}$, the circuit above is executed $N_{\mathrm{diag}}$ times, with measurements performed in the computational basis. If we let $N_{\boldsymbol{0}}$ denote the number of outcomes corresponding to the all-zeros state $|\boldsymbol{0}\rangle$, the unbiased estimator is given by
\begin{equation}\label{eq:unbiased-est-diagonal-element}
\widehat Z_{ii} = \frac{N_{\boldsymbol{0}}}{N_{\mathrm{diag}}}.
\end{equation}

It is worth reiterating that this circuit does not require any controlled operations. This scheme is particularly suitable for hardware architectures such as small trapped-ion systems, where the number of qubits is limited, but high gate fidelities and long coherence times permit the execution of deeper circuits. Conversely, for architectures where circuit depth is a primary bottleneck, it may be more practical to implement a SWAP test between a copy of $\rho_\psi$ and one of $\rho_\psi(t_i)$ to estimate $\operatorname{Tr}[\rho_\psi\rho_\psi(t_i)]$. We provide a detailed discussion of this SWAP-test-based approach for estimating the diagonal elements of $Z$ in Appendix~\ref{app:estimating-diagonal-ele-swap}.

\subsection{Measuring the Off-Diagonal Elements}
\label{sec:measurement-off-diag}

Turning now to the off-diagonal elements of $Z$, we note that these are generally complex. They are given by
\begin{equation}
    Z_{ij} = f_i\overline{f_j} = f(t_i)\overline{f(t_j)} = \langle\psi|\mathrm{e}^{-\mathrm{i}Ht_i}|\psi\rangle\langle\psi|\mathrm{e}^{\mathrm{i}Ht_j}|\psi\rangle = \operatorname{Tr}\left[\mathrm{e}^{\mathrm{i}Ht_i}\rho_\psi \mathrm{e}^{-\mathrm{i}Ht_i}\rho_\psi \mathrm{e}^{\mathrm{i}H(t_j-t_i)}\right].
\end{equation}
Due to their complex nature, the real and imaginary parts of these elements must be estimated separately.

To this end, we propose the circuit depicted below, where a phase shift gate $R_z(\theta)$ is applied to the control qubit immediately following the first Hadamard gate:
\begin{center}
    \begin{quantikz}
    \lstick{\ket{0}}  & \gate{\text{Had}}  & \gate{R_{z}(\theta)} &\ctrl{1}   & \gate{\text{Had}} & \qw  &\meter{x} \\
     \lstick{\ket{\boldsymbol{0}}} &\gate{U_\psi} & \qw & \gate{\mathrm{e}^{\mathrm{i}H(t_j-t_i)}} & \gate{\mathrm{e}^{\mathrm{i}Ht_i}} & \gate{U_\psi^{\dagger}} & \meter{\boldsymbol{y}}
    \end{quantikz}
\end{center}
Here $x$ is a single bit, and $\boldsymbol{y}$ is a bitstring of length $\lceil \log d \rceil$. Prior to the final measurements, the state of the system is
\begin{equation}
    \lvert\chi\rangle=\frac{1}{2}\left(\ket{0}\otimes U_\phi^\dagger(I+\mathrm{e}^{\mathrm{i}\theta}U)U_\psi\ket{\boldsymbol{0}}
    +\ket{1}\otimes U_\phi^\dagger(I-\mathrm{e}^{\mathrm{i}\theta} U)U_\psi\ket{\boldsymbol{0}}\right),
\end{equation}
where we have adopted the shorthand notations $U \equiv \mathrm{e}^{\mathrm{i}H(t_j-t_i)}$, $U_{\phi} \equiv \mathrm{e}^{-\mathrm{i}Ht_i} U_{\psi}$, and $|\phi\rangle \equiv U_{\phi} |\boldsymbol{0}\rangle$ to streamline the derivation. 

The amplitude associated with observing the outcome $(x,\boldsymbol{y})=(0,\boldsymbol{0})$ reads
\begin{equation}
    (\bra{0} \otimes \bra{\boldsymbol{0}})|\chi\rangle=\frac{1}{2}\bra{\boldsymbol{0}} U_\phi^\dagger(I+\mathrm{e}^{\mathrm{i}\theta} U)U_\psi\ket{\boldsymbol{0}}=
    \frac{1}{2}\langle \phi|I+\mathrm{e}^{\mathrm{i}\theta}U|\psi\rangle,
\end{equation}
with probability
\begin{align}
    P(0,\boldsymbol{0}) &=\frac{1}{4}\left[ |\langle\phi|\psi\rangle|^2+|\bra{\phi}U|\psi\rangle|^2
+\mathrm{e}^{-\mathrm{i}\theta}\langle\phi|\psi\rangle\langle \psi| U^\dagger|\phi\rangle +
 \mathrm{e}^{\mathrm{i}\theta}\langle\psi|\phi\rangle\langle \phi | U | \psi \rangle\right]  \\
 &=\frac{1}{4}\left[ |\langle\phi|\psi\rangle|^2+|\langle \phi | U | \psi \rangle|^2
+2 \operatorname{Re}[\mathrm{e}^{\mathrm{i}\theta}\langle\psi|\phi\rangle\langle \phi | U | \psi \rangle]\right].
\end{align}
Similarly, it is easy to see that 
\begin{equation}
    P(1,\boldsymbol{0})=\frac{1}{4}\left[ |\langle\phi|\psi\rangle|^2+|\langle \phi | U | \psi \rangle|^2
-2 \operatorname{Re}[\mathrm{e}^{\mathrm{i}\theta}\langle\psi|\phi\rangle\langle \phi | U | \psi \rangle]\right].
\end{equation}
Combining these yields the difference
\begin{equation}
    P(0,\boldsymbol{0})-P(1,\boldsymbol{0})=\operatorname{Re}[\mathrm{e}^{\mathrm{i}\theta}\langle\psi|\phi\rangle\langle \phi | U | \psi \rangle].
\end{equation}
Since $\langle\psi|\phi\rangle\langle \phi | U | \psi \rangle = \overline{Z_{ij}}$, we have
\begin{equation}
    P(0,\boldsymbol{0})-P(1,\boldsymbol{0})=\operatorname{Re}[\mathrm{e}^{\mathrm{i}\theta}Z_{ij}].
\end{equation}
By choosing $\theta$ appropriately, we can isolate the real and imaginary parts of $Z_{ij}$:
\begin{align}
    \theta = 0
\quad &\Rightarrow\quad
P(0,\boldsymbol{0}) - P(1,\boldsymbol{0}) = \operatorname{Re}[Z_{ij}]\\
\theta = \frac{\pi}{2}
\quad & \Rightarrow\quad
P(0,\boldsymbol{0}) - P(1,\boldsymbol{0}) = \operatorname{Im}[Z_{ij}].
\end{align}

Running the above circuit $N_{\mathrm{off}}$ times and letting $N_{0\boldsymbol{0}}$ and $N_{1\boldsymbol{0}}$ denote the counts of the outcomes $(0,\boldsymbol{0})$ and $(1,\boldsymbol{0})$, respectively, we obtain the unbiased estimators
\begin{align}
\widehat{\operatorname{Re}[Z_{ij}]}
    & = \frac{N_{0\boldsymbol{0}} - N_{1\boldsymbol{0}}}{N_{\mathrm{off}}}
    \quad (\theta=0) \\
\widehat{\operatorname{Im}[Z_{ij}]}
   &  = \frac{N_{0\boldsymbol{0}} - N_{1\boldsymbol{0}}}{N_{\mathrm{off}}}
    \quad \left(\theta=\frac{\pi}{2}\right).
\end{align}
The full complex unbiased estimator for $Z_{ij}$ is then
\begin{equation}
\widehat Z_{ij}
    = \widehat{\operatorname{Re}[Z_{ij}]}
      + \mathrm{i}\,\widehat{\operatorname{Im}[Z_{ij}]}.
\end{equation}

Finally, consistent with our discussion of diagonal entries, deep circuits and coherence-time limits may render this specific approach impractical on certain hardware platforms. In such cases, a SWAP-test--based method offers a viable alternative for estimating $Z_{ij}$. Further details on this alternative are provided in Appendix~\ref{app:estimating-off-diagonal-ele-swap}.

\section{\texorpdfstring{$K$}{K}-Banded Observation}\label{sec:K-banded-observation}

\subsection{Signal Class \texorpdfstring{$\mathcal{S}_K$}{Sₖ}}

For a bandwidth parameter $K \in \{1, 2, \ldots, T-1\}$, consider the following class of signals:
\begin{equation}\label{def:signal-class}
    \mathcal{S}_K \coloneqq \left\{ f \in \mathbb{C}^T : f \text{ does not contain a run of } K \text{ consecutive zeros} \right\}.
\end{equation}
That is, $\mathcal{S}_K$ consists precisely of those signals whose intrinsic width $W$, which we defined before as the length of the largest run of consecutive zero entries, satisfies $W \le K-1$.

Intuitively, the relevance of this class becomes clear when one views recovery through the lifted matrix $Z$. 
Each nonzero off-diagonal entry $Z_{ij} = f_i \overline{f_j}$ acts as a "bridge" that relates the phases of $f_i$ and $f_j$. 
If a signal contains a long stretch of zeros, then these bridges disappear across that region, preventing phase information from propagating from one side of the signal to the other using only local correlations.
In contrast, if the signal never vanishes for $K$ consecutive indices, then for every index $i$ there exists a nearby index $j$ with $|i-j|\le K$ and $f_j \neq 0$. 
This guarantees that phase information can be passed sequentially across the entire signal using only entries from its $K$-band.

\paragraph{Concrete example.}
Consider a signal of length $T=10$ and let $W=3$. Define
\begin{equation}
f = (1,\, 0.7,\, 0,\, 0,\, 0,\, 0.5,\, 0.3,\, 0,\, 0.4,\, 0.2).
\end{equation}
This signal contains a run of three consecutive zeros at indices $\{2,3,4\}$, and therefore does not belong to the signal class $\mathcal{S}_3$. 
Notice that the longest string of zeros is length three, so it does belong to $\mathcal{S}_4$. 

For any index $i$ in the zero block $\{2,3,4\}$, we have $f_i=0$, and hence
\begin{equation}
Z_{ij} = Z_{ji} = 0, \qquad \text{for all } j.
\end{equation}
In particular, all entries of $Z$ in the $3$-band that would connect indices $\{0,1\}$ to indices $\{5,6,7, 8, 9\}$ vanish. 
Consequently, the $3$-band of $Z$ decomposes into two disconnected regions: one supported on indices $\{0,1\}$ and another on indices $\{5,6,7,8,9\}$. 
Any relative phase between these two regions is therefore unobservable from the $3$-band alone.

\bigskip
We formalize this intuition with the following lemma, which establishes that $\mathcal{S}_K$ is exactly the class of signals that can be uniquely recovered from the $K$-band of $Z$.

\begin{lemma}[Identifiability]\label{lem:identifiability}
    A signal $f$ is uniquely determined by the $K$-band of its lifted matrix $Z = ff^\dagger$ if and only if $f \in \mathcal{S}_K$.
\end{lemma}
\begin{proof}
    ($\Leftarrow$) If $f \in \mathcal{S}_K$, then it can be uniquely reconstructed from the $K$-band of $Z$. We provide a constructive proof of this in Theorem~\ref{thm:exact-rec-algebraic-est}, Theorem~\ref{thm:exact-rec-eigen-est}, and Theorem~\ref{thm:exact-recovery-quantum-phaselift-estimator}, where we show that our estimators recover $f \in \mathcal{S}_K$ exactly using only $K$-band entries of $Z$.

    ($\Rightarrow$) We prove the contrapositive: if $f \notin \mathcal{S}_K$, then $f$ cannot be uniquely determined from the $K$-band of $Z$ alone. Suppose $f$ contains a run of $K$ consecutive zeros starting at index $s \ge 1$. That is, $f_i = 0$ for all $i \in \{s, \ldots, s+K-1\}$. We construct a different signal $g \neq f$ that produces the exact same $K$-band on its lifted matrix. Let $\varphi \in (0, 2\pi)$ be an arbitrary phase and define $g$ as:
        \begin{equation}
        g_i = \begin{cases} 
            f_i & \text{for } i < s+K, \\
            \mathrm{e}^{\mathrm{i}\varphi} f_i & \text{for } i \ge s+K.
        \end{cases}
    \end{equation}
    
Let $G \coloneqq gg^\dagger$. For any pair of indices $(i, j)$ such that $|i-j| \leq K$, either $(i,j)< s+K$, or $(i,j)\ge s+K$, or the pair straddles the zero run; in all cases, we have
    \begin{align}
        G_{ij}=
\begin{cases}
f_i\overline{f_j}, & i,j < s+K,\\
\mathrm{e}^{\mathrm{i}\varphi}f_i\,\overline{\mathrm{e}^{\mathrm{i}\varphi}f_j}=f_i\overline{f_j}, & i,j\ge s+K,\\
0, & \text{otherwise}.
\end{cases}
    \end{align}

Thus, $G$ and $Z$ are identical on the $K$-band, but $g \neq f$. This implies that any signal with a run of $K$ zeros cannot be uniquely determined just from the $K$-band of its lifted matrix. This concludes the proof.
\end{proof}

The lemma above establishes the theoretical limits of recoverability. This naturally raises two practical questions:
\begin{enumerate}
    \item Do signals of physical interest, specifically, time-correlation functions $f(t) = \langle \psi | \mathrm{e}^{-\mathrm{i} H t} | \psi\rangle$, actually contain consecutive zeros?
    \item If so, how can we efficiently determine the length of these zero runs for an unknown signal to choose the correct bandwidth $K$?
\end{enumerate}
We address these questions in the following subsections.

\subsection{Example of a Signal with Consecutive Zeros}

First, we confirm that the signals with consecutive zeros are not pathological edge cases but can arise even for simple quantum Hamiltonians. 
The following lemma illustrates this for an \emph{integer Hamiltonian} with a maximally mixed initial state, where the resulting discrete signal has long runs of zeros.

\begin{lemma}[Integer Hamiltonian Example]\label{lem:integer-H} 
    Let a system of $n$ qubits with dimension $N=2^n$ evolve under the Integer Hamiltonian:
    \begin{equation}
        H = \sum_{j=0}^{n-1} 2^j \left( \frac{I - Z_j}{2} \right).
    \end{equation}
    If initialized in the maximally mixed state $\rho = \frac{1}{N}I$ with dimension $N = W+1$, the discrete signal $f(t) = \operatorname{Tr}[\rho e^{-i H \frac{2\pi}{N} t}]$ contains a run of $W$ consecutive zeros at integers $t = 1, 2, \ldots, W$ and exhibits a revival at $t=W+1$, and this pattern repeats periodically. 
\end{lemma}

\begin{proof}
    The proof is deferred to Appendix~\ref{app:integer-H-proof}.
\end{proof}

This example demonstrates that even simple quantum dynamics can generate signals where $K$-banded observations fail if $K$ does not 
satisfy $K\geq W+1$, where $W$ is the length of the largest run of consecutive zeros in the signal $f$.

\subsection{Adaptive Bandwidth Selection}
\label{sec:adaptive-bandwidth-selection}

The relationship between the width of zero runs $W$ and recoverability motivates a practical, adaptive strategy for selecting $K$. Rather than fixing the bandwidth \emph{a priori}, we can determine it experimentally by inspecting the magnitude profile of the signal.
Recall from Section~\ref{sec:measurement-diag} that the diagonal entries of the lifted matrix $Z$ satisfy $Z_{ii} = |f_i|^2$ and can be estimated using simple circuits without requiring any controlled time evolutions. 
Recall that we refer to these individual estimates as $\widehat{Z}_{ii}$ (see Eq.~\eqref{eq:unbiased-est-diagonal-element}). 
This allows for a resource-efficient two-stage protocol:
\begin{enumerate}
    \item \textbf{Diagonal Scan:} First, estimate only the diagonal entries to obtain the magnitude profile $\{\widehat{Z}_{ii}\}_{i=0}^{T-1}$. These measurements are less costly and provide a map of the support of the underlying signal at a low cost.
    \item \textbf{Determine Bandwidth:} Inspect the matrix diagonal to estimate $W$, i.e., the largest run of consecutive zeros. 
    We then set the $K \geq W + 1$.
\end{enumerate}
This strategy ensures that the subsequent off-diagonal measurements will be sufficient to bridge all gaps in the signal, guaranteeing unique recoverability while minimizing the experimental cost of implementing the controlled time evolutions with large time differences.

In a \textbf{realistic scenario}, having access to finite circuit shots induces shot noise in the estimation of the matrix elements, making an exact zero indistinguishable from any value lower than its shot noise. Therefore, one needs to adapt the ideal scenario above to identify sequences of indices where the estimates $\widehat{Z}_{ii}$ drop below a chosen noise threshold $\chi > 0$. 
In addition, instead of immediately fixing $K=W+1$, we will show later that selecting $K \geq W + c$, where $c$ is a small constant $c\geq 1$ selected heuristically, can help improve the shot counts of the signal recovery at a small additional circuit complexity cost.
We will see this behaviour explicitly when we discuss our numerical simulations in Section~\ref{sec:numerical-simulations}.

\subsection{\texorpdfstring{$K$}{K}-Banded Observation Matrix}
\label{sec:k-banded-observation-matrix}

Once the bandwidth $K$ is selected, we proceed to measure the off-diagonal entries using quantum circuits shown in Section~\ref{sec:measurement-off-diag}. 
We aggregate the experimentally estimated $K$-band of $Z$ into a single Hermitian matrix $\widehat{Z}^{(K)}_{\vphantom{a}}$. Only entries $Z_{ij}$ with $|i-j|\le K$ are measured. 
As we do not observe the entries of $Z$ with $|i-j| > K$, we initialize our estimates of these $Z_{ij}$ to zero; these estimates will be revised in post-processing. 
Thus, our initial estimated entries are recorded as
\begin{equation}\label{def:K-banded-observation-matrix}
    \widehat{Z}_{ij}^{\,(K)} \coloneqq
    \begin{cases}
        Z_{ij} + \varepsilon_{ij}, & |i-j| \le K,\\
        0, & |i-j| > K,
    \end{cases}
\end{equation}
where $\varepsilon_{ij}$ denotes statistical or experimental noise.

The construction of $\widehat{Z}^{(K)}_{\vphantom{a}}$ highlights a fundamental design trade-off between experimental cost and reconstruction stability. 
As mentioned before, from an information-theoretic standpoint, for any signal $f \in \mathcal{S}_K$, the noiseless $K$-band contains sufficient information to uniquely identify $f$. 
Specifically, if the longest run of zeros in $f$ has length $W$, then a bandwidth of $K_{\min} = W + 1$ is strictly necessary and sufficient to ensure connectivity of the phase information. 
However, limiting measurements to this theoretical minimum ($K = K_{\min}$) empirically yields a reconstruction that is highly sensitive to noise. 
In such a "barely connected" regime, there is typically only a single path of non-zero correlations linking any given $f_i$ to the reference $f_0$; consequently, if any measurement along this path is inaccurate, the error propagates to all subsequent indices.

By selecting a bandwidth $K > K_{\min}$, we introduce redundancy into the observation matrix, and this redundancy is crucial for noise mitigation. 
This means we measure correlations $f_i \overline{f_j}$ not just for immediate neighbors, but also for more distant pairs. 
This extra connectivity creates multiple independent pathways for phase propagation that stabilize the final global estimate against measurement errors. 
Therefore, while minimizing $K$ reduces the circuit depth (since controlled time-evolutions scale with $|t_i - t_j|$), increasing $K$ improves the reliability of the estimate. 
We exploit this trade-off between efficiency and statistical robustness in the construction of our estimators in Section~\ref{sec:estimators} and demonstrate it numerically in Section~\ref{sec:numerical-simulations}.

\section{Phaselift and Estimators}
\label{sec:estimators}

With the measurement model in place, the signal recovery problem (Definition~\ref{def:signal-recovery}) reduces to a structured phaselift problem: given a $K$-banded observation matrix $\widehat{Z}^{(K)}_{\vphantom{a}}$, approximating $Z = f f^{\dagger}$ on its $K$-band, reconstruct $f \in \mathcal{S}_K$.

We are now in the position to present the three estimators, that is, the block-by-block algebraic estimator, the block-by-block eigenvector estimator, and the least-squares estimator, that we introduced in the summary (Section~\ref{sec:summary}) in more detail.

\subsection{Block-By-Block Algebraic Estimator}
\label{sec:block-by-block-algebraic-estimator}

 The block-by-block algebraic estimator reconstructs the signal $f$ sequentially from index $0$ to $T-1$, using only simple algebraic identities that relate the entries on the $K$-band of $Z$ to the magnitudes and phases of $f$.

The estimator begins by fixing the global phase. Since we know the first element of the true signal $f$ exactly, that is, $f_0 = 1$, the estimator sets $\widehat{f}_0 = 1$. This fixes the global phase. 

Then for each subsequent index $i$, the estimator estimates the magnitude of $f_i$ directly from the diagonal of $\widehat{Z}^{(K)}_{\vphantom{a}}$ (see Figure~\ref{fig:algebraic-estimator} for an illustration with $K=2$):
\begin{equation}
    |\widehat{f}_i| = \sqrt{|\widehat Z_{ii}^{(K)}|},
\end{equation}
where taking the absolute value of $\widehat Z_{ii}$ ensures a non-negative estimate of the magnitude. 
This step accounts for the possibility that when the true value $Z_{ii}$ is close to zero, statistical noise may cause the observed value $\widehat Z_{ii}^{(K)}$ to take on small negative values.

Furthermore, the estimator uses the following algebraic identity to estimate the phase of $f_i$:
\begin{equation}
    Z_{ji} = f_j \overline{f_i} \quad\Longrightarrow\quad
\arg(f_i) = \arg(f_j) - \arg(Z_{ji})~~(\operatorname{mod} 2\pi),
\end{equation}
which relates the phase of $f_i$ to that of any neighboring nonzero entry $f_j$.
This identity motivates a collection of initial candidate phase estimates for $f_i$, one for each previously reconstructed index $j$ in the $K$-size window $\{i-K, \ldots, i-1\}$:
\begin{equation}
    \widehat\theta_i^{\,\,(j)}
    = \arg(\widehat f_j) - \arg(\widehat Z_{ji}^{\,(K)})~~(\operatorname{mod} 2\pi).
    \label{eq:candidate-phase-estimate}
\end{equation}

To combine multiple candidate estimates into a single robust estimate, the estimator performs a weighted circular mean of these candidate estimates. 
Due to the fact that phases are periodic, we average their corresponding unit phasors rather than raw angles. Let $w_j \geq 0$ be weights with $\sum_j w_j = 1$. 
For our case, we use weights proportional to the magnitudes of the observed entries, that is, 
\begin{equation}
    w_j \propto |\widehat{Z}_{ji}|, \qquad \sum_{j=\max(0,i-K)}^{i-1} w_j = 1.
\end{equation}
The final robust estimate for $\arg(f_i)$ is then
\begin{equation}
    \arg(\widehat{f}_i) =\arg\Big(\sum_{j=\max(0,i-K)}^{i-1}
w_j \mathrm{e}^{\mathrm{i} \widehat\theta_i^{\,\,(j)}}\Big)\label{eq:weighted-circular-average},
\end{equation}
which is simply the argument of the weighted phasor sum.

The reasoning behind weighing by $|\widehat{Z}_{ji}|$ is that it captures the reliability of each measurement into the phase estimate. 
Since the magnitude of $\widehat{Z}_{ji}$ reflects the strength of the correlation between $f_j$ and $f_i$, larger values correspond to observations with higher signal-to-noise ratios and thus greater statistical confidence. 
The weighted phasor sum in~\eqref{eq:weighted-circular-average} can therefore be interpreted as computing a weighted mean direction on the complex unit circle, where each term contributes proportionally to its reliability.

\begin{figure}
    \centering
    \includegraphics[width=\linewidth]{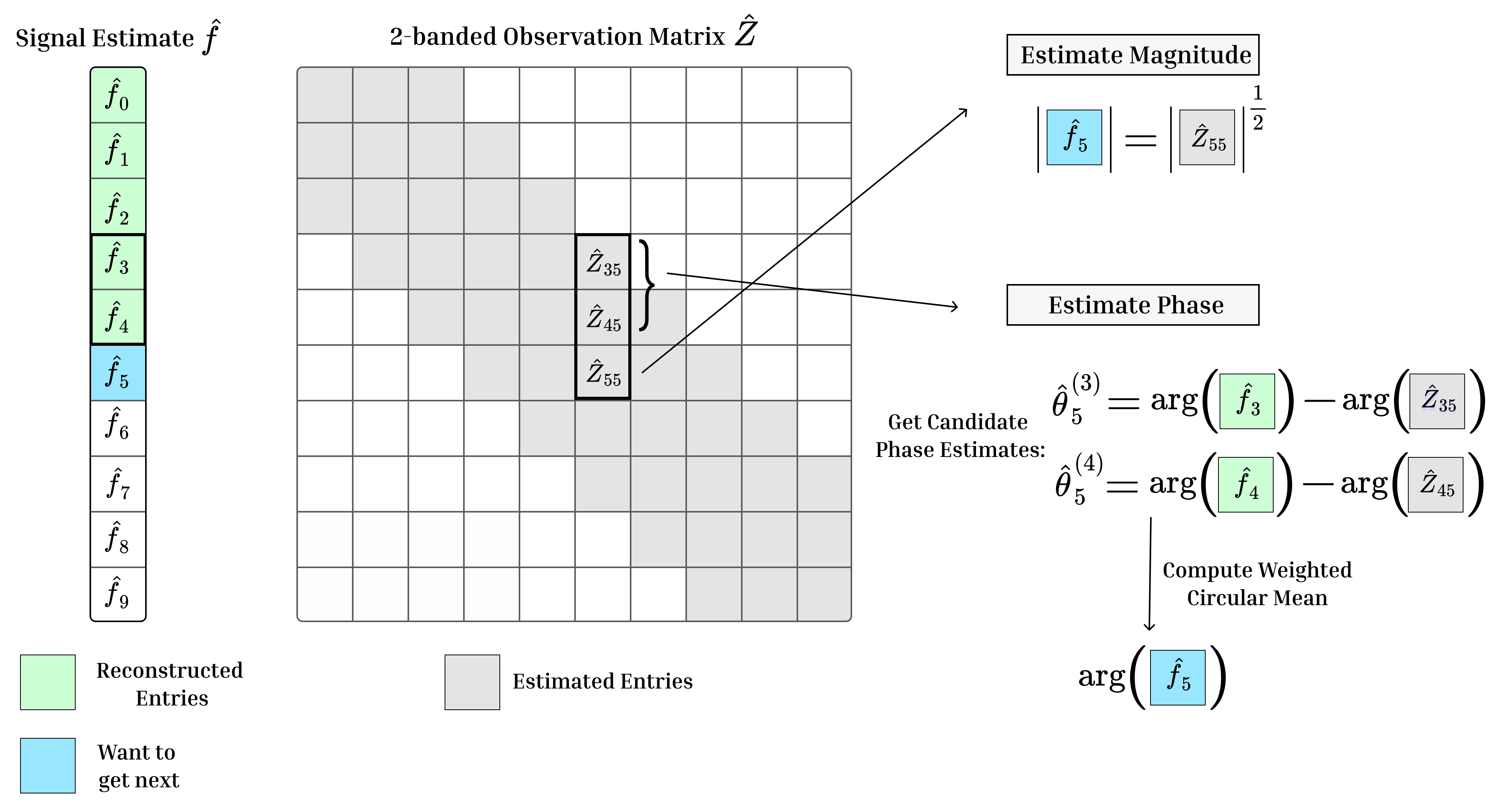}
    \caption{Schematic of the Block-by-Block Algebraic Estimator (shown for $K=2$).}
    \label{fig:algebraic-estimator}
\end{figure}

For the first few indices $i \leq K$, their respective sliding windows contain fewer than $K$ elements, but the same weighted averaging applies over whatever indices are available in these windows. 
To this end, the pseudocode of this estimator is given by Algorithm~\ref{algo:algebraic}.

\begin{algorithm}[H]
\caption{Block-by-Block Algebraic Estimator}
\label{algo:algebraic} 
\begin{algorithmic}[1]
\State \textbf{Input:} $K$-banded observation matrix $\widehat{Z}^{(K)}_{\vphantom{a}} \in \mathbb{C}^{T \times T}$
\State \textbf{Output:} Signal estimate $\widehat{f}\in \mathbb{C}^T$

\State Set phase reference: $\widehat{f}_0 \gets 1$
\For{$i = 1$ to $T-1$}
    \State Estimate magnitude: 
           $|\widehat{f}_i| \gets \sqrt{|\widehat{Z}_{ii}^{\,(K)}|}$
    \State Estimate phase from previously reconstructed entries:
        \For{$j = \max(0, i-K)$ to $i-1$}
            \State Phase estimate: 
                   $\widehat{\theta}_i^{\,\,(j)} \gets 
                   \arg(\widehat{f}_j) - \arg(\widehat{Z}_{ji}^{\,(K)})~~(\operatorname{mod} 2\pi)$
            \State Weight: $w_j \gets |\widehat{Z}_{ji}|$
        \EndFor
        \State $\widehat{\theta}_i \gets$ weighted circular mean of 
               $\{\widehat{\theta}_i^{\,\,(j)}\}_j$ as shown in~\eqref{eq:weighted-circular-average}
    \State $\widehat{f}_i \gets |\widehat{f}_i|\,  \mathrm{e}^{\mathrm{i}\widehat{\theta}_i}$
\EndFor
\end{algorithmic}
\end{algorithm}

\subsubsection{Exact Recovery in the Noiseless Case}

In the noiseless regime, where the $K$-band of $\widehat{Z}^{(K)}_{\vphantom{a}}$ coincides with that of $Z$, the algebraic estimator recovers the signal exactly whenever $f$ has no run of $K$ consecutive zeros.

\begin{theorem}[Exact Recovery for Block-By-Block Algebraic Estimator in the Noiseless case]\label{thm:exact-rec-algebraic-est}
The block-by-block algebraic estimator (Algorithm~\ref{algo:algebraic}) guarantees exact recovery of an unknown signal $f$ if and only if $f \in \mathcal{S}_K$, that is, $f$ contains no run of $K$ consecutive zeros.
\end{theorem}

\begin{proof}
    The proof is deferred to Appendix~\ref{app:proof-exact-recovery-algebraic-estimator}.
\end{proof}

\subsubsection{Stability in the Noisy Case}

The algebraic estimator above is stated for a general bandwidth $K$, where each new entry is inferred from a sliding window of the previous $K$ estimates. For the stability analysis, we focus on the simplest regime: signals in $\mathcal{S}_1$, meaning every entry of $f$ is nonzero. 
In this case, only the $1$-band of $Z$ needs to be measured, and each step of the recursion depends solely on the immediately preceding entry. 
This eliminates the additional combinatorial effects of larger windows and isolates the core phenomenon: the estimator propagates information forward one index at a time, so any noise in the measured matrix elements can accumulate as the reconstruction proceeds. 

The main question is, therefore, how accurately the $1$-band entries of $Z$ must be estimated to ensure that this error propagation remains controlled. 
The next theorem answers this, deriving explicit bounds on the number of circuit repetitions needed to guarantee that the final estimate $\widehat f$ remains within a prescribed error tolerance.

\begin{theorem}[Stability of Block-by-Block Algebraic Estimator]\label{thm:stability-base-case}
Let $f \in \mathbb{C}^T$ be an unknown signal as defined in~\eqref{def:discrete-time-signal} such that all its entries are nonzero, that is, $f \in \mathcal{S}_1$. Define $\gamma \coloneqq \min_i |f_i|^2 > 0$ and let $\eta > 0$. Suppose we estimate all the elements of the $1$-band of $Z = f f^\dagger$ as follows:
\begin{itemize}
    \item each diagonal element $Z_{ii}$ is estimated using $N_{\mathrm{diag}}$ repetitions of the circuit presented in Section~\ref{sec:measurement-diag};
    \item each real and imaginary part of each first off-diagonal element $Z_{i,i+1}$ is estimated using $N_{\mathrm{off}}$ repetitions of the corresponding circuit presented in Section~\ref{sec:measurement-off-diag}. 
\end{itemize}
If 
\begin{equation}\label{eq:circuit-repetitions-bound}
    N_{\mathrm{diag}} \ge \frac{8T}{\gamma \eta^2} \log\!\left(\frac{4T}{\delta}\right), \qquad N_{\mathrm{off}} \geq \frac{32T(2T^3 - 3T^2 + T)}{3\gamma \eta^4}
    \log\!\left(\frac{8(T-1)}{\delta}\right),
\end{equation}
then with probability at least $1 - \delta$, the block-by-block algebraic estimator (Algorithm~\ref{algo:algebraic}) produces an estimate $\widehat f$ satisfying
\begin{equation}
    \left \|\widehat f - f \right \|_2 \le \eta.
\end{equation}
\end{theorem}

\begin{proof}
    The proof is deferred to Appendix~\ref{app:stability-base-case}.
\end{proof}

\subsubsection{Query Complexity}

With the stability guarantee in place, we now turn to counting the total resource requirements. From~\eqref{eq:circuit-repetitions-bound}, we fix
\begin{equation}
    N_{\mathrm{diag}} = \left\lceil \frac{8T}{\gamma \eta^2} \log\!\left(\frac{4T}{\delta}\right) \right\rceil,
    \qquad
    N_{\mathrm{off}} = \left\lceil \frac{32T(2T^3 - 3T^2 + T)}{3\gamma \eta^4}
    \log\!\left(\frac{8(T-1)}{\delta}\right) \right\rceil,
    \label{eq:NdNoff-final-short}
\end{equation}
which ensures that $\|\widehat f - f\|_2 \le \eta$ with probability at least $1-\delta$.

Recall from Section~\ref{sec:lifted-representation} that the circuits used to estimate the entries of $Z$ rely on a small set of primitive unitaries (along with standard gates such as Hadamard, $R_z(\theta)$, and CNOT). 
These primitives are: 1) the state-preparation unitary $U_{\psi}$ and its inverse $U_{\psi}^\dagger$,  2) the time-evolution unitary $e^{\mathrm{i}H\Delta}$, and  3) the controlled time-evolution unitary, $\mathrm{controlled}\text{-}e^{\mathrm{i}H\Delta}$. 
We define the cost of our estimator in terms of the number of queries (query complexity) to these primitive unitaries. Note that a query to $U_{\psi}^{\dagger}$ is counted the same as a query to $U_{\psi}$, since it can be implemented by reversing the circuit of $U_{\psi}$.

Now recall that the sampled time points are $t_i = i \Delta$, where $i \in \{0, 1, \ldots T-1\}$. For estimating one diagonal element $Z_{ii}$, the circuit uses:
\begin{align}
    N_{\mathrm{diag}} \text{ queries each to } U_{\psi} \text{ and } U_{\psi}^{\dagger} & \implies 2 N_{\mathrm{diag}} \text{ queries to } U_{\psi},\\
    N_{\mathrm{diag}} \text{ queries to } \mathrm{e}^{-\mathrm{i}Ht_i} & \implies i N_{\mathrm{diag}} \text{ queries to } \mathrm{e}^{-\mathrm{i}H \Delta}.
\end{align}
Similarly, for estimating one off-diagonal element $Z_{i, i+1}$ (both real and imaginary parts combined), the circuits use:
\begin{align}
    2 N_{\mathrm{off}} \text{ queries each to } U_{\psi} \text{ and } U_{\psi}^{\dagger} & \implies 4N_{\mathrm{off}} \text{ queries to } U_{\psi},\\
    2 N_{\mathrm{off}} \text{ queries to } \mathrm{e}^{-\mathrm{i}Ht_i} & \implies 2 i N_{\mathrm{off}} \text{ queries to } \mathrm{e}^{-\mathrm{i}H \Delta}, \\
    2 N_{\mathrm{off}} \text{ queries to } \text{controlled-}\mathrm{e}^{-\mathrm{i}H(t_{i+1} - t_i)} & \implies 2 N_{\mathrm{off}} \text{ queries to } \text{controlled-}\mathrm{e}^{-\mathrm{i}H \Delta}.
\end{align}
Furthermore, there are $T-1$ diagonal (the first diagonal element is always 1 since $f_0 = 1$) and $T-1$ off-diagonal elements to be estimated. 
With that, the total queries to the aforementioned primitive unitaries are given in Table~\ref{tab:cost-summary-short}.
\begin{table}[h!]
\centering
\renewcommand{\arraystretch}{1.2}
\begin{tabular}{l l l}
\toprule
\textbf{Unitary} & \textbf{Source of calls} & \textbf{Total queries}\\
\midrule
$U_{\psi}$ or $U_{\psi}^{\dagger}$ & Diagonal + Off-diagonal & $2(T-1) N_{\mathrm{diag}} + 4(T-1) N_{\mathrm{off}}$\\
$\mathrm{e}^{-\mathrm{i}H\Delta}$ & Diagonal + Off-diagonal & $T(T-1) N_{\mathrm{diag}}/2 + (T-1)(T-2) N_{\mathrm{off}}$\\
$\mathrm{controlled}\text{-}\mathrm{e}^{-\mathrm{i}H\Delta}$ & Off-diagonal only & $2(T-1) N_{\mathrm{off}}$\\
\bottomrule
\end{tabular}
\caption{Total queries to the primitive unitaries required to estimate the $1$-band of $Z = f f^\dagger$.}
\label{tab:cost-summary-short}
\end{table}

\subsection{Block-By-Block Eigenvector Estimator}
\label{sec:block-by-block-eigenvector-estimator}

The block-by-block eigenvector estimator also reconstructs $f$ sequentially, but now by exploiting the local rank-one structure of $\widehat{Z}^{(K)}_{\vphantom{a}}$ on its overlapping $(K+1)\times(K+1)$ diagonal blocks. 
Refer to Figure~\ref{fig:eigenvector-estimator}, where we show an illustration with $K=2$, for a high-level understanding of how this estimator works, so that it is easier to follow the technical details of this section.

\subsubsection{Local Blocks and Eigenpairs}

We begin by defining the $b$-th $(K+1) \times (K+1)$ diagonal block of $\widehat{Z}^{(K)}_{\vphantom{a}}$ as
\begin{equation}\label{eq:Q-b-blocks-definition}
Q^{(b)} \coloneqq \widehat{Z}_{b:b+K, b:b+K}^{(K)}\,,
\end{equation}
where $b \in \{0, \ldots, T-K-1\}$ denotes the starting index of this block along the diagonal of $\widehat{Z}^{(K)}_{\vphantom{a}}$. 

In the noiseless case, this block relates to the true signal segment 
\begin{equation}
    f_{b:b+K} = (f_b, \dots, f_{b+K})^\top
\end{equation}
via the outer product
\begin{equation}\label{eq:Q-b-blocks-noiseless}
    Q^{(b)} = \widehat{Z}_{b:b+K, b:b+K}^{(K)} = Z_{b:b+K, b:b+K} = f_{b:b+K}^{\vphantom{\dagger}} f_{b:b+K}^{\dagger}.
\end{equation}
This means that $Q^{(b)}$, in the noiseless case, is positive semidefinite and rank-one, with the principal eigenvalue equal to $\|f_{b:b+K}\|_2^2$ and the principal eigenvector proportional to $f_{b:b+K}$.

In the noisy regime, the observed block $Q^{(b)}$ is a perturbation of this rank-one matrix. 
To extract a robust local signal estimate, the estimator computes the eigendecomposition of $Q^{(b)}$ and extracts the principal eigenpair
\begin{equation}
\left(\lambda_{\max}^{(b)},\, v_{\max}^{(b)}\right ),
\end{equation}
where $ v_{\max}^{(b)} \in \mathbb{C}^{K+1}$ has unit norm. The local $(K+1)$-length estimate of the true signal $f$ on indices $b, \ldots, b+K$, that is, an estimate of the true signal segment $f_{b:b+K}$, is then
\begin{equation}
    \widehat{f}^{\,(b)} = \sqrt{\lambda_{\max}^{(b)}}\cdot v_{\max}^{(b)}.
\end{equation}

The estimator repeats this process for every $b \in \{0, \ldots, T-K-1\}$. 

\begin{figure}
    \centering
    \includegraphics[width=\linewidth]{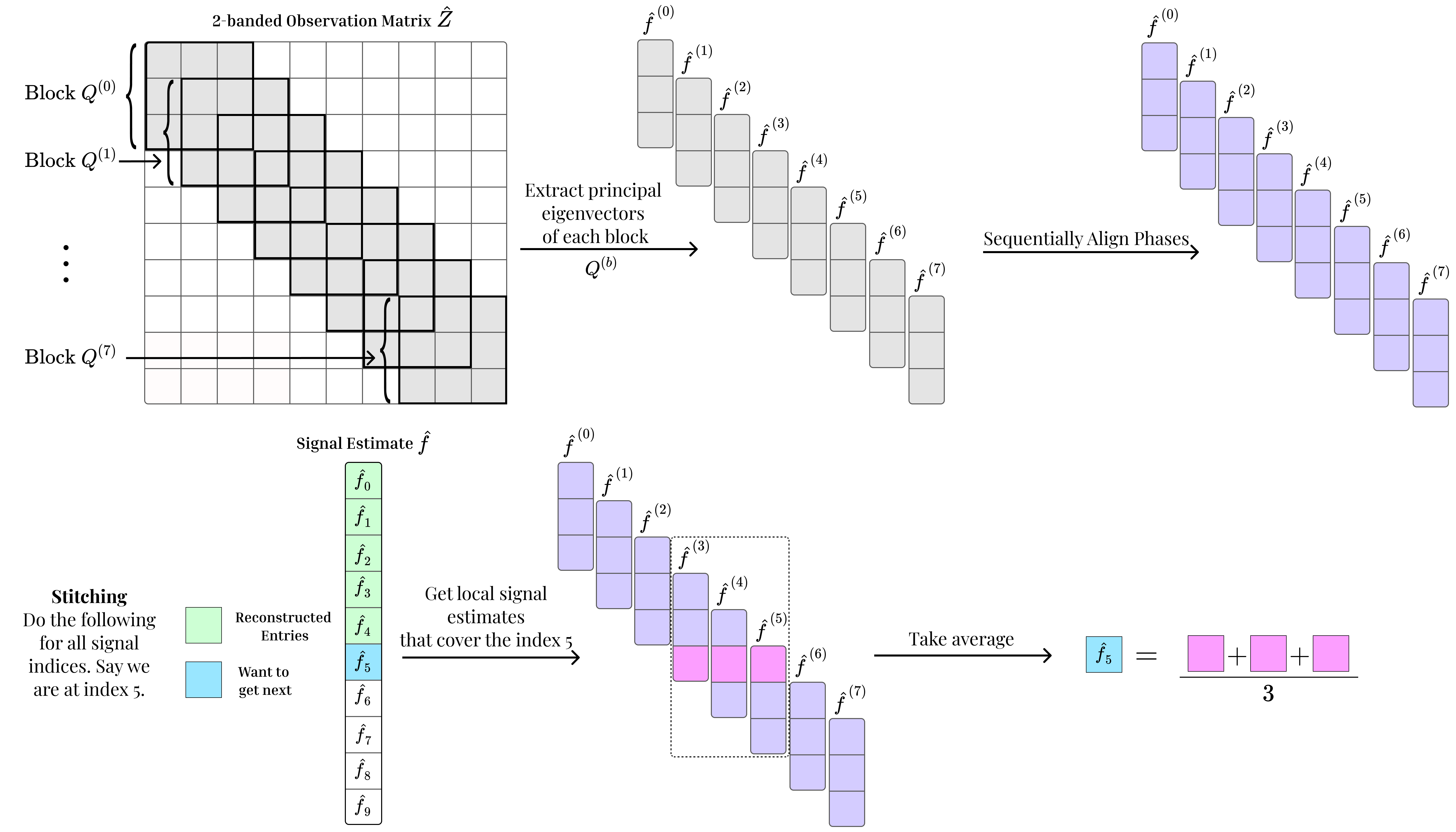}
    \caption{Schematic of the Block-by-Block Eigenvector Estimator (shown for $K=2$).}
    \label{fig:eigenvector-estimator}
\end{figure}

\subsubsection{Sequential Phase Alignment}

Since each local estimate $\widehat{f}^{\,(b)}$ is derived from an independent eigendecomposition, it is defined only up to an arbitrary global phase. 
Consequently, these estimates cannot simply be concatenated to form a global signal estimate $\widehat{f}$. Instead, it is helpful to view them as initial, unaligned estimates of the corresponding signal segments.

A key observation is that consecutive local estimates overlap in their indices: the local estimate $\widehat{f}^{\,(b)}$, which provides an initial estimate for $f_{b:b+K}$, and the estimate $\widehat{f}^{\,(b+1)}$ for $f_{b+1:b+K+1}$ share the $K$ indices $b+1,\ldots,b+K$. Refer to Figure~\ref{fig:eigenvector-estimator} to visualize this overlap. 
This overlap provides a common region where the two local estimates describe the same underlying signal values, differing only by a phase factor. 
The estimator exploits such shared regions between local estimates to align their phases before stitching them together.

The estimator performs this alignment sequentially, proceeding from the first local estimate $\widehat{f}^{\,(0)}$ to the last $\widehat{f}^{\,(T-K-1)}$. 
It begins by anchoring the first local estimate $\widehat{f}^{\,(0)}$. 
Given that we have the constraint $f_0 = 1$, the estimator simply rotates and rescales $\widehat{f}^{\,(0)}$ so that its first entry is exactly equal to one, that is, $\widehat{f}_0^{\,\,(0)} = 1$.

For every subsequent local estimate $\widehat{f}^{\,(b)}$ with $b \geq 1$, the estimator relies on the overlap between $\widehat{f}^{\,(b)}$ and the set of previously processed (that is, phase-aligned) local estimates $\widehat{f}^{\,(0)}, \widehat{f}^{\,(1)}, \ldots, \widehat{f}^{\,(b-1)}$.
For phase-alignment for this current local estimate $\widehat{f}^{\,(b)}$, the estimator aligns it against the "consensus" of these previously phase-aligned estimates within the overlap region. 

To understand this more formally, we define two vectors $u, v \in \mathbb{C}^K$ to represent the current local estimate and the consensus estimate, respectively, on the $K$ overlapping indices. 
Specifically, the vector $u$ consists of the first $K$ entries of the current local estimate $\widehat{f}^{\,(b)}$. 
We exclude the final entry, as it does not overlap with any of the previous local estimates. To this end, we define
\begin{equation}
    u_k \coloneqq \widehat{f}^{\,\,(b)}_k, \quad \text{for } k=0, \ldots, K-1.
\end{equation} 
Furthermore, the vector $v$ represents the current consensus estimate at the corresponding global indices $i = b+k$, for $k=0, \ldots, K-1$. 
The estimator computes $v_k$ by averaging the aligned values from all previously phase-aligned local estimates $\widehat{f}^{\,(p)}$, where $p < b$, that cover the index $i$:
\begin{equation}
    v_k = \frac{1}{|\mathcal{P}_{i}|} \sum_{p \in \mathcal{P}_{i}} \widehat{f}^{\,\,(p)}_{i-p},
\end{equation}
where 
\begin{equation}
    \mathcal{P}_{i} = \{ p \in \{0, \dots, b-1\} \mid p \le i \le p+K \}
\end{equation}
is the set of previous blocks that overlap at global index $i$.
One can also consider a weighted version of this average if the data are heteroskedastic, but we leave analysis of such estimators to future work. 

The estimator then seeks a phase rotation $\phi$ that aligns $u$ with $v$. 
Mathematically, this corresponds to minimizing the squared Euclidean distance between $v$ and $\mathrm{e}^{\mathrm{i}\phi} u$:
\begin{equation}
    \min_{\phi \in [0,2\pi)} \left\| v - \mathrm{e}^{\mathrm{i}\phi} u \right\|_2^2.
\end{equation}
Expanding this norm, we get
\begin{align}
    \left\| v - \mathrm{e}^{\mathrm{i}\phi} u \right\|_2^2 &= \langle v - \mathrm{e}^{\mathrm{i}\phi} u, v - \mathrm{e}^{\mathrm{i}\phi} u \rangle \\
    &= \|v\|^2 + \|u\|^2 - 2 \operatorname{Re}\left[ \overline{\mathrm{e}^{\mathrm{i}\phi}} \langle u, v \rangle \right] \\
    &= \|v\|^2 + \|u\|^2 - 2 \operatorname{Re}\left[ \mathrm{e}^{-\mathrm{i}\phi} \langle u, v \rangle \right].
\end{align}
To minimize the error, the term $\operatorname{Re}[ \mathrm{e}^{-\mathrm{i}\phi} \langle u, v \rangle ]$ must be maximized since this is the only term dependent on $\phi$. 
This is achieved by choosing $\phi$ to cancel out the phase of $\langle u, v \rangle$, effectively rotating the complex number $\langle u, v \rangle$ onto the positive real axis. 
Therefore, the optimal phase shift is simply the argument of $\langle u, v \rangle$:
\begin{equation}
    \phi_{\text{opt}} = \arg(\langle u, v \rangle).
\end{equation}
This expression directly implies that the alignment is determined by a weighted consensus over the overlap region. 
Since the optimal phase is derived from the sum $\langle u, v \rangle = \sum_k \overline{u}_k v_k$, terms with larger magnitudes $|\overline{u}_k v_k|$ contribute more weight to the final result. 
This naturally enhances robustness, as it prevents the alignment from being skewed by low-amplitude indices where the signal-to-noise ratio is poor.

Once $\phi_{\text{opt}}$ is computed, the estimator updates the current local estimate by applying the rotation
\begin{equation}
    \widehat{f}^{\,(b)} \leftarrow e^{\mathrm{i}\phi_{\text{opt}}} \widehat{f}^{\,(b)}.
\end{equation}

\subsubsection{Stitching}

Once all the local estimates $\widehat{f}^{\,(b)}$ have been phase-aligned, the estimator stitches them together to form the final global signal estimate $\widehat{f}$. 
Since the local estimates overlap, most indices $i$ of the signal are covered by multiple local estimates. 
To mitigate noise, the estimator computes the final estimate $\widehat{f}_i$ by averaging all aligned local estimates that cover that specific index.

Formally, for each index $i$, the estimator averages all aligned local signal estimates that cover that index in the following way:
\begin{equation}\label{eq:stitching-estimates}
    \widehat{f}_i = \frac{1}{N_i} \sum_{b=\max(0, i-K)}^{\min(i, T-K-1)} \widehat{f}^{\,\,(b)}_{i-b},
\end{equation}
where $N_i$ is the number of local estimates covering index $i$, and it is given as
\begin{equation}
    N_i = \min(i, T-K-1) - \max(0, i-K) + 1.
\end{equation}
In the bulk of the signal (specifically for indices $i$ where $K \le i < T-K$), the formula in~\eqref{eq:stitching-estimates} averages $K+1$ independent estimates. 
Near the boundaries, the number of overlapping terms naturally decreases.
To this end, a pseudocode of this estimator is given by Algorithm~\ref{algo:eigen-estimator-averaging}.

\begin{algorithm}[t]
\caption{Block-by-Block Eigenvector Estimator}
\label{algo:eigen-estimator-averaging} 
\begin{algorithmic}[1]
\State \textbf{Input:} $K$-banded observation matrix $\widehat{Z}^{(K)}_{\vphantom{a}} \in \mathbb{C}^{T \times T}$
\State \textbf{Output:} Signal estimate $\widehat{f} \in \mathbb{C}^T$

\State Extract overlapping blocks of size $(K+1) \times (K+1)$: $Q^{(b)} \gets \widehat{Z}_{b:b+K,b:b+K}$ for all $b=0,\ldots,T-K-1$
\For{each block $Q^{(b)}$}
    \State Compute principal eigenvector $v_{\max}^{(b)}$ and eigenvalue $\lambda_{\max}^{(b)}$ of $Q^{(b)}$
    \State Store Local estimate: $\widehat{f}^{\,(b)} = \sqrt{\lambda_{\max}^{(b)}} \cdot v_{\max}^{(b)}$
\EndFor
\State Rotate and rescale $\widehat{f}^{\,(0)}$ so that its first entry is 1
\For{$b=1$ to $T-K-1$}
    \State $u \gets (\widehat{f}^{\,\,(b)}_0, \dots, \widehat{f}^{\,\,(b)}_{K-1})^\top$ \Comment{First $K$ entries of current local estimate $\widehat{f}^{\,(b)}$}
    \For{$k=0$ to $K-1$} \Comment{Construct consensus estimate $v$}
        \State $i \gets b + k$ \Comment{Corresponding global index}
        \State $\mathcal{P}_i \gets \{ p \in \{0, \dots, b-1\} \mid p \le i \le p+K \}$ \Comment{Previous local estimates $\widehat{f}^{(p)}$ covering $i$}
        \State $v_k \gets \frac{1}{|\mathcal{P}_i|} \sum_{p \in \mathcal{P}_i} \widehat{f}^{\,\,(p)}_{i-p}$\label{step:average-prev-aligned-estimates} \Comment{Average of previous aligned estimates}
    \EndFor
    \State Compute $\langle u, v \rangle$
    \State Compute phase: $\phi_{\mathrm{opt}} \gets \arg(\langle u, v \rangle)$
    \State Align current local estimate: $\widehat{f}^{\,(b)} \gets \widehat{f}^{\,(b)} \mathrm{e}^{\mathrm{i}\phi_{\mathrm{opt}}}$
\EndFor
\For{$i = 0$ to $T-1$}
    \State $\mathcal{B}_i \gets \{b \mid b \le i \le b+K\}$ \Comment{All local estimates $\widehat{f}^{\,(b)}$  covering global index $i$}
    \State $\widehat{f}_i \gets \frac{1}{|\mathcal{B}_i|} \sum_{b \in \mathcal{B}_i} \widehat{f}^{\,\,(b)}_{i-b}$ \Comment{Average of all aligned estimates}
\EndFor
\end{algorithmic}
\end{algorithm}

\subsubsection{Exact Recovery in the Noiseless Case}

Here we examine how the estimator behaves in the noiseless regime, where the $K$-band of $\widehat{Z}^{(K)}_{\vphantom{a}}$ matches that of $Z$. 
The theorem below shows that, under these ideal conditions, the eigenvector estimator recovers the signal exactly whenever $f$ has no run of $K$ consecutive zeros.

\begin{theorem}[Exact Recovery for Block-By-Block Eigenvector Estimator in the Noiseless Case]\label{thm:exact-rec-eigen-est}
The block-by-block eigenvector estimator (Algorithm~\ref{algo:eigen-estimator-averaging}) guarantees exact recovery of an unknown signal $f$ if and only if $f \in \mathcal{S}_K$, that is, $f$ contains no run of $K$ consecutive zeros.
\end{theorem}

\begin{proof}
    The proof is deferred to Appendix~\ref{app:proof-exact-recovery-eigenvector-estimator}.
\end{proof}

\subsection{Least-Squares Estimator} \label{sec:quantum-phaselift-estimator}

Unlike the block-by-block estimators, which reconstruct an unknown signal $f$ sequentially by operating locally on small diagonal blocks of $\widehat{Z}^{(K)}_{\vphantom{a}}$, the least-squares estimator treats the problem as a matrix completion problem, where the aim is to reconstruct the entire lifted matrix $Z$ using the $K$-band measurements in one step by solving a least-squares optimization problem and then returning the (unnormalized) principal eigenvector of this reconstructed matrix as a signal estimate~$\widehat{f}$.

To formulate this optimization problem, we first need a convenient way to express our measurements. 
Each entry we measure, whether a real diagonal value or the real and imaginary parts of an off-diagonal, can be viewed as the result of applying a specific linear functional to the matrix $Z$. 
To formalize this, we define the following Hermitian matrices:
\begin{equation}
    A^{(i, i)} \coloneqq |i\rangle \langle i|, \qquad A^{(i, j, \text{Re})} \coloneqq \tfrac{1}{2}\left (|i\rangle \langle j| +  |j\rangle \langle i|\right), \qquad A^{(i, j, \text{Im})} \coloneqq \tfrac{1}{2\mathrm{i}}\left (|i\rangle \langle j| -  |j\rangle \langle i|\right).
\end{equation}
It is easy to see that for any Hermitian matrix $X$, these operators essentially extract the entries we are interested in:
\begin{equation}
    \operatorname{Tr}\!\left[A^{(i, i)} X\right]=X_{ii},\quad
\operatorname{Tr}\!\left[A^{(i, j, \text{Re})} X\right]=\operatorname{Re}\left[X_{ij}\right],\quad
\operatorname{Tr}\!\left[A^{(i, j, \text{Im})} X\right]=\operatorname{Im}\left [X_{ij}\right].
\end{equation}

The observations of the $K$-band entries of $Z$ thus induce the following collection of real-valued measurements:
\begin{equation}
    \left \{ b^{(i, i)}\right \}_{i=0}^{T-1}, \qquad \left \{ b^{(i, j, \text{Re})}, b^{(i, j, \text{Im})} \right\}_{\substack{0\le i<j\le T-1\\ 1\le j-i\le K}}
\end{equation}
with 
\begin{align}
    b^{(i, i)} & \coloneqq \operatorname{Tr}\!\left[A^{(i, i)} Z\right] + \varepsilon^{(i, i)} = \operatorname{Tr}\!\left[A^{(i, i)} \widehat{Z}^{(K)}_{\vphantom{a}}\right]\\
    b^{(i, j, \text{Re})} & \coloneqq \operatorname{Tr}\!\left[A^{(i, j, \text{Re})} Z\right] + \varepsilon^{(i, j, \text{Re})} = \operatorname{Tr}\!\left[A^{(i, j, \text{Re})} \widehat{Z}^{(K)}_{\vphantom{a}}\right]\\
    b^{(i, j, \text{Im})} & \coloneqq \operatorname{Tr}\!\left[A^{(i, j, \text{Im})} Z\right] + \varepsilon^{(i, j, \text{Im})} = \operatorname{Tr}\!\left[A^{(i, j, \text{Im})} \widehat{Z}^{(K)}_{\vphantom{a}}\right].
\end{align}
Here $\varepsilon^{(i,i)}$ denotes the error due to the measurement noise on diagonal entries for all $i \in \{0,\ldots,T-1\}$, while $\varepsilon^{(i,j,\mathrm{Re})}$ and $\varepsilon^{(i,j,\mathrm{Im})}$ denote the errors due to the measurement noise on the real and imaginary parts of the off-diagonal entries for all pairs $(i,j)$ with $0 \leq i < j \leq T-1$ and $1 \leq j-i \leq K$. 

Finally, stacking all the measurement matrices into a single ordered collection
\begin{equation}
    \{A_\ell\}_{\ell=1}^m = \{A^{(i,i)} : 0 \le i \le T-1\} \cup\, \{A^{(i,j,\text{Re})}, A^{(i,j,\text{Im})} : 0 \le i < j \le T-1,  1 \le j-i \le K\} \label{eq:A-l-measurement-matrices},
\end{equation}
where $m =  T+ K (2T-K-1)$, since there are $T$ diagonal measurements and, for each offset $k = 1,\ldots,K$, there are $T-k$ off-diagonal entries, each contributing measurements for both real and imaginary parts.
Similarly, stacking the corresponding noisy observations and measurement errors in a single vector
\begin{align}
    b & \coloneqq (b_1, b_2, \ldots, b_m)^{\mathsf{T}},\\
    \varepsilon &  \coloneqq (\varepsilon_1, \varepsilon_2, \ldots, \varepsilon_m)^{\mathsf{T}},\label{eq:b-l-observations}
\end{align}
we arrive at the following compact form:
\begin{equation}
    b_\ell = \operatorname{Tr}\left [A_\ell Z\right] + \varepsilon_\ell, 
    \qquad \text{for all } \ell = 1, \dots, m.
\end{equation}

The estimator then solves the least-squares problem
\begin{equation}\label{eq:least-square-problem}
    X^* = \operatorname*{argmin}_{X \succeq 0} 
    \sum_{\ell = 1}^{m} \left( \operatorname{Tr}[A_{\ell} X] - b_{\ell} \right)^2,
\end{equation}
seeking the PSD matrix most consistent with the noisy $K$-band measurements. 
We restrict the search to the PSD cone because the true matrix $Z = f f^\dagger$ is rank one and positive semidefinite. 
The constraint $X \succeq 0$ is convex, and each map $X \mapsto \operatorname{Tr}[A_\ell X]$ is real-valued on Hermitian inputs, so the objective is a real least-squares loss over a convex feasible set.

Once $X^*$ is obtained, the final step is to extract a signal estimate from it. 
In the noiseless setting, the minimizer of the above least-squares problem is unique, and it is exactly equal to the lifted matrix $Z$, that is, $X^* = Z = f f^{\dagger}$ (we will see rigorously why this is the case in the proof of Theorem~\ref{thm:exact-recovery-quantum-phaselift-estimator}). 
This implies that the principal eigenvector of $X^*$ is proportional to $f$ up to some global phase. 
This global phase is arbitrary, and we can fix it by using the fact that in our setting we have $f_0=1$. 

In the noisy setting, $X^*$ is simply a perturbation of $Z$, and its principal eigenvector still provides the best rank-one approximation in the sense of the spectral norm. 
We therefore take $v_{\max}$, the normalized principal eigenvector of $X^*$, scale it by the square-root of the corresponding eigenvalue $\lambda_{\max}$, that is, $\sqrt{\lambda_{\max}}$, and fix its global phase by enforcing $f_0 = 1$. 
The resulting vector is returned as the signal estimate $\widehat{f}$. 
To this end, a pseudocode for the least-squares estimator is given in Algorithm~\ref{algo:least-squares}.

\begin{algorithm}[t]
\caption{Least-Squares Estimator}
\label{algo:least-squares} 
\begin{algorithmic}[1]
\State \textbf{Input:} $K$-banded observation matrix $\widehat{Z} \in \mathbb{C}^{T \times T}$
\State \textbf{Output:} Signal estimate $\widehat{f}  \in \mathbb{C}^T$
\State Set $m \gets T+ K (2T-K-1)$
\State Initialize measurement matrices $A_1 \ldots A_m \in \mathbb{C}^{T\times T}$ according to~\eqref{eq:A-l-measurement-matrices}
\State Initialize observation vector $b \in \mathbb{R}^{m}$ according to~\eqref{eq:b-l-observations}
\State Solve the least-squares optimization problem (see~\eqref{eq:least-square-problem})
\begin{equation}
    X^*  = \operatorname*{argmin}_{X \succeq 0} \sum_{\ell = 1}^{m} \left ( \operatorname{Tr}\!\left [ A_{\ell} X\right] - b_{\ell} \right)^2
\end{equation}
\State Compute principal eigenvector $v_{\max}$ and eigenvalue $\lambda_{\max}$ of $X^*$
\State Rotate and rescale the vector $\sqrt{\lambda_{\max}} \cdot v_{\max}$ so that its first entry is 1
\State Write $\widehat f \gets \sqrt{\lambda_{\max}} \cdot v_{\max}$
\end{algorithmic}
\end{algorithm}

\subsubsection{Exact Recovery in the Noiseless Case}

In the noiseless setting, with exact $K$-band measurements and $f \in \mathcal{S}_K$, the least-squares estimator recovers the true signal exactly.

\begin{theorem}[Exact Recovery for Least-Squares Estimator  in the Noiseless Case]\label{thm:exact-recovery-quantum-phaselift-estimator}
The least-squares estimator (Algorithm~\ref{algo:least-squares}) guarantees exact recovery of an unknown signal $f$ if and only if $f \in \mathcal{S}_K$, that is, $f$ contains no run of $K$ consecutive zeros.
\end{theorem}
\begin{proof}
    The proof is deferred to Appendix~\ref{app:proof-exact-recovery-phaselift-estimator}.
\end{proof}

\subsubsection{Stability in the Noisy Case}

In this section, we establish the stability of the least-squares estimator in the presence of measurement noise. 
We adopt the geometric framework of the descent cone analysis, which is standard in the literature of robust low-rank matrix recovery~\cite{Chandrasekaran2012}. 
This approach relies on the concept of the smallest conic singular value of the measurement map restricted to the set of feasible error directions.

Let $\mathcal{A}: \mathbb{C}^{T \times T} \to \mathbb{R}^{m}$ be the linear measurement map induced by the collection $\{A_{\ell}\}_{\ell}$ (see~\eqref{eq:A-l-measurement-matrices}) such that 
\begin{equation}\label{eq:measurement-map}
    [\mathcal{A}(X)]_{\ell} = \operatorname{Tr}[A_\ell X]
\end{equation}
Then the observations can be written compactly as $b = \mathcal{A}(Z) + \varepsilon$, and the estimator minimizes the least-squares loss $L(X) \coloneqq \|\mathcal{A}(X) - b\|^2_2 $ subject to $X \succeq 0$. 

Since $X^*$ minimizes $L$ and $Z$ is a feasible solution, the error direction $X^* - Z$ must lie in a restricted set of directions known as the descent cone at $Z$. 
Intuitively, this cone contains all possible directions $\Delta$ that constitute a feasible perturbation from $Z$ and along which the objective function $L$ does not increase. 
More formally, we define it as follows:
\begin{equation}\label{def:descent-cone}
\mathcal{D}(L, Z) \coloneqq \left \{\Delta \in \mathbb{C}^{T \times T}: \exists \tau > 0 \text{ s.t. } Z + \tau \Delta \succeq 0 \text{ and } L(Z + \tau \Delta) \leq L(Z) \right\}.
\end{equation}
Furthermore, as we will see later in Theorem~\ref{thm:stability-proof-least-squares-estimator}, the stability of the recovery depends on the smallest conic singular value of the measurement map $\mathcal{A}$ restricted to the cone $\mathcal{D}(L, Z)$. Formally, it is defined as follows:
\begin{equation}\label{eq:mini-conic-singular-value}
    \sigma_{\mathrm{min}}(\mathcal{A}, \mathcal{D}(L, Z)) \coloneqq \inf_{\Delta \in \mathcal{D}(L, Z)\setminus \{0\}} \frac{\left \| \mathcal{A}(\Delta) \right \|_2}{ \|\Delta\|_2}.
\end{equation}

First and foremost, it is important to verify that $\sigma_{\mathrm{min}}(\mathcal{A}, \mathcal{D}(L, Z)) > 0$ because, as we will see later in Theorem~\ref{thm:stability-proof-least-squares-estimator}, the stability bound, which is given by~\eqref{eq:stability-bound}, explicitly depends on the inverse of $\sigma_{\mathrm{min}}(\mathcal{A}, \mathcal{D}(L, Z))$. 
By definition, the condition $\sigma_{\mathrm{min}}(\mathcal{A}, \mathcal{D}(L, Z)) > 0$ holds if and only if the numerator never goes to zero for any non-zero perturbation $\Delta$. 
Equivalently, the measurement map $\mathcal{A}$ must not annihilate any non-zero direction in the descent cone; that is, the intersection of the cone with the kernel of the map must be trivial: $\mathcal{D}(L, Z) \cap \ker(\mathcal{A}) = \{0\}$. 
This condition is guaranteed by the fact that $\mathcal{A}$ is injective on the set of signals $\mathcal{S}_K$, a property we established before in the exact recovery proof of Theorem~\ref{thm:exact-recovery-quantum-phaselift-estimator}. 
The following lemma formalizes the equivalence between the injectivity of the map $\mathcal{A}$ on the signal set $\mathcal{S}_K$ and the trivial intersection of the descent cone with the kernel of this map.
\begin{lemma}\label{lemma:injectivity-kernel-equivalence}
    The measurement map $\mathcal{A}$ is injective on the signal set $\mathcal{S}_K$ if and only if $\mathcal{D}(L, Z) \cap \ker(\mathcal{A}) = \{0\}$ for all $Z = ff^{\dagger}$ generated by signals $f \in \mathcal{S}_K$.
\end{lemma}
\begin{proof}
    The proof is deferred to Appendix~\ref{app:proof-for-lemma-injectivity-kernel-equivalence}.
\end{proof}

With that, we now state the main stability result of the least-squares estimator.

\begin{theorem}[Stability of Least-Squares Estimator]\label{thm:stability-proof-least-squares-estimator}
    Let $f \in \mathcal{S}_K$ and let $Z = ff^{\dagger}$ be the corresponding lifted matrix. Suppose the least-squares estimator (Algorithm~\ref{algo:least-squares}) outputs $\widehat{f}$ as the final signal estimate, then the following holds up to some global phase $\phi \in [0, 2\pi)$:
    \begin{equation}\label{eq:stability-bound}
        \left \| \widehat{f} - \mathrm{e}^{\mathrm{i} \phi}f \right \|_2 \leq \frac{2(2\sqrt{2} + 1)}{\sigma_{\mathrm{min}}(\mathcal{A}, \mathcal{D}(L, Z)) } \cdot \frac{\left \|\varepsilon\right \|_2}{\left \| f\right \|_2}, 
    \end{equation}
    where $\sigma_{\mathrm{min}}(\mathcal{A}, \mathcal{D}(L, Z))$ is defined as in~\eqref{eq:mini-conic-singular-value} and $\varepsilon$ is the measurement noise vector defined in~\eqref{eq:b-l-observations}.
\end{theorem}
\begin{proof}
    The proof is deferred to Appendix~\ref{app:stability-proof-least-squares-estimator}.
\end{proof}

It is important to emphasize that, unlike the analysis for the algebraic estimator, the stability result in Theorem~\ref{thm:stability-proof-least-squares-estimator} does not provide an explicit query complexity bound. 
This is because the stability depends critically on the value of the smallest conic singular value $\sigma_{\mathrm{min}}(\mathcal{A}, \mathcal{D}(L, Z))$, for which we do not have a closed-form analytical expression in terms of the parameters $T$ and $K$. 
In the literature of phase retrieval, obtaining meaningful lower bounds on $\sigma_{\mathrm{min}}(\mathcal{A}, \mathcal{D}(L, Z))$ is one of the central tasks. 
Existing results typically rely on strong randomness assumptions on the measurement ensemble and on sophisticated concentration-of-measure arguments. 
In contrast, the measurement map considered here is deterministic and highly structured, corresponding to banded observations, which places it outside the scope of these techniques.
We therefore leave the derivation of explicit lower bounds on $\sigma_{\mathrm{min}}(\mathcal{A}, \mathcal{D}(L,Z))$ to future work and, for this reason, refrain from stating an explicit query complexity bound for the least-squares estimator.

\section{Numerical Simulations}
\label{sec:numerical-simulations}

In this section, we numerically investigate the performance of all three estimators, that is, the block-by-block algebraic estimator (Algorithm~\ref{algo:algebraic}), block-by-block eigenvector estimator (Algorithm~\ref{algo:eigen-estimator-averaging}), and least-squares estimator (Algorithm~\ref{algo:least-squares}). 

\subsection{Hamiltonians}
We generate three families of datasets with signal entries $f_i = \langle\psi| \mathrm{e}^{\mathrm{i} H t_i}|\psi\rangle$ sampled at times $t_i = i \Delta$, and for the purposes of our simulations, we fix the signal length to $T=50$. 
We will analyze time-series from the following two Hamiltonians:

\paragraph{Fermi-Hubbard Model.}
The Fermi-Hubbard Hamiltonian is
\begin{equation}\label{eq:fermi-hubbard}
H_{\mathrm{FH}}
= -t_0 \sum_{\langle i,j \rangle,\sigma}
\left(
  \widehat{c}^{\,\,\dagger}_{i\sigma} \widehat{c}_{j\sigma}^{\vphantom{\dagger}}
  + \widehat{c}^{\,\,\dagger}_{j\sigma} \widehat{c}_{i\sigma}^{\vphantom{\dagger}}
\right)
+ U \sum_{i} \widehat{n}_{i\uparrow} \widehat{n}_{i\downarrow},
\end{equation}
where $\widehat{c}^{\,\,\dagger}_{i\sigma}$ and $\widehat{c}_{i\sigma}\vphantom{\dagger}$ are fermionic creation and annihilation operators, respectively, on site $i$ with spin $\sigma \in \{\uparrow,\downarrow\}$ and $\widehat{n}_{i\sigma} = \widehat{c}^{\,\,\dagger}_{i\sigma} \widehat{c}_{i\sigma}\vphantom{\dagger}$.
The parameters $t_0$ and $U$ denote the hopping strength and the on-site interaction, respectively. 
The sum over $\langle i,j \rangle$ indicates a sum over nearest neighbors on the underlying lattice geometry. 

\paragraph{Transverse-Field Ising Model.} 
The transverse-field Ising Hamiltonian is
\begin{equation}\label{eq:tf-ising-model}
H_{\mathrm{ISING}}
= -J \sum_{\langle i,j \rangle} \widehat{\sigma}^{\,z}_i \widehat{\sigma}^{\,z}_j
 - h_z \sum_i \widehat{\sigma}^{\,z}_i
 - h_x \sum_i \widehat{\sigma}^{\,x}_i,
\end{equation}
where $\widehat{\sigma}^{\,x}_i$ and $\widehat{\sigma}^{\,z}_i$ denote Pauli operators acting on site $i$, and the parameters $J$, $h_x$, and $h_z$ denote the coupling strengths.
\newline

\subsection{Datasets}
\label{sec:datasets}

We obtain each dataset by using a first-order Trotter formula to simulate the unitary evolution $\mathrm{e}^{\mathrm{i} H t_i}$. 
We use $i$ Trotter steps to simulate each time $t_i$, and each Trotter step has step size $\Delta$. 
Note that we are simply interested in testing the proposed estimators on typical signals one might encounter in practice, not their physical relevance or accuracy, and so are not concerned with the accuracy of Trotterized dynamics itself.

\paragraph{First Dataset.} 
The first dataset considers a 2D lattice geometry for the aforementioned Fermi-Hubbard model, more concretely, a $2\times 2$ square lattice. 
The parameters $t_0$ and $U$ are drawn independently and uniformly at random from an interval $t_0 \in [0,1]$ and $U \in [0,4]$.
As input state, we consider the following coherent Gibbs state:
\begin{equation}
    |\psi\rangle=\sum_k \sqrt{p_k}|E_k\rangle \ \ \ \ \text{with} \ \ \ \ p_k=\frac{e^{-\beta E_k}}{Z} \ \ \ \ \text{and} \ \ \ \ Z=\sum_k e^{-\beta E_k},
\end{equation}
where $\{\ket{E_k}\}_k$ and $\{E_k\}_k$ are the eigenvectors and eigenvalues, respectively, of $H_{\mathrm{FH}}$.
For this dataset, we fix the input state with $\beta = 0.5$, and we set the sampling interval as $\Delta = 0.2$.

\paragraph{Second Dataset.}
The second dataset also considers a 2D lattice geometry for the Fermi-Hubbard Model, but with a $2\times 3$ rectangular lattice.
For this dataset, we fix the input state as a simple product Fock state with exactly one fermion in the spin-up sector and one fermion in the spin-down sector, and we set the sampling interval as $\Delta = 0.05$. 
The parameters $J$, $h_x$, and $h_z$ are drawn independently and uniformly at random from an interval $[0,1]$. 

\paragraph{Third Dataset.} 
The third dataset considers the aforementioned transverse-field Ising model on a $4 \times 3$ rectangular lattice.
Here, we fix the input state as the uniform superposition state, and we set the sampling interval as $\Delta = 0.05$. The parameters $J$, $h_x$, and $h_z$ are drawn independently and uniformly at random from an interval $[0,1]$. 

\subsection{Time-Series Recovery Examples}

We begin by illustrating the performance of our three estimators on some example signals from each of the three datasets presented above.
Figure~\ref{fig:FH-example-reconstruction-high-shots} shows the real and imaginary parts of the exact signal $f$ alongside the reconstructed signal.
For each example, we estimate all entries of the $K$-band of the lifted matrix $Z$ with bandwidth $K=6$, using $3000$ shots per matrix entry. 
\begin{figure}[!ht]
\begin{subfigure}[h]{\textwidth}
\centering
\includegraphics[width=0.9\linewidth]{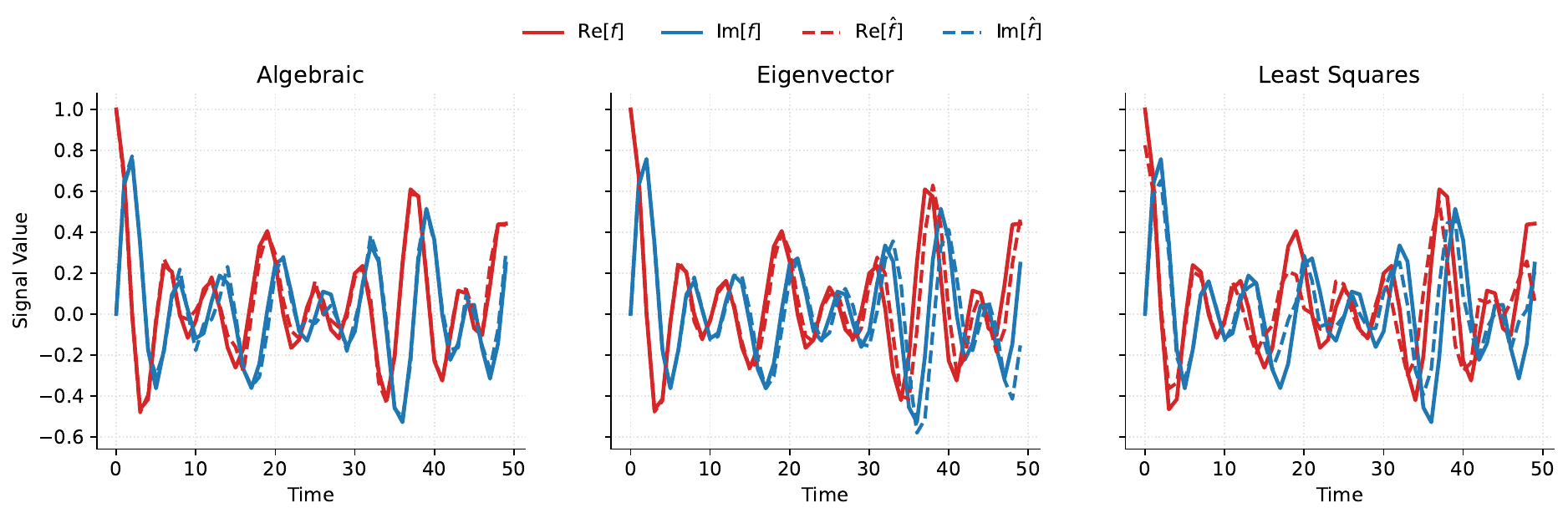}
\caption{$2 \times 2$ Fermi-Hubbard Model}
\end{subfigure}
\begin{subfigure}[h]{\textwidth}
\centering
\includegraphics[width=0.9\linewidth]{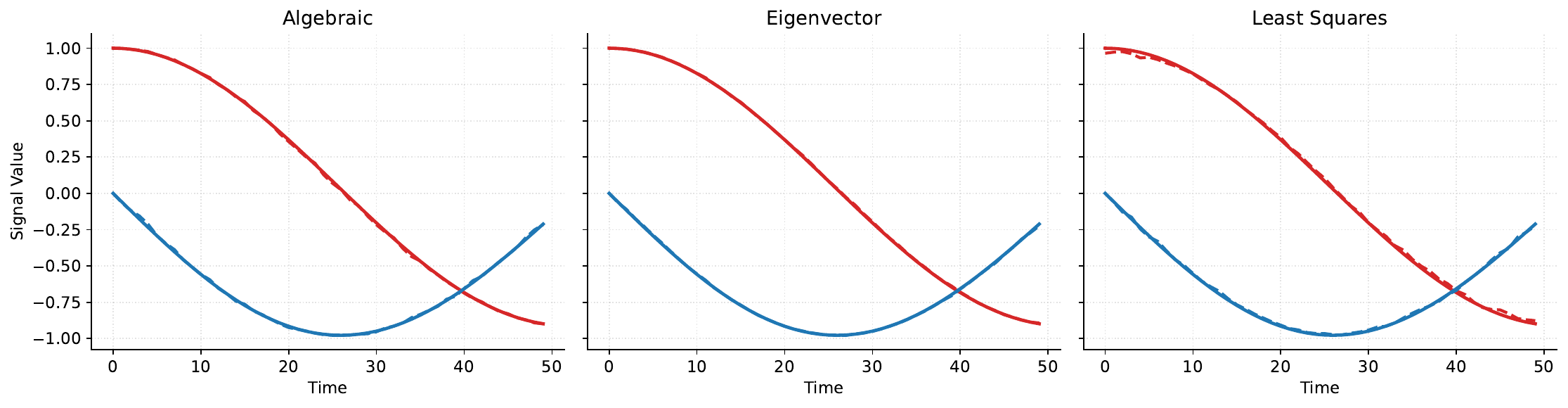}
\caption{$2 \times 3$ Fermi-Hubbard Model}
\end{subfigure}
\begin{subfigure}[h]{\textwidth}
\centering
\includegraphics[width=0.9\linewidth]{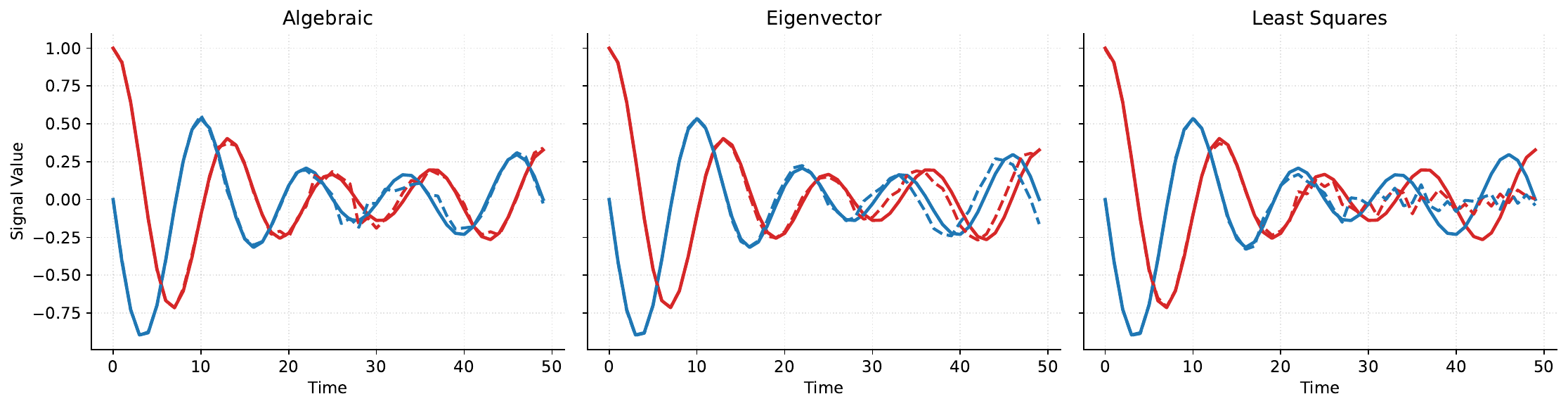}
\caption{$4 \times 3$ Transverse-Field Ising Model}
\end{subfigure}
\caption{ Signal recovery using our three proposed estimators (algebraic, eigenvector, and least-squares) for some example signals corresponding to the $2\times 2$ Fermi-Hubbard Model, $2\times 3$ Fermi-Hubbard Model, and $4 \times 3$ transverse-field Ising Model. 
For recovering each of the above signals, we use $3000$ shots for estimating each entry of the $6$-band of the corresponding lifted matrix $Z$.}
    \label{fig:FH-example-reconstruction-high-shots}
\end{figure}

We observe that the algebraic and eigenvector estimators perform more reliably than the least-squares estimator. 
We also observed that the computational cost of the least-squares is significantly higher than the other two estimators.
This is due to the fact that the least-squares estimator solves an optimization problem over $T \times T$ matrices, whereas the other two estimators operate on small local blocks and scale linearly in $T$. 
Due to the weaker performance of the least-squares estimator, from now on, we will only focus on the block-by-block algebraic and eigenvector estimators.

\subsection{The Length of Consecutive Zeros Constrains the Matrix Bandwidth}

To demonstrate Lemma \ref{lem:identifiability}, stating that a signal $f$ with $W$ consecutive zeros  
is uniquely determined by a $K$-band satisfying $K\geq W+1$ of its lifted matrix, we numerically study the signal defined in the statement of Lemma \ref{lem:integer-H}. 
We choose this signal to have easy control of the parameter $W$.
We fix the signal length to $T=20$ and generate signals with $W = 2, \ldots, T-1$.
For each signal, we perform recovery using measurements of the $K$-band of the lifted matrix, with $K$ ranging from $2$ to $T-1$. 
For each pair $(W, K)$, we perform $45$ independent trials and compute the average normalized Euclidean error (reconstruction error) for each estimator.

\begin{figure}[h!]
\begin{subfigure}[h]{\textwidth}
   \centering
\includegraphics[width=0.7\linewidth]{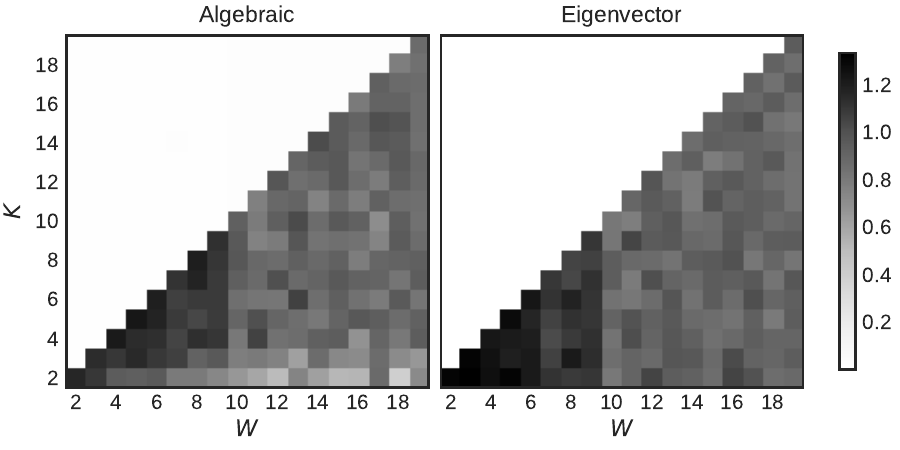}
\caption{$10^8$ shots per measurement.}
\end{subfigure}
\begin{subfigure}[h]{\textwidth}
    \centering
\includegraphics[width=0.7\linewidth]{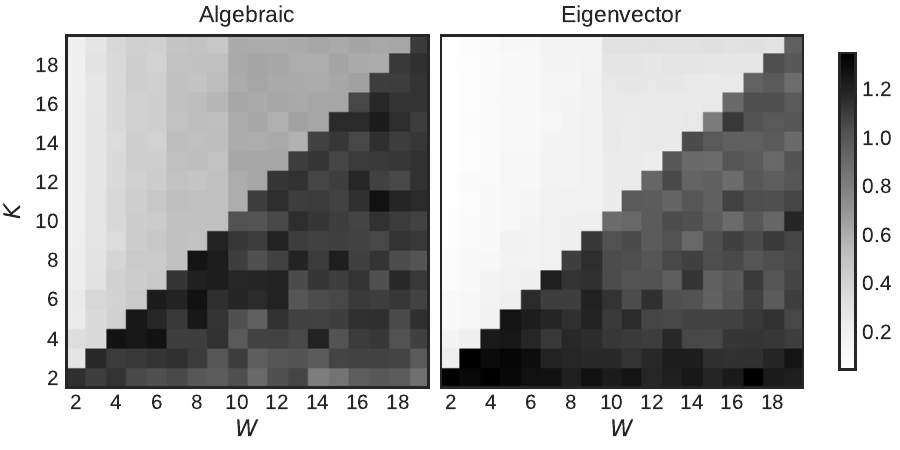}
\caption{$100$ shots per measurement.}
\end{subfigure}
\caption{  Signal recovery for the signal of Lemma \ref{lem:integer-H} using the algebraic and eigenvector estimators. 
We plot the average normalized Euclidean error in the recovered signal in grayscale for all values of $K= 2, \ldots T-1$ and $W= 2, \ldots, T-1$}
    \label{fig:Ksig-vs-Kmeas-example}
\end{figure}
The results in Figure~\ref{fig:Ksig-vs-Kmeas-example} using grayscale show two cases, one where the number of shots for measuring each entry of the $K$-band of $Z$ is $100$, and one where the number of shots is $10^{8}$ (effectively noiseless). 
At each noise level, the relationship $K\geq W+1$ imposed by Lemma~\ref{lem:identifiability} is clearly visible and remains significantly robust against shot noise.

\subsection{Recovery Performance as a Function of the Total Number of Shots}

In this subsection, we study how the reconstruction error scales with the total number of measurement shots and how this scaling depends on the relationship between the measurement bandwidth $K$ and the intrinsic structure of the signal.

Throughout this analysis, the total number of shots is distributed uniformly across all measured entries of the $K$-band of the lifted matrix $Z$.
As established in Lemma~\ref{lem:identifiability}, accurate recovery is only possible when the measurement bandwidth satisfies $K \geq W+1$, where $W$ denotes the length of the largest run of consecutive zero entries in the signal.

In numerical simulations, however, exact zeros are rarely observed due to shot noise.
We mentioned this before in Section~\ref{sec:adaptive-bandwidth-selection} as well.
To obtain a more realistic notion of intrinsic width, we therefore define $W$ operationally using a fixed noise threshold.
Given a threshold $\chi>0$, we say that a signal has intrinsic width $W$ if it contains a run of $W$ consecutive entries whose absolute value is smaller than $\chi$.
Throughout this section, we take $\chi=0.1$.

Using this definition, we group signals according to their estimated intrinsic width $W$ and analyze the reconstruction error within each group. We generate random signals for each dataset by randomly sampling the parameters of each Hamiltonian from the intervals previously specified. We sort them into groups labeled by $W$ according to whether they contain a run of $W$ consecutive entries whose absolute value is smaller than $\chi = 0.1$. The majority of these signals have $W=0$ at $\chi=0.1$. We generated random signals until we had $10$ random signals with $W=24$ to compare against $10$ random signals with $W=0$.

Figure~\ref{fig:Ksig-vs-Kmeas-realistic-FH22} illustrates this behavior for the $2\times 2$ Fermi-Hubbard model with a coherent Gibbs input state ($\beta = 0.5$) and signal length $T=50$.
We focus on two representative families of signals, characterized by intrinsic widths $W=0$ and $W=24$, respectively. 
The estimator is run on $5$ noisy versions of every signal in each $W$ family at a given shot budget. 
For each signal in each families ($W=0$ and $W=24$) we generate $5$ noisy versions using Bernoulli sampling from the expectation value, and we compute its normalized Euclidean reconstruction error.
Finally, for each family we compute the mean, maximum and minimum of these errors, with the mean reconstruction error shown in a solid line and the minimum and maximum indicated by the shaded region in Figure~\ref{fig:Ksig-vs-Kmeas-realistic-FH22}.

\begin{figure}[t]
\begin{subfigure}[h]{\textwidth}
\includegraphics[width=\linewidth]{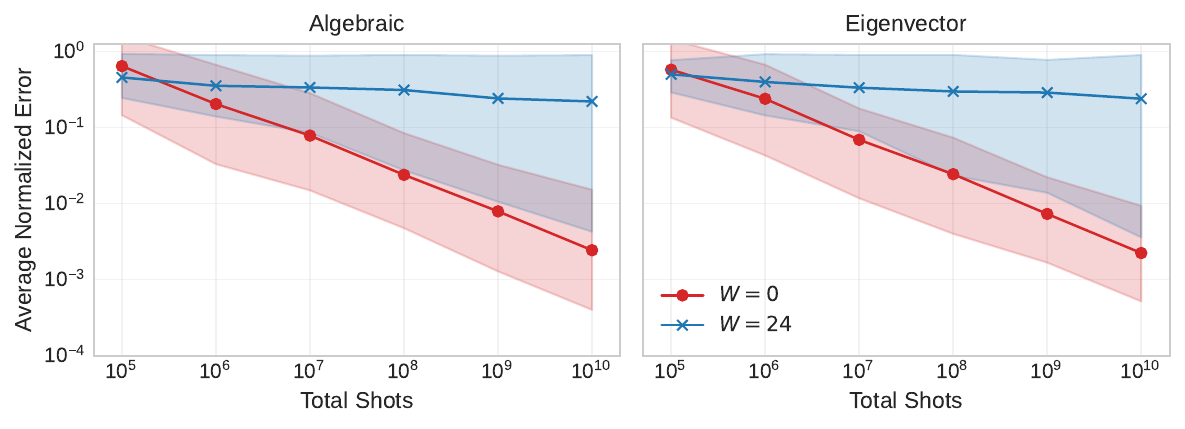}
\caption{With $K=2$ for all groups of signals. }
\end{subfigure}
\begin{subfigure}[h]{\textwidth}
\includegraphics[width=\linewidth]{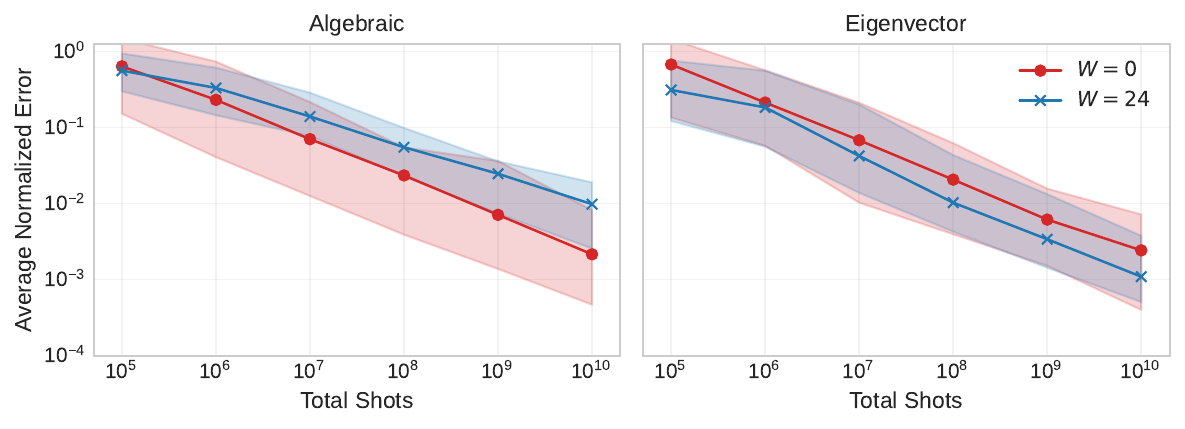}
\caption{With $K=W+2$ for each group of signals.}
\end{subfigure}
\caption{ \textbf{$2\times 2$ Fermi-Hubbard Model.} For signals with a coherent Gibbs state with $\beta=0.5$ as input, $T=50$ and randomly chosen Hamiltonian parameters. The signals are grouped into two sets of ten random signals labeled by $W =0$ and $W=24$. This grouping is done according to whether they contain either $0$ or $24$ consecutive entries below a noise threshold of $0.1$. 
}
    \label{fig:Ksig-vs-Kmeas-realistic-FH22}
\end{figure}

Figure~\ref{fig:Ksig-vs-Kmeas-realistic-FH22}(a) corresponds to a fixed measurement bandwidth $K=2$.
In this setting, signals with $W=0$ satisfy the identifiability condition $K \geq W+1$, and their reconstruction error decreases steadily as the total number of shots increases.
In contrast, signals with $W=24$ violate this condition, and their reconstruction error exhibits a pronounced plateau, indicating that increasing the number of shots cannot compensate for insufficient band information.

In Figure~\ref{fig:Ksig-vs-Kmeas-realistic-FH22}(b), we instead choose the measurement bandwidth adaptively as $K=W+2$.
This guarantees that the identifiability condition is satisfied for both signal families.
In this regime, the reconstruction error decreases with the total number of shots for all signals, confirming that once $K$ is chosen appropriately relative to $W$, any additional measurements directly translate into improved reconstruction accuracy.

It is important to point out that the shadowed area corresponds to the value between the maximum and minimum inside a given family and not the estimation error of a given single signal, which is much smaller, as one can see in Figure~\ref{fig:flat-to-decay}. 
The larger deviation observed for larger deviation for large shot budgets result from some signals inside the family becoming recoverable beyond some large shot budget threshold due to the following phenomena. 
We note that as the shot noise decreases, due to an increase in the shot budget, there is an associated effective width $W^*\leq W$, as previously suppressed time-series entries may rise above the noise threshold.
Once the condition $W\geq K\geq W^*+1$ is satisfied, we observe a transition from a plateau regime, where recovery is impossible, to a regime where the error begins to decay with the number of shots.
We discuss this effect later in Appendix~\ref{app:extranumerics} in Figure~\ref{fig:flat-to-decay}.
Therefore, the wide variation inside a family result from signal having the same $W$ for a threshold of noise of $0.1$ value but having a set a very different effective $W^*$ as the shot noise decreases.

Finally, we observe the same qualitative behavior for the other datasets considered in this work; corresponding results are reported in Appendix~\ref{app:extranumerics}.
These results demonstrate that the scaling of the reconstruction error with the total number of shots is governed primarily by the relationship between the measurement bandwidth $K$ and the effective intrinsic width $W$ of the signal, rather than by the shot count alone.

\subsection{Fixed Shots Budget Analysis}

In realistic quantum experiments, the total number of available measurement shots is limited and is often fixed. 
This means that increasing the number of matrix elements we choose to estimate necessarily reduces the number of shots allocated to each individual measurement. 
In this subsection, we study this experimentally relevant trade-off by fixing the total number of shots and varying the bandwidth $K$ used for reconstruction. 
Specifically, we fix the total number of shots to approximately one million and distribute them evenly across all measured entries of the $K$-band of the lifted matrix $Z$. 
We restrict attention to signals that do not contain two consecutive entries below the noise threshold of $0.1$, corresponding to an effective intrinsic width $W=1$ under the operational definition introduced earlier.
For each of the three datasets, we filter such signals and then evaluate the average normalized Euclidean reconstruction error for the algebraic and eigenvector estimators.

\begin{figure}[!ht]
    \centering
    \includegraphics[width=0.9\linewidth]{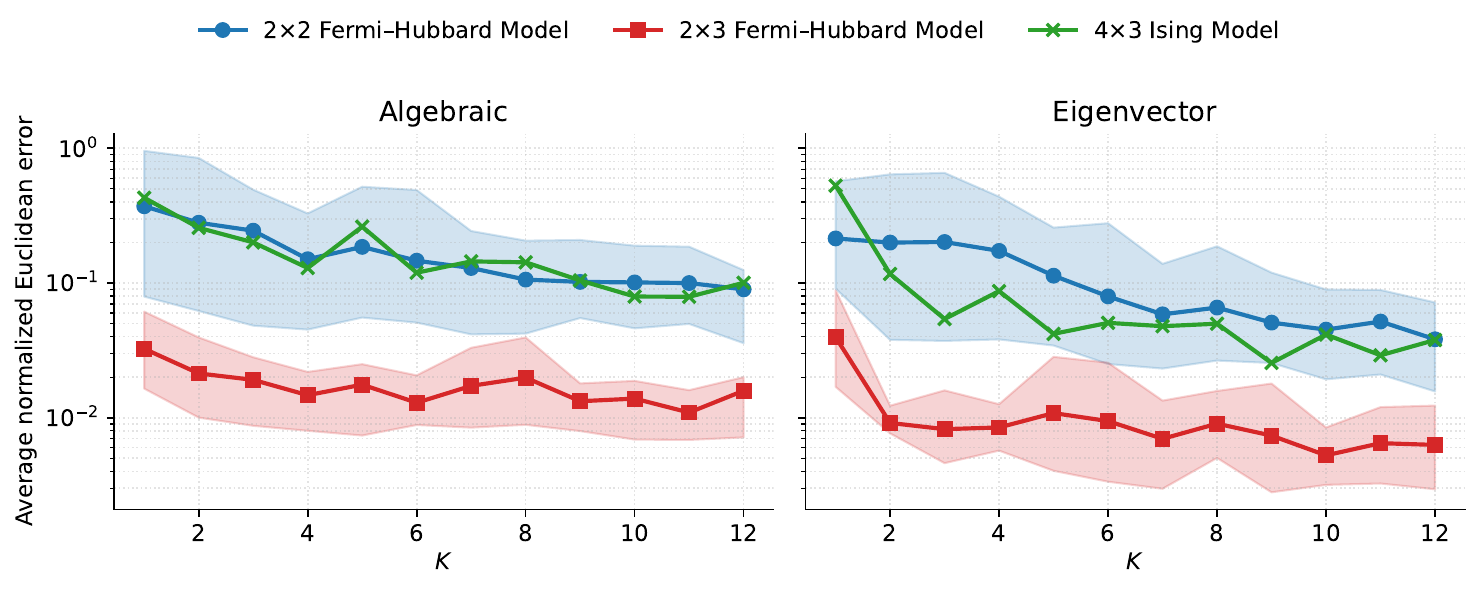}
    \caption{\textbf{Fixed Shots Budget Analysis.} Average normalized Euclidean reconstruction error of the block-by-block algebraic (left) and block-by-block eigenvector (right) estimators as a function of the measurement bandwidth $K$, with the total number of measurement shots fixed to approximately one million and distributed uniformly across the measured entries of the $K$-band of the lifted matrix $Z$.
Results are averaged over randomly generated signals with intrinsic width $W=1$ (defined operationally with noise threshold $\chi=0.1$) from the $2\times2$ and $2\times3$ Fermi-Hubbard models and the $4\times3$ transverse-field Ising model.
Solid lines denote the mean reconstruction error across filtered signals for each dataset, while shaded regions indicate the range between the minimum and maximum observed errors across signals.
No shaded region is shown for the transverse-field Ising model, since our corresponding dataset contains only one signal with $W=1$.}
    \label{fig:constant-shots}
\end{figure}

The results, shown in Figure~\ref{fig:constant-shots}, illustrate a clear trade-off. 
Increasing $K$ initially improves reconstruction accuracy because it provides more information about the signal.
However, as $K$ increases further, the fixed total shot budget must be divided among a growing number of matrix entries.
This reduces the number of shots per entry, increases statistical noise in each measurement, and eventually limits the achievable reconstruction accuracy.
Beyond a certain point, this increased noise outweighs the benefit of measuring additional matrix elements, and the reconstruction error begins to plateau.

The quality of the reconstruction achieved under a fixed shot budget is particularly strong. 
As discussed before in Figure~\ref{fig:constant-shots-summary} in Section~\ref{sec:summary}, we observe that even for long signals with length $T=150$, the recovered time-series closely matches the exact signal across the entire time window, despite operating under a strict shot budget of approximately one million total shots. 
This demonstrates that the proposed estimators, especially the eigenvector estimator, remain effective well beyond the moderate signal lengths used in the above numerical analysis and that their performance does not degrade noticeably as the time-series length increases.
In addition, we include a second representative example (Figure~\ref{fig:fix-shots-budget-analysis-example-signal}) in the appendix showing the reconstruction of a typical signal under the same fixed-shot-budget regime.

We emphasize that increasing the signal length $T$ further introduces an additional constraint: even if the measurement bandwidth $K$ is held fixed, the total number of measured matrix elements grows with $T$, reducing the number of shots per entry under a fixed budget.
This effect naturally leads to a plateau in reconstruction accuracy as $T$ increases, analogous to the behavior observed when increasing $K$. 

On near-term quantum devices where controlled operations are expensive and noisy, increasing $K$ under the fixed shot budget constraint will, for large enough $K$, eventually lead to \textit{worse} accuracy rather than just diminishing returns. 
A complete understanding of the optimal point $K^\star$ which balances the noise from these extra controlled operations would require a detailed noise model for specific hardware, so we leave this to future work.

\section{From Time-Series Estimation to Spectral Estimation}
\label{sec:time-series-to-spectrum}

Having established robust estimators in Section~\ref{sec:estimators} that recover the time-series $f$ (defined in~\eqref{def:discrete-time-signal}) and produce an estimate $\widehat{f}$, in this section, we now explain how such a time-series estimate can be converted into an estimate of the underlying energy spectrum of $H$ within the subspace actually probed by the input state $|\psi\rangle$.
We also analyze how the error in the time series estimation propagates into the error in estimating the eigenvalues of $H$ and clarify the roles played by the sampling step $\Delta$, the number of samples or distinct time points, $T$, and the spectral norm of the Hamiltonian.

Consider that $H$ has the following spectral decomposition:
\begin{equation}\label{eq:H-spectral-decomposition}
    H = \sum_{k=1}^{d} E_k \ket{E_k}\!\bra{E_k},
\end{equation}
where each $\ket{E_k} \in \mathbb{C}^d$ is an eigenvector with the corresponding eigenvalue $E_k \in \mathbb{R}$.
Assume that the input state $|\psi\rangle$ has support on exactly $r$ eigenvectors, i.e.,
\begin{equation}\label{eq:psi-decomposition}
    |\psi\rangle = \sum_{k=1}^{r} c_k |E_k\rangle,
    \qquad c_k \neq 0 \text{ and } \sum_{k=1}^r |c_k|^2 = 1.
\end{equation}
Here, $r$ is left arbitrary and may range from $1$ to $d$, and we introduce it simply to keep the analysis more general.
That being said, the signal takes the following form:
\begin{equation}
    f(t) = \sum_{k=1}^r |c_k|^2\,\mathrm{e}^{-\mathrm{i}E_k t}.
    \label{eq:signal-energy-mixture}
\end{equation}
Note that only the $r$ eigenvalues $\{E_k\}_{k=1}^r$ that appear in~\eqref{eq:signal-energy-mixture} are identifiable from $f(t)$.
Therefore, all of the analysis in this section concerns the restriction of $H$ to the $r$-dimensional subspace
\begin{equation}
    \mathcal{H}_{\psi} \coloneqq \mathrm{span}\{|E_1\rangle, \ldots, |E_r\rangle\},
\end{equation}
since eigenvectors orthogonal to $|\psi\rangle$ contribute zero weight and are invisible to the time-series.

Sampling $f(t)$ on the grid $t_n = n\Delta$ gives
\begin{equation}
    f_n = f(n\Delta) = \sum_{k=1}^r |c_k|^2\,\mathrm{e}^{-\mathrm{i}E_k n\Delta},
    \qquad n=0,\ldots,T-1.
    \label{eq:discrete-signal-energy}
\end{equation}
Discrete-time sampling induces an intrinsic ambiguity in the energies.
Indeed, for any integer $m\in\mathbb{Z}$ and any $n$,
\begin{equation}
    \mathrm{e}^{-\mathrm{i}(E + \frac{2\pi m}{\Delta})\, n\Delta}
    = \mathrm{e}^{-\mathrm{i}E n\Delta}\,\mathrm{e}^{-\mathrm{i}2\pi mn}
    = \mathrm{e}^{-\mathrm{i}E n\Delta}.
\end{equation}
Therefore, from the discrete samples $\{f_n\}_{n=0}^{T-1}$ alone, an energy $E$ cannot be distinguished from
$E + 2\pi m/\Delta$.
To make the energies uniquely identifiable (within $\mathsf{S}_\psi$), we assume a known energy window and choose $\Delta$ small enough so that this window fits strictly inside one aliasing period.
To be precise, we assume that there exists a known $E_{\max}>0$ such that the $r$ energies visible to $|\psi\rangle$ satisfy
\begin{equation}\label{eq:energy-upper-bound}
    |E_k| \le E_{\max}
    \qquad \text{for all } k=1,\ldots,r,
\end{equation}
and the sampling step $\Delta$ is chosen so that
\begin{equation}
    \Delta < \frac{\pi}{E_{\max}}.
    \label{eq:aliasing-avoidance}
\end{equation}
This is exactly the Nyquist--Shannon sampling rate. 
In practice, one can interpret this as a simple rescaling rule: if one has an a priori bound $E_{\max}$ on the energies that $|\psi\rangle$ can overlap with, then one chooses $\Delta$ so that it satisfies the above condition.

However, even with a proper sampling rate, extracting energies from $\widehat{f}$ is non-trivial due to two error sources:
\begin{enumerate}
    \item \textbf{Spectral Leakage:} Since we sample for a finite duration $T$, this sets the best achievable resolution. So if two energies are extremely close, then over short times they remain nearly indistinguishable. 
    Due to this, it is important to assume that there exists a nonzero separation between energies restricted to $\mathcal{H}_{\psi}$ (energies outside $\mathcal{H}_{\psi}$ are irrelevant as discussed above):
    \begin{equation}\label{def:energy-gap}
        \gamma \coloneqq \min_{j\neq k} | E_j - E_k | > 0.
    \end{equation}
    \item \textbf{Time-Series Estimation Error:} Since we are reconstructing the spectrum from a time-series estimate $\widehat{f}$ which deviates from the true time-series by $\| \widehat{f} - f\|_2 \leq \varepsilon$, this error propagates in the estimation of the spectrum.
\end{enumerate}

The following theorem now establishes the sufficient conditions on the total simulation time $T$, the estimation error $\varepsilon$, and other parameters, to guarantee that the true spectral peaks are distinguishable from noise.

\begin{theorem}[Stability of Spectral Estimation]\label{thm:stability-spectral-estimation}
    Let $f \in \mathbb{C}^T$ be the discrete time-series (defined in~\eqref{def:discrete-time-signal}) induced by a Hamiltonian $H$ with spectral decomposition given by~\eqref{eq:H-spectral-decomposition} and the input state $|\psi\rangle$ defined in~\eqref{eq:psi-decomposition}. Furthermore, let $\widehat{f}$ be a time-series estimate such that $\|\widehat f-f\|_2\le \varepsilon$ for $\varepsilon>0$. Suppose that the $r$ energies visible to $|\psi\rangle$ satisfy~\eqref{eq:energy-upper-bound} and are separated by a gap $\gamma$ as defined in~\eqref{def:energy-gap}. Also assume that the sampling step $\Delta > 0$ satisfies the Nyquist--Shannon sampling rate as given by~\eqref{eq:aliasing-avoidance}. 

    If the total number of time steps $T$ is sufficiently large, such that:
    \begin{enumerate}
        \item \textbf{Resolution Condition:} $T \geq 2\pi/(\Delta\gamma) $, and
        \item \textbf{Signal-to-Noise Condition:}
        \begin{equation}
            \frac{2\min_{k=1}^r |c_k|^2}{\pi} \sqrt{T} > \frac{4 \pi}{\gamma \Delta \sqrt{T}} + 2\varepsilon.
        \end{equation}
    \end{enumerate}
    Then for every eigenvalue $E_k$, the discrete Fourier Transform of the estimate $\widehat{f}$ exhibits a local maximum at the grid index $\ell_k$ such that the eigenvalue estimate $\widehat{E}_k = \sfrac{2\pi \ell_k}{T \Delta}$ satisfies:
    \begin{equation}
        \left | \widehat{E}_k - E_k \right | \leq \frac{\pi}{T \Delta}.
    \end{equation}
    Consequently, to achieve a target precision $\delta > 0$, it suffices to choose $T \geq \sfrac{2\pi}{\delta \Delta}$.
\end{theorem}
\begin{proof}
    The proof is deferred to Appendix~\ref{app:spectral-estimation-proof}.
\end{proof}

\paragraph{Binning of Nearly Degenerate Energies.}
The resolution condition in Theorem~\ref{thm:stability-spectral-estimation} makes explicit a fundamental limitation of time-series-based spectral estimation: resolving two eigenvalues separated by a small gap $\gamma$ requires total simulation time $T$ that scales as $1/(\Delta \gamma)$. 
In many physically relevant Hamiltonians, particularly in many-body systems, the spectrum may contain clusters of very closely spaced energies. In such a situation, enforcing strict extraction of every individual eigenvalue may be both unnecessary and impractical, as the required total simulation time $T$ becomes prohibitively large. 
A natural alternative is to group nearby eigenvalues into spectral bins and aim to estimate their aggregate spectral weight rather than resolving each eigenvalue individually. 
From the perspective of the time-series $f$ of length $T$ sampled at rate $\Delta$, this corresponds to treating eigenvalues within a resolution window of width $O(1/(T\Delta))$ as effectively indistinguishable. 
One can readily extend the analysis of Theorem~\ref{thm:stability-spectral-estimation} to this setting: by relaxing the gap condition and allowing eigenvalues within a bin to be merged, one trades fine spectral resolution for reduced total simulation time and, with that, increased robustness to noise. 
In practice, one can choose the bin width as a function of the target accuracy and noise level,  yielding a controlled coarse-grained spectral estimate that is often sufficient for physical applications.

\section{Conclusion and Open Problems}
\label{sec:conclusion}

In this paper, we introduced Quantum Phaselift, a lifting-based framework that enables the estimation of quantum time-series without the prohibitive circuit depths required by standard methods.
By estimating the entries of the rank-one matrix $Z = ff^{\dagger}$ rather than directly estimating the signal entries, our approach requires controlled time evolution over short times.
We demonstrated that measuring a narrow band of size $K$ of this matrix is information-theoretically sufficient for unique signal reconstruction. 

This provides a substantial resource advantage: the depth of the controlled time evolution circuits scales only with $K$, which is $O(1)$ for generic signals, rather than with the total evolution time, effectively decoupling the total evolution time from the control depth. 
We showed that in a noiseless scenario, the bandwidth $K$ is determined by the largest run of consecutive zeros in the time-series vector $f$. 
In the realistic scenario of shot noise induced by the access to a finite number of samples, $K$ is determined by the largest segment of $f$ with values under a pre-selected noise threshold.
We also demonstrated that choosing a band with a slightly larger width, increased by a small constant, can have a significant improvement in the required number of shots to achieve the same quality of signal recovery.

Compared to standard phase-estimation circuits, the only additional requirement of our approach is the need to undo the uncontrolled state preparation and measure the system register.
While this introduces a small overhead in circuit depth, it enables the use of error-mitigation techniques similar to those employed in verified phase estimation~\cite{OBrien2021}, which can substantially improve robustness on near-term hardware.

We introduced three classical post-processing algorithms to reconstruct the underlying signal from the noisy observations of this narrow $K$-band of $Z$: the block-by-block algebraic, block-by-block eigenvector, and least-squares estimators. 
We rigorously proved that all three estimators guarantee exact recovery in the noiseless setting and derived stability bounds for the block-by-block algebraic estimator in the case when the signal has no zero entries and for the least-square estimator in terms of the smallest conic singular value of the associated $K$-banded measurement map, leaving its explicit characterization and the resulting query-complexity bounds to future work.
We observed that among the three estimators, the block-by-block eigenvector performs most efficiently in our numerical experiments.

We numerically demonstrated that one could recover with high quality the time-series generated by Fermi-Hubbard and transverse-field Ising Hamiltonians.
In particular, we demonstrated accurate reconstruction of signals with more than $100$ time points using only a few million measurement shots and modest classical post-processing time.
These results indicate that the proposed framework is well suited for near-term experimental implementations.

In the current analysis, we have considered a banded matrix, with the same bandwidth $K$ on every row. 
In practice, one could design strategies that optimize the selection of matrix elements $Z_{ij}$ to estimate, in order to preserve the quality of the recovery while reducing the shot counts.

More broadly, the lifting framework introduced here is not limited to time-evolution unitaries.
It applies equally to estimating expectation values of the form $\langle \psi | U | \psi \rangle$ for general unitary circuits.
Given a decomposition $U = U_N U_{N-1} \cdots U_1$, one may define a sequence $f(i) = \langle \psi | U_i \cdots U_1 | \psi \rangle$ and apply Quantum Phaselift to estimate the final value $f(N)$.
As in the time-evolution setting, this approach requires controlling only a constant number of circuit layers, substantially reducing implementation cost.

The above idea of \textit{slicing} $U$ into simpler unitaries could also be applied to the (Trotterized) circuit implementing the controlled time-evolution of $\mathrm{e}^{\mathrm{i}HK\Delta}$ in order to reduce further the depth of its controlled component at the cost of increasing further the number of required shots by a factor proportional to the amount of additional circuit slices, as we now must solve a longer time-series vector.

This last suggestion has strong similarities with the recent result by Schiffer et al.~\cite{Schiffer2025}, named Sequential Hadamard Test, where they decompose a Trotterized time-dynamics simulation into circuit slice composed of single gates. 
Their reconstruction algorithm is based on a telescoping product procedure, which shares conceptual similarities with our algebraic estimator when specialized to the case $K = 1$. 
Nonetheless, we expect that the reconstruction techniques developed in our work can be used to improve the robustness and efficiency of the classical post-processing in the Sequential Hadamard Test. 
A fundamental limitation of the Sequential Hadamard Test is that the total number of shots scales with the number of gates in the circuit, in addition to other problem parameters,  which can make practical implementations challenging. 
Again, this can be mitigated by allowing for slices of a larger number of layers of gates, optimizing the trade-off between the reduction of the depth of the controlled circuit and the additional sample cost. 

Finally, the lifting framework presented here admits natural extensions beyond simple time-correlation functions as in~\eqref{eq:L-amplitude}.
It may be adapted to estimate more general observables of the form $\langle \psi | O \mathrm{e}^{-\mathrm{i} H t} | \psi \rangle$, as well as other complex quantities arising in Krylov and subspace methods~\cite{yoshiokaKrylovDiagonalizationLarge2025,yuQuantumCentricAlgorithmSampleBased2025,mottaSubspaceMethodsElectronic2024}.
It also connects naturally to modern spectral-estimation techniques such as dynamic mode decomposition and ESPRIT-style methods~\cite{shenEfficientMeasurementDrivenEigenenergy2024,shenEstimatingEigenenergiesQuantum2025,dingESPRITAlgorithmHigh2024,stroeksSpectralEstimationHamiltonians2022}, and may find applications in the estimation of out-of-time-order correlators~\cite{Abanin2025}.
Exploring these directions is an interesting avenue for future work.

\section{Acknowledgments}
STF and DJP were partially funded by the U.S.\ Department of Energy, Office of Science, National Quantum Information Science Research Centers, Co-design Center for Quantum Advantage under contract number DE-SC0012704 and by the National Science Foundation National Quantum Virtual Laboratory ERASE grant number OSI-2435244.

\appendix

\section{SWAP measurement}

\subsection{Estimating the Diagonal Elements}\label{app:estimating-diagonal-ele-swap}

As mentioned in the main text, for certain hardware architectures, particularly those where deep circuits or long coherence times are prohibitive, it can be more practical to implement a SWAP test between a copy of $\rho_\psi$ and one of $\rho_\psi(t)$. 
In this case, the overlap can be obtained as
\begin{equation}
\operatorname{Tr}\!\left[\operatorname{SWAP}(\rho_\psi\otimes\rho_\psi(t))\right]=\operatorname{Tr}\!\left[\rho_\psi\rho_\psi(t)\right]
\end{equation} 
at the expense of doubling the number of qubits, as shown below:
\begin{center}
    \begin{quantikz}
    \lstick{$\ket{\boldsymbol{0}}$} &\gate{U_\psi} &\gate{\mathrm{e}^{-\mathrm{i}Ht}} & \meter[2]{z}\\
    \lstick{$\ket{\boldsymbol{0}}$} &\gate{U_\psi} &\qw & 
 \end{quantikz}
\end{center}
Here, $z$ is a single bit and is the outcome of jointly measuring both the registers in the eigenbasis of the SWAP operator, whose eigenvalues are $+1$ and $-1$, corresponding respectively to the symmetric and antisymmetric subspaces. 
We next show this in more detail. Let the two registers contain $n$ qubits each. This measurement can be efficiently implemented without an explicit SWAP or controlled-SWAP gate, by performing a layer of Bell-basis measurements across the two registers. 

Concretely, one first applies a layer of $n$ Hadamard gates to all qubits in the first register, followed by a layer of pairwise CNOT gates coupling each qubit $i$ of the first register to the corresponding qubit $i$ of the second register. 
Each qubit is then measured in the computational basis, producing bitstrings
\begin{equation}
    \boldsymbol{x} = x_1 \ldots x_n \quad \text{and} \quad \boldsymbol{y} = y_1 \ldots y_n
\end{equation}
for the two registers, respectively. 
The joint measurement outcome corresponding to the SWAP eigenvalue is obtained by computing the overall parity of the bitwise logical ANDs between corresponding qubits:
\begin{equation}\label{eq:z-parity}
z=\bigoplus_{i=1}^n x_i\cdot y_i,
\end{equation}
where $\cdot$ denotes the logical AND and $\oplus$ denotes XOR (addition modulo $2$). 
The parity bit $z$ thus determines whether the joint two-register state lies in the symmetric $(z=0)$ or antisymmetric $(z=1)$ subspace. 
Averaging $(-1)^z$ over $N$ independent runs of the circuit gives an unbiased estimator for the expectation value of the SWAP operator with respect to the joint state $\rho_\psi\otimes\rho_\psi(t)$,
\begin{equation}
    \langle (-1)^z \rangle = \operatorname{Tr}\!\left[\operatorname{SWAP}(\rho_\psi\otimes\rho_\psi(t))\right],
\end{equation}
which is what we want. 
In other words, an unbiased estimator $\widehat Z_{ii}$ for the $i$-th diagonal element $Z_{ii}$ is given by
\begin{equation}
    \widehat Z_{ii} =  \frac{1}{N}\sum_{i=1}^N (-1)^{z_i},
\end{equation}
where $z_i$ is the measurement outcome of the $i$-th run.

\subsection{Estimating the Off-diagonal Elements}\label{app:estimating-off-diagonal-ele-swap}

One can use the relation $\text{Tr}\left[\operatorname{SWAP}\left(A\otimes B\right)\right]=\text{Tr}\left[AB\right]$ and write
\begin{equation}
    Z_{ij} =
\operatorname{Tr}\!\left[\operatorname{SWAP} \left (\mathrm{e}^{\mathrm{i}H(t_j-t_i)}\otimes I \right)(\rho_\psi(t_i)\otimes\rho_\psi)\right].
\end{equation}
This implies that the same quantity, that is, $\operatorname{Re}[Z_{ij}]$ can be obtained using a shallower circuit that trades circuit depth for an increased qubit count:
\begin{center}
    \begin{quantikz}
    \lstick{$\ket{0}$}  & \gate{\text{Had}}  &\qw        &\ctrl{1}       & \qw                & \gate{\text{Had}}  &\meter{x} \\
     \lstick{$\ket{\boldsymbol{0}}$} &\gate{U_\psi} & \gate{\mathrm{e}^{-\mathrm{i}Ht_i}} & \gate{\mathrm{e}^{\mathrm{i}H(t_j-t_i)}} & \meter[2]{z}                    \\
     \lstick{$\ket{\boldsymbol{0}}$} &\gate{U_\psi} \qw & \qw & \qw   &\qw  
    \end{quantikz}
\end{center}
Here, the control qubit is measured in the computational basis, yielding an outcome $x\in\{0,1\}$, while the two lower registers are jointly measured in the Bell basis, yielding the parity bit $z\in\{0,1\}$ as previously defined in~\eqref{eq:z-parity}.  
Averaging $(-1)^{x + z}$ over $N$ independent runs of the above circuit gives an unbiased estimator for $\operatorname{Re}[Z_{ij}]$,
\begin{equation}
\widehat{\operatorname{Re}[Z_{ij}]} = \frac{1}{N}\sum_{i=1}^N (-1)^{x_i + z_i},
\end{equation}
where $(x_i, z_i)$ is the measurement outcome of the $i$-th run. This can be easily seen from what follows.

Using the shorthand notations as before, that is, $U \equiv \mathrm{e}^{\mathrm{i}H(t_j-t_i)}$, $U_{\phi} \equiv \mathrm{e}^{-\mathrm{i}Ht_i} U_{\psi}$, and $|\phi\rangle \equiv U_{\phi} |\boldsymbol{0}\rangle$, the joint state after the controlled-$U$ operation is
\begin{equation}
    \ket{\eta} = \frac{1}{\sqrt{2}} \left ( \ket{0} \otimes \ket{\phi} \otimes \ket{\psi} +  \ket{0} \otimes U \ket{\phi} \otimes \ket{\psi} \right)
\end{equation}
Measuring the control bit in the $X$ basis (via the final Hadamard gate and the computational basis measurement) and the two lower registers in the Bell basis, that is, the eigenbasis of the SWAP operator, corresponds to measuring the observable $X \otimes \operatorname{SWAP}$. 
The eigenvalue of this observable for a single run of the circuit is $(-1)^{x+z}$. 
The expected value is then
\begin{align}
    \langle (-1)^{x+z} \rangle & = \operatorname{Tr}\!\left [ (X \otimes \operatorname{SWAP}) \ket{\eta}\! \bra{\eta}\right]\\
    & = \frac{1}{2}\operatorname{Tr}\!\left [  \operatorname{SWAP} \left (U \ket{\phi}\! \bra{\phi} \otimes \ket{\psi}\! \bra{\psi} +  \ket{\phi}\! \bra{\phi} U^{\dagger} \otimes \ket{\psi}\! \bra{\psi}\right) \right]\\
    & = \operatorname{Re}\!\left[\operatorname{Tr}\!\left [  \operatorname{SWAP} \left (U \ket{\phi}\! \bra{\phi} \otimes \ket{\psi}\! \bra{\psi}\right) \right]\right]\\
    & = \operatorname{Re}\!\left[Z_{ij}\right].
\end{align}

\section{Proof for Lemma~\ref{lem:integer-H}}\label{app:integer-H-proof}

The operator $P_j = (I - Z_j)/2$ is a projector onto the state $\ket{1}$ for the $j$-th qubit (eigenvalue 1 for $\ket{1}$, 0 for $|0\rangle$). 
For a computational basis state $|x\rangle = |x_{n-1}\dots x_0\rangle$ where $x_j \in \{0,1\}$, the Hamiltonian $H$ acts as:
    \begin{equation}
        H |x\rangle = \sum_{j=0}^{n-1} 2^j P_j |x\rangle = \left( \sum_{j=0}^{n-1} x_j 2^j \right) |x\rangle.
    \end{equation}
    The term in the parentheses is simply the integer value $k$ represented by the binary string $x$. Since $x$ ranges over all binary strings of length $n$, the spectrum of $H$ comprises exactly the integers $E_k = k$ for $k \in \{0, 1, \dots, N-1\}$.

    The signal $f(t)$ is then
    \begin{equation}
        f(t) = \frac{1}{N} \sum_{k=0}^{N-1} \langle k| \mathrm{e}^{-\mathrm{i} H \frac{2\pi}{N} t} |k\rangle = \frac{1}{N} \sum_{k=0}^{N-1} \mathrm{e}^{-\mathrm{i} \frac{2\pi t}{N} k}.
    \end{equation}
    This is a geometric series $\frac{1}{N} \frac{1 - r^N}{1 - r}$ with ratio $r = \mathrm{e}^{-\mathrm{i} \frac{2\pi t}{N}}$. Now there are two cases: 1) If the integer $t$ is a multiple of $N$, then $r=1$ and the sum is $\frac{1}{N}(N) = 1$, and 2)  If the integer $t$ is not a multiple of $N$, then $r \neq 1$. However, for integer values of $t$, we have $r^N = \mathrm{e}^{-\mathrm{i} 2\pi t} = 1$.  Consequently, the numerator $1-r^N$ vanishes, so $f(t) = 0$.

    Setting $N=K+1$, the signal at integer time steps $t=0, 1, \dots$ follows the sequence:
    \begin{equation}
        1, \underbrace{0, 0, \ldots, 0}_{K \text{ times}}, 1, 0, \ldots.
    \end{equation}
This concludes the proof.

\section{Useful Lemmas}
\label{app:useful-lemmas}

In this section, we collect several technical lemmas, along with their proofs, that will be invoked in our main results.

\subsection{Phase Perturbation Lemma}
The following lemma quantifies how perturbation in complex amplitudes translate into perturbation in phase.

\begin{lemma}[Phase Perturbation]\label{lem:phase-diff-z_1-z_2}
    Let $z_1$ and $z_2$ be two complex numbers such that $ \varepsilon < |z_1| \leq 1$, for some $\varepsilon > 0$, $|z_2| \leq 1$, and $|z_1 - z_2| \eqqcolon \delta < \varepsilon$. Then, the following holds:
    \begin{equation}
        |\mathrm{e}^{\mathrm{i}\theta} - 1|\leq \sqrt{2 - 2 \sqrt{1 - \frac{\delta^2}{\varepsilon^2}}} \label{eq:lemma-phase-diff},
    \end{equation}
    where $\theta$ is the angle between $z_1$ and $z_2$.
\end{lemma}
\begin{proof}
    Consider the following:
    \begin{equation}
        |z_1 - z_2| = \sqrt{|z_1|^2 + |z_2|^2 - 2 |z_1||z_2|\cos(\theta)} = \delta
    \end{equation}
    This implies that
    \begin{equation}
        \frac{|z_1|^2 + |z_2|^2 - \delta^2}{2 |z_1||z_2|} = \cos(\theta) \label{eq:cos-z_1-z_2}.
    \end{equation}
    Since we are interested in the upper bound of~\eqref{eq:lemma-phase-diff}, we want to find a complex number $z_2$ that makes the largest possible angle with $z_1$, given that $\varepsilon < |z_1| \leq 1$. 
    This translates to finding the value of $|z_2|$ that minimizes the right-hand side of the above equation, given that $|z_2|$ satisfies the conditions in the lemma statement. 
    To this end, we take the derivative of the right-hand side of the above equation with respect to $|z_2|$ and set it to zero. 
    We get that the following $|z_2|$ minimizes the above equation:

    \begin{equation}
        |z_2| = \sqrt{|z_1|^2 - \delta^2}.
    \end{equation}
    Plugging this back in~\eqref{eq:cos-z_1-z_2}, we get 
    \begin{equation}
        \cos(\theta) = \sqrt{1 - \frac{\delta^2}{|z_1|^2}}.
    \end{equation}

    Finally, using this to bound~\eqref{eq:lemma-phase-diff} from above:
    \begin{equation}
        |\mathrm{e}^{\mathrm{i}\theta} - 1| = \sqrt{2 - 2\cos(\theta)} = \sqrt{2 - 2 \sqrt{1 - \frac{\delta^2}{|z_1|^2}}} \leq \sqrt{2 - 2 \sqrt{1 - \frac{\delta^2}{\varepsilon^2}}}.
    \end{equation}
    This concludes the proof.
\end{proof}

\subsection{Concentration Lemmas}

Next, we present two concentration lemmas that apply to any matrix element whose real and imaginary parts are estimated as empirical averages of independent Bernoulli samples.

\begin{lemma}[Element-Wise Concentration]\label{lem:element-wise}
Let $Z \in \mathbb{C}^{T\times T}$ be a fixed matrix and let $\widehat Z$ be an estimator constructed as follows:
\begin{itemize}
    \item For each diagonal entry $Z_{ii}$, the estimate $\widehat Z_{ii}$ is the empirical mean of $N_{\mathrm{diag}}$ independent Bernoulli samples with mean $Z_{ii}$.
    \item For each off-diagonal entry $Z_{ij}$ with $i \neq j$, the real and imaginary parts $\operatorname{Re}(\widehat Z_{ij})$ and $\operatorname{Im}(\widehat Z_{ij})$ are empirical means of $N_{\mathrm{off}}$ independent Bernoulli samples with means $\operatorname{Re}(Z_{ij})$ and $\operatorname{Im}(Z_{ij})$, respectively.
\end{itemize}
Then, for any $\varepsilon > 0$ and any indices $i,j$,
\begin{itemize}
    \item for each diagonal entry $Z_{ii}$,
    \begin{equation}
        \mathbb{P}\left(|\widehat Z_{ii} - Z_{ii}| \ge \varepsilon\right)
        \le 2 \exp\!\left(-\frac{N_{\mathrm{diag}}\,\varepsilon^2}{2}\right),
    \end{equation}
    \item for each off-diagonal entry $Z_{ij}$ with $i \neq j$,
    \begin{equation}
        \mathbb{P}\left(|\widehat Z_{ij} - Z_{ij}| \ge \varepsilon\right)
        \le 4 \exp\!\left(-\frac{N_{\mathrm{off}}\,\varepsilon^2}{2}\right).
    \end{equation}
\end{itemize}
\end{lemma}

\begin{proof}
For each diagonal entry, $\widehat Z_{ii}$ is the empirical mean of $N_{\mathrm{diag}}$ independent Bernoulli random variables taking values in $[0,1]$ with mean $Z_{ii}$. By Hoeffding’s inequality,
\begin{equation}
    \mathbb{P}\left(|\widehat Z_{ii} - Z_{ii}| \ge \varepsilon\right)
    \le 2 \exp\!\left(-\frac{N_{\mathrm{diag}}\,\varepsilon^2}{2}\right).
\end{equation}

For an off-diagonal entry $Z_{ij}$ with $i \neq j$, the real and imaginary parts are estimated separately as empirical means of $N_{\mathrm{off}}$ independent Bernoulli variables, again taking values in $[0,1]$. 
Applying Hoeffding’s inequality with threshold $\varepsilon/2$ gives
\begin{align}
    \mathbb{P}\!\left(|\operatorname{Re}(\widehat Z_{ij}) - \operatorname{Re}(Z_{ij})| \ge \varepsilon/2\right)
        &\le 2 \exp\!\left(-\frac{N_{\mathrm{off}}\,\varepsilon^2}{2}\right),\\
    \mathbb{P}\!\left(|\operatorname{Im}(\widehat Z_{ij}) - \operatorname{Im}(Z_{ij})| \ge \varepsilon/2\right)
        &\le 2 \exp\!\left(-\frac{N_{\mathrm{off}}\,\varepsilon^2}{2}\right).
\end{align}
By the triangle inequality,
\begin{equation}
    |\widehat Z_{ij} - Z_{ij}|
    \le \left|\operatorname{Re}(\widehat Z_{ij}) - \operatorname{Re}(Z_{ij})\right|
       + \left|\operatorname{Im}(\widehat Z_{ij}) - \operatorname{Im}(Z_{ij})\right|.
\end{equation}
If both real and imaginary parts deviate by at most $\varepsilon/2$, then $|\widehat Z_{ij} - Z_{ij}| \le \varepsilon$. Therefore, by the union bound,
\begin{align}
    \mathbb{P}\left(|\widehat Z_{ij} - Z_{ij}| \ge \varepsilon\right)
    &\le \mathbb{P}\!\left(|\operatorname{Re}(\widehat Z_{ij}) - \operatorname{Re}(Z_{ij})| \ge \varepsilon/2\right) + \mathbb{P}\!\left(|\operatorname{Im}(\widehat Z_{ij}) - \operatorname{Im}(Z_{ij})| \ge \varepsilon/2\right) \\
    &\le 4 \exp\!\left(-\frac{N_{\mathrm{off}}\,\varepsilon^2}{2}\right),
\end{align}
which completes the proof.
\end{proof}

The next lemma upgrades these element-wise bounds to a simultaneous bound over all entries in a given index set.

\begin{lemma}[Max-Element Concentration]\label{lem:max-element}
Let $D \subseteq \{0,\dots,T-1\}$ be a set of diagonal indices and $O \subseteq \{0,\dots,T-1\}^2$ a set of off-diagonal index pairs $(i,j)$ with $i \neq j$. Suppose the estimator $\widehat Z$ satisfies the element-wise bounds of Lemma~\ref{lem:element-wise}. Fix tolerances $\varepsilon > 0$, $\varepsilon' > 0$, and a failure probability $\delta \in (0,1)$. If
\begin{equation}
    N_{\mathrm{diag}}  \ge  \frac{2}{\varepsilon^2} \log\!\left(\frac{4|D|}{\delta}\right)
    \quad\text{and}\quad
    N_{\mathrm{off}}   \ge  \frac{2}{\varepsilon'^2} \log\!\left(\frac{8|O|}{\delta}\right),
\end{equation}
then with probability at least $1 - \delta$,
\begin{equation}
    \max_{i \in D} \left|\widehat Z_{ii} - Z_{ii}\right| \le \varepsilon
    \quad\text{and}\quad
    \max_{(i,j)\in O} \left|\widehat Z_{ij} - Z_{ij}\right| \le \varepsilon'.
\end{equation}
\end{lemma}

\begin{proof}
For each $i \in D$, Lemma~\ref{lem:element-wise} gives
\begin{equation}
    \mathbb{P}\left(|\widehat Z_{ii} - Z_{ii}| \ge \varepsilon\right)
    \le 2 \exp\!\left(-\frac{N_{\mathrm{diag}}\,\varepsilon^2}{2}\right).
\end{equation}
By the union bound,
\begin{equation}
    \mathbb{P}\!\left(\max_{i\in D} |\widehat Z_{ii} - Z_{ii}| \ge \varepsilon\right)
    \le 2|D| \exp\!\left(-\frac{N_{\mathrm{diag}}\,\varepsilon^2}{2}\right).
\end{equation}
To make this at most $\delta/2$, it suffices to require
\begin{equation}
    2|D| \exp\!\left(-\frac{N_{\mathrm{diag}}\,\varepsilon^2}{2}\right) \le \frac{\delta}{2},
\end{equation}
which rearranges to
\begin{equation}
    N_{\mathrm{diag}} \ge \frac{2}{\varepsilon^2} \log\!\left(\frac{4|D|}{\delta}\right).
\end{equation}

Similarly, for each $(i,j)\in O$, Lemma~\ref{lem:element-wise} gives
\begin{equation}
    \mathbb{P}\left(|\widehat Z_{ij} - Z_{ij}| \ge \varepsilon'\right)
    \le 4 \exp\!\left(-\frac{N_{\mathrm{off}}\,\varepsilon'^2}{2}\right).
\end{equation}
Applying the union bound over all $(i,j)\in O$,
\begin{equation}
    \mathbb{P}\!\left(\max_{(i,j)\in O} |\widehat Z_{ij} - Z_{ij}| \ge \varepsilon'\right)
    \le 4|O| \exp\!\left(-\frac{N_{\mathrm{off}}\,\varepsilon'^2}{2}\right).
\end{equation}
To make this at most $\delta/2$, it suffices that
\begin{equation}
    4|O| \exp\!\left(-\frac{N_{\mathrm{off}}\,\varepsilon'^2}{2}\right) \le \frac{\delta}{2},
\end{equation}
which implies
\begin{equation}
    N_{\mathrm{off}} \ge \frac{2}{\varepsilon'^2} \log\!\left(\frac{8|O|}{\delta}\right).
\end{equation}

Finally, another union bound over the diagonal and off-diagonal events shows that both inequalities hold simultaneously with probability at least $1 - \delta$. 
This completes the proof.
\end{proof}

\subsection{PSD Structural Lemma}
The following lemma characterizes when entries of a PSD matrix force linear dependence between its columns.

\begin{lemma}[Singular $2\times 2$ Principal PSD Submatrix Implies Column Dependence]\label{lem:singular-submatrix-column-dependence}
Let $X \succeq 0$ be a positive semidefinite matrix. For any matrix $M$, let $\mathrm{col}_k(M)$ denote its $k$-th column. Fix indices $i \neq j$ and suppose that the $2  \times 2$ principal submatrix of $X$ on $\{i, j\}$ is singular:
\begin{equation}
    X_{ii} X_{jj} = \left | X_{ij}\right|^2.
\end{equation}
Then, whenever $X_{ii} > 0$, the $\mathrm{col}_i(X)$ and $\mathrm{col}_j(X)$ are linearly dependent, with
\begin{equation}
    \mathrm{col}_j(X) = \frac{X_{ij}}{X_{ii}} \mathrm{col}_i(X).
\end{equation}
\end{lemma}

\begin{proof}
    Since $X \succeq 0$, there exists a matrix $Y$ such that $X = Y^{\dagger} Y$ (Cholesky decomposition). Then
    \begin{equation}
        X_{ij} = \langle \text{col}_i(Y), \text{col}_j(Y)  \rangle, \qquad X_{ii} = \left \Vert \text{col}_i(Y)  \right \Vert_2^2, \qquad X_{jj} = \left \Vert \text{col}_j(Y)  \right \Vert_2^2.
    \end{equation}
    Since $X_{ii} X_{jj} = \left | X_{ij}\right|^2$, we have
    \begin{equation}
        \left | \langle \text{col}_i(Y), \text{col}_j(Y)  \rangle\right|^2 = \left \Vert \text{col}_i(Y)  \right \Vert_2^2 \left \Vert \text{col}_j(Y)  \right \Vert_2^2.
    \end{equation}
    This is essentially the Cauchy-Schwarz inequality between vectors $\text{col}_i(Y)$ and $\text{col}_j(Y)$, and since there is an equality, it implies that $\text{col}_i(Y)$ and $\text{col}_j(Y)$ are linearly dependent. Thus let $\text{col}_j(Y) = \alpha \text{col}_i(Y)$ for some scalar $\alpha \in \mathbb{C}$. Applying $Y^{\dagger}$ gives
    \begin{equation}
        \text{col}_j(X) = Y^{\dagger} \text{col}_j(Y) = \alpha Y^{\dagger} \text{col}_i(Y) = \alpha \text{col}_i(X).
    \end{equation}
    Finally, comparing entries, we get
    \begin{equation}
        \alpha = \frac{\langle \text{col}_i(Y), \text{col}_j(Y)  \rangle}{\left \Vert \text{col}_i(Y)  \right \Vert_2^2} = \frac{X_{ij}}{X_{ii}}.
    \end{equation}
    This concludes the proof.
\end{proof}

\section{Proof for Theorem~\ref{thm:exact-rec-algebraic-est}}
\label{app:proof-exact-recovery-algebraic-estimator}

In this section, we prove the exact recovery for the block-by-block algebraic estimator.

We prove ($\Leftarrow$) with induction. Let $Z = f f^{\dagger}$ and let $\widehat Z = \mathsf{band}_K(Z)$ be the input to the algebraic estimator and $\widehat f$ be the reconstructed signal. 
We know that $f_0 = 1$. 
Therefore, for $i=0$, the estimator simply sets $\widehat f_0 = f_0 = 1$. 
Next, we assume that the estimator has already reconstructed the signal exactly for all indices $j < i$, that is, we have $\widehat f_j = f_j$ for all $j < i$. 
If $f_i = 0$, then $Z_{ii} = 0$ so the estimator estimates the magnitude of $f_i$ as $|\widehat f_i| = 0$. 
This implies that $\widehat f_i = 0 = f_i$. 
If $f_i \neq 0$, then from the assumption that $f \in \mathcal{S}_K$, that is, it contains no run of $K$ consecutive zeros, we have that the set of indices $\{i-K, \ldots, i-1\}$ contains at least one index $j$ with $f_j\neq 0$. For any such $j$, the candidate phase estimate $\widehat{\theta}^{\,\,(j)}_i$ (see~\eqref{eq:candidate-phase-estimate}) follows
    \begin{equation}
        \widehat\theta_i^{\,\,(j)}=\arg(\widehat f_j)-\arg(Z_{ji}) =\arg(f_j)-\arg(f_j\overline{f_i})=\arg(f_i),
    \end{equation}
where the induction hypothesis yields $\widehat{f}_j = f_j$ for the second equality above. 
Hence, all candidate phase estimates are equal to $\arg(f_i)$; therefore, when the estimator computes their weighted circular average, as given by~\eqref{eq:weighted-circular-average}, it returns $\arg(\widehat f_i) = \arg(f_i)$. 
Together with magnitude estimation of $f_i$ as $|\widehat f_i|=\sqrt{|\widehat{Z}_{ii}^{\,(K)}|} = \sqrt{Z_{ii}}=|f_i|$, we finally get $\widehat f_i = f_i$.

Next, the proof for $(\Rightarrow)$ is the same as the proof for $(\Rightarrow)$ of Lemma~\ref{lem:identifiability}.

This concludes the proof.

\section{Proof for Theorem~\ref{thm:exact-rec-eigen-est}}
\label{app:proof-exact-recovery-eigenvector-estimator}
In this section, we prove the exact recovery for the block-by-block eigenvector estimator.

We begin by recalling that, in the noiseless case, each $(K+1) \times (K+1)$ block $Q^{(b)}$ possesses a rank-one structure, as shown in~\eqref{eq:Q-b-blocks-noiseless}:
\begin{equation}
    Q^{(b)} = \widehat{Z}_{b:b+K, b:b+K}^{(K)} = Z_{b:b+K, b:b+K} = f_{b:b+K}^{\vphantom{\dagger}} f_{b:b+K}^\dagger.
\end{equation}
So $Q^{(b)}$ is positive semidefinite, rank one, with the principal eigenpair
\begin{equation}
    \left (\lambda_{\max}^{(b)},\, v_{\max}^{(b)} \right),
\end{equation}
where 
\begin{align}
    \lambda_{\max}^{(b)} & = \left \|f_{b:b+K}\right\|_2^2\\
        v_{\max}^{(b)} & = \mathrm{e}^{\mathrm{i}\phi_b} \frac{f_{b:b+K}}{\left \|f_{b:b+K}\right\|_2^2},
\end{align}
for some $\phi_b \in [0, 2\pi)$.
Consequently,
\begin{equation}
    \widehat f^{\,(b)} = \sqrt{\lambda_{\max}^{(b)}}\cdot v_{\max}^{(b)} = \mathrm{e}^{\mathrm{i}\phi_b} f_{b:b+K}\label{eq:local-estimates-noiseless}.
\end{equation}

With this in place, we begin by proving $(\Leftarrow)$ with induction over the indices $b = 0, 1,, \ldots, T-K-1$.
Let us start with the base case $b=0$. 
Recall that the estimator rotates $\widehat f^{\,(0)}$ so that its first entry equals 1; by the above equality, this enforces $\mathrm{e}^{\mathrm{i}\phi_0} = 1$ and hence
\begin{equation}
    \widehat f^{\,(0)} = f_{0:K}.
\end{equation}
This concludes the base case. 
    
Now assume that the estimator has already phase-aligned $b$ local estimates, that is, $\widehat f^{(0)}, \ldots, \widehat f^{(b-1)}$, and has established that 
\begin{equation}\label{eq:inductive-hypo}
    \widehat f^{\,(b')} = f_{b':b'+K}
\end{equation}
holds for all $b' = 0, 1, \ldots, b-1$.

With that, now consider the $(b+1)$-th local estimate $\widehat{f}^{\,(b)}$. 
Recall that the vector $u \in \mathbb{C}^K$ denotes the first $K$ entries of this (unaligned) local estimate and the vector $v\in \mathbb{C}^K$ denotes the consensus vector formed by averaging previously phase-aligned estimates on the overlap indices (see Step~\ref{step:average-prev-aligned-estimates} of Algorithm~\ref{algo:eigen-estimator-averaging}). 
Since all previously phase-aligned local estimates are identical to their corresponding true signal segment by the induction hypothesis in~\eqref{eq:inductive-hypo}, we have the following for each overlap index $k\in\{0,\ldots,K-1\}$:
\begin{equation}
    v_k = f_{b+k}.
\end{equation}
On the other hand, by~\eqref{eq:local-estimates-noiseless}, the unaligned current local estimate $\widehat{f}^{\,(b)}$ satisfies
\begin{equation}
    u_k = \widehat f^{\,\,(b)}_k = \mathrm{e}^{\mathrm{i}\phi_b} f_{b+k}, \qquad k=0,\ldots,K-1.
\end{equation}
Consequently,
\begin{equation}\label{eq:proof-inner-product}
\langle u, v\rangle
= \sum_{k=0}^{K-1} \overline{u_k}\, v_k
= \sum_{k=0}^{K-1} \overline{\mathrm{e}^{\mathrm{i}\phi_b} f_{b+k}}\, f_{b+k}
= \mathrm{e}^{-\mathrm{i}\phi_b} \sum_{k=0}^{K-1} |f_{b+k}|^2.
\end{equation}
Since $(\Leftarrow)$, we are assuming that $f \in \mathcal{S}_K$, that is, $f$ has no run of $K$ consecutive zeros. 
This implies that the overlap segment $(f_b,\ldots,f_{b+K-1})$ cannot be identically zero. 
Therefore,
\begin{equation}
    \sum_{k=0}^{K-1} \left |f_{b+k} \right |^2 > 0,
\end{equation}
and~\eqref{eq:proof-inner-product} together implies that $\langle u,v\rangle \neq 0$ with
\begin{equation}
    \arg(\langle u,v\rangle) \equiv -\phi_b~~(\operatorname{mod} 2\pi).
\end{equation}
The estimator then chooses $\phi = \arg(\langle u,v\rangle)$ and updates the current local estimate via
\begin{equation}
    \widehat f^{\, (b)} \leftarrow \mathrm{e}^{\mathrm{i}\phi}\widehat f^{\,(b)}.
\end{equation}
Combining this with~\eqref{eq:local-estimates-noiseless} and $\phi \equiv -\phi_b$ yields
\begin{equation}
    \widehat f^{\,(b)} = \mathrm{e}^{\mathrm{i}(\phi+\phi_b)} f_{b:b+K} = f_{b:b+K},
\end{equation}
which establishes the induction step and proves~\eqref{eq:inductive-hypo} for all $b' = 0, 1, \ldots, T-K-1$.

Finally, the estimator outputs $\widehat f$ by averaging all phase-aligned local estimates that cover each index $i$. 
Since every phase-aligned local estimate equals the corresponding true signal segment, as shown above, each summand contributing to $\widehat f_i$ is exactly equal to $f_i$. 
Hence, the average is also $f_i$, and, with that, we finally obtain $\widehat f = f$. 
This concludes the proof for $(\Leftarrow)$

Next, the proof for $(\Rightarrow)$ is the same as the proof for $(\Rightarrow)$ of Lemma~\ref{lem:identifiability}.

This concludes the proof.

\section{Proof for Theorem~\ref{thm:exact-recovery-quantum-phaselift-estimator}}
\label{app:proof-exact-recovery-phaselift-estimator}

In this section, we prove the exact recovery for the least-squares estimator. 

Due to the fact that the lifted matrix $Z$ is a PSD matrix, we can take $X = Z$ as one of the feasible solutions for the least-squares problem in~\eqref{eq:least-square-problem}. In the noiseless setting, the objective function of this problem at this $X$ is zero since
\begin{equation}
    \operatorname{Tr}[A_\ell X] = \operatorname{Tr}[A_\ell Z] = b_{\ell}, \quad \text{for all } \ell = 1, \ldots, m.
\end{equation}
This, along with the fact that the objective function $\geq 0$ for all $X$, imply that every minimizer $X^{\star}$ of the least-squares problem satisfies
\begin{equation}
    \operatorname{Tr}[A_\ell X^{\star}] = \operatorname{Tr}[A_\ell Z], \quad \text{for all } \ell = 1, \ldots, m,
\end{equation}
which is equivalent to the entrywise equalities
\begin{equation}
    X^\star_{ii}=Z_{ii}=|f_i|^2,\quad X^\star_{ij}=Z_{ij}=f_i\overline{f_j},\quad\text{whenever }1\le |i-j|\le K\label{eq:entry-agreement}.
\end{equation}

With this in place, we show that the minimizer of the least-squares problem in is unique and that it is $X^{\star} = Z$ under the assumption that $f$ has no run of $K$ consecutive zeros; that is, we prove $(\Rightarrow)$. 
We begin by defining the index set of nonzero entries of $f$
\begin{equation}
    V\coloneqq\{i\in\{0,\dots,T-1\}: f_i\neq 0\}.
\end{equation}
We then construct the graph $G = (V, E)$ with an undirected edge $\{i, j\} \in E$ iff $1 \leq |i-j| \leq K$. 
It is easy to see that $G$ is connected since $f$ has no run of $K$ consecutive zeros.

Next, consider an arbitrary edge $\{i, j\} \in E$. 
We examine the $2 \times 2$ principal submatrix of $X^{\star}$ restricted to these indices:
\begin{equation}
    \begin{pmatrix}
    X^\star_{ii} & X^\star_{ij}\\
    \overline{X}^\star_{ij} & X^\star_{jj}
    \end{pmatrix}
   =
    \begin{pmatrix}
     |f_i|^2 & f_i\overline{f_j}\\
     \overline{f_i} f_j & |f_j|^2
    \end{pmatrix},
\end{equation}
where the equality follows directly from the agreement on the $K$-band~\eqref{eq:entry-agreement}. 
Clearly, the determinant of the above submatrix is zero:
\begin{equation}
    |f_i|^2 |f_j|^2 - |f_i\overline{f_j}|^2 = 0.
\end{equation}
In other words, 
\begin{equation}
    |X^\star_{ij}|^2 = X^\star_{ii} X^\star_{jj}.
\end{equation}
Since $X^{\star} \succeq 0 $ due to the PSD constraint of our least-squares problem, the above equality implies that the corresponding column vectors, that is, $ \text{col}_i(X^\star) $ and $\text{col}_j(X^\star)$, are linearly dependent (see Lemma~\ref{lem:singular-submatrix-column-dependence}). 
Here we denote the $i$-th column vector and the $i$-th row vector of a matrix $A$ by $\text{col}_i(A) $ and $\text{row}_i(A)$, respectively. 
More explicitly, from Lemma~\ref{lem:singular-submatrix-column-dependence}, we can write
\begin{equation}\label{eq:ratio-i-j-col}
    \text{col}_j(X^\star) = \frac{X_{ij}^{\star}}{X_{ii}^{\star}} \text{col}_i(X^\star) = \frac{f_i \overline{f_j}}{|f_i|^2} \text{col}_i(X^\star) =  \frac{\overline{f_j}}{\overline{f_i}} \text{col}_i(X^\star).
\end{equation}

This local dependency can be extended globally. 
Since $G$ is connected, we can iterate relation~\eqref{eq:ratio-i-j-col} along a path connecting a fixed root vertex $r \in V$ to any arbitrary vertex $j \in V$. 
This yields
\begin{equation}\label{eq:all-cols}
    \text{col}_j(X^\star) = \frac{\hphantom{i}\overline{f_j}\hphantom{i}}{\overline{f_r}} \text{col}_r(X^\star).
\end{equation}
For brevity, let us define the vector 
\begin{equation}
    h \equiv \frac{\hphantom{i}1\hphantom{i}}{\overline{f_r}} \text{col}_r(X^\star).
\end{equation}
We can thus rewrite the equality in~\eqref{eq:all-cols} as follows:
\begin{equation}\label{eq:col-factorization}
    \text{col}_j(X^\star) = h \overline{f_j}.
\end{equation}
Also using the fact that $X^{\star}$ is Hermitian, we can derive a corresponding expression for the $i$-th row, $\text{row}_i(X^\star)$:
\begin{equation}\label{eq:row-factorization}
    \text{row}_i(X^\star) = f_i h^{\dagger},
\end{equation}
for all $i \in V$. Hence, for all $i, j \in V$, we have
\begin{equation}
    X^\star_{ij} = h_i \overline{f_j} = f_i \overline{h_j}.
\end{equation}

Now using the fact that we know the $K$-band entries of any minimizer $X^{\star}$ (see~\eqref{eq:entry-agreement}), the following holds for any edge $\{i, j\} \in E$
\begin{equation}
    h_i \overline{f_j} = X^\star_{ij} = f_i \overline{f_j}.
\end{equation}
Since $f_j \neq 0$, we get $h_i = f_i$ for all $i \in V$ that share an edge with $j \in V$ in G. 
This can then be argued for all $i \in V$, so
\begin{equation}\label{eq:h-i-f-i}
    h_i = f_i,
\end{equation}
holds for all $i \in V$.

For indices $i \notin V$, we have $f_i = 0$ and by~\eqref{eq:entry-agreement}, we have $X_{ii}^{\star} = 0$. 
The positive-semidefinite property of $X^{\star}$ then forces $X_{ij}^{\star} = 0$ for all $j$. 
This can be seen from the Cauchy-Schwarz inequality
\begin{equation}
    |X^\star_{ij}|^2 \le X^\star_{ii} X^\star_{jj} = 0 \quad\Rightarrow\quad X^\star_{ij}=0.
\end{equation}
Thus $X^{\star}$ has zero $i$-th row and column for all $i \notin V$.

Combining~\eqref{eq:col-factorization},~\eqref{eq:h-i-f-i}, and the vanishing rows and columns of $X^{\star}$, we obtain
\begin{equation}
    X^\star_{ij} = h_i \overline{f_j} = f_i \overline{f_j}\quad\text{for all }i,j,
\end{equation}
that is, $X^{\star} = ff^{\dagger} = Z$. 
This concludes the proof that $Z$ is the unique minimizer of the least-squares problem in~\eqref{eq:least-square-problem}, given that $f$ has no run of $K$ consecutive zeros and thus proves~$(\Rightarrow)$.

Next, the proof for $(\Rightarrow)$ is the same as the proof for $(\Rightarrow)$ of Lemma~\ref{lem:identifiability}.

This concludes the proof.

\section{Proof for Theorem~\ref{thm:stability-base-case}}\label{app:stability-base-case}

For the purpose of this proof, we introduce some notations. Let the true signal vector be $f = (r_0, r_1 \mathrm{e}^{\mathrm{i}\theta_1}, r_2 \mathrm{e}^{\mathrm{i}\theta_2}, \ldots, r_{T-1} \mathrm{e}^{\mathrm{i}\theta_{T-1}})^{\mathsf{T}}$ and let the signal estimate be $\widehat{f} = (\widehat{r}_0, \widehat{r}_1 \mathrm{e}^{\mathrm{i}\widehat{\theta}_1}, \widehat{r}_2 \mathrm{e}^{\mathrm{i}\widehat{\theta}_2}, \ldots, \widehat{r}_{T-1} \mathrm{e}^{\mathrm{i}\widehat{\theta}_{T-1}})^{\mathsf{T}}$. 
Note that the estimator sets $\widehat r_0 = r_0 = 1$. 

We begin by deriving the reconstruction error bound, which demonstrates how the errors in measuring the diagonal and first off-diagonal entries of $Z$ propagate through the estimator and contribute to the final reconstruction error in $\widehat{f}$. 
Once we have this error bound, we use it to get the bounds on the number of circuit repetitions required to guarantee that the final estimate $\widehat f$ stays within $\eta$ distance of $f$ with probability at least $1-\delta$. 

\subsection{Reconstruction Error Bound}

Consider the following:
\begin{align}
    \left \Vert f - \widehat{f}\, \right \Vert^2_2 & = \sum_{i=0}^{T-1} \left |r_i \mathrm{e}^{\mathrm{i}\theta_i} - \widehat{r}_i \mathrm{e}^{\mathrm{i}\widehat{\theta}_i}  \right|^2 = \sum_{i=1}^{T-1} \left |r_i \mathrm{e}^{\mathrm{i}\theta_i} - \widehat{r}_i \mathrm{e}^{\mathrm{i}\widehat{\theta}_i}  \right|^2\leq 2 \sum_{i=1}^{T-1} \left ( \left |r_i  - \widehat{r}_i\right|^2 + r_i^2 \left |\mathrm{e}^{\mathrm{i}\theta_i} -  \mathrm{e}^{\mathrm{i}\widehat{\theta}_i} \right|^2\right)\label{eq:Eucl-dist-bound},
\end{align}
where the first equality holds due to the fact that we exactly know the first entry of $f$.

To bound the first term of the above inequality, assume that we can estimate all the diagonal elements of $Z$ with an error of at most $\varepsilon >  0$ compared to their true values, that is, 
\begin{equation}
    \max_{i \in \{1, \ldots, T-1\}}\left | Z_{ii} - \widehat{Z}_{ii}\right| \leq \varepsilon,
\end{equation}
holds. Furthermore, assuming that $\gamma = \min_i |f_i|^2 = \min_i Z_{ii} > \varepsilon$, we can apply the above inequality to obtain
\begin{align}
    \left |r_i  - \widehat{r}_i\right| = \frac{\left |r_i^2  - \widehat{r}_i^{\,2}\right|}{r_i + \widehat{r}_i} = \frac{ \left | Z_{ii} - \widehat{Z}_{ii}\right|}{\sqrt{Z_{ii}} + \sqrt{\widehat{Z}_{ii}}} \leq \frac{\varepsilon}{\sqrt{Z_{ii}} + \sqrt{Z_{ii} - \varepsilon}} \leq \frac{\varepsilon}{r_i}\label{eq:bound-on-magnitudes}.
\end{align}

Next, we bound the square root of the second term of~\eqref{eq:Eucl-dist-bound}:
\begin{align}
    r_i \left |\mathrm{e}^{\mathrm{i}\theta_i} -  \mathrm{e}^{\mathrm{i}\widehat{\theta}_i} \right| & =  r_i \left |\frac{\mathrm{e}^{\mathrm{i}(\theta_i-\theta_{i-1})}}{\mathrm{e}^{-\mathrm{i}\theta_{i-1}}} -  \frac{\mathrm{e}^{\mathrm{i}(\widehat{\theta}_i - \widehat{\theta}_{i-1})}}{\mathrm{e}^{-\mathrm{i}\widehat{\theta}_{i-1}}} \right|\\
    & \leq r_i \left(\left |\frac{\mathrm{e}^{\mathrm{i}(\theta_i-\theta_{i-1})}}{\mathrm{e}^{-\mathrm{i}\theta_{i-1}}} -  \frac{\mathrm{e}^{\mathrm{i}(\widehat{\theta}_i - \widehat{\theta}_{i-1})}}{\mathrm{e}^{-\mathrm{i}\theta_{i-1}}} \right| + \left |\frac{\mathrm{e}^{\mathrm{i}(\widehat{\theta}_i - \widehat{\theta}_{i-1})}}{\mathrm{e}^{-\mathrm{i}\theta_{i-1}}} -  \frac{\mathrm{e}^{\mathrm{i}(\widehat{\theta}_i - \widehat{\theta}_{i-1})}}{\mathrm{e}^{-\mathrm{i}\widehat{\theta}_{i-1}}} \right|\right)\\
    & = r_i \left(\left |\mathrm{e}^{\mathrm{i}(\theta_i-\theta_{i-1})} -  \mathrm{e}^{\mathrm{i}(\widehat{\theta}_i - \widehat{\theta}_{i-1})} \right| + \left |\frac{1}{\mathrm{e}^{-\mathrm{i}\theta_{i-1}}} -  \frac{1}{\mathrm{e}^{-\mathrm{i}\widehat{\theta}_{i-1}}} \right|\right)\\
    & = r_i \left(\left |\mathrm{e}^{\mathrm{i}(\theta_i-\theta_{i-1})} -  \mathrm{e}^{\mathrm{i}(\widehat{\theta}_i - \widehat{\theta}_{i-1})} \right| + \left |\mathrm{e}^{\mathrm{i}\theta_{i-1}} -  \mathrm{e}^{\mathrm{i}\widehat{\theta}_{i-1}} \right|\right)\label{eq:error-in-phases-j}.
\end{align}

To bound the first term of the above inequality, we assume that we can estimate all the first off-diagonal entries of $Z$ with an error of at most $\varepsilon' > 0$ compared to their true values, that is,
\begin{equation}
    \max_{i \in \{0, 1, \ldots, T-2\}}\left | Z_{i,i+1} - \widehat{Z}_{i,i+1}\right| \leq \varepsilon',
\end{equation}
holds. 
The above inequality is equivalent to
\begin{align}
    \max_{i \in \{0, 1, \ldots, T-2\}}\left |r_ir_{i+1} \mathrm{e}^{\mathrm{i}(\theta_i - \theta_{i+1})} -  \left | \widehat{Z}_{i,i+1} \right| \mathrm{e}^{\mathrm{i}(\widehat{\theta}_i - \widehat{\theta}_{i+1})} \right| \leq \varepsilon'.
\end{align}
Assuming $\varepsilon' < \varepsilon $, we have the following from Lemma~\ref{lem:phase-diff-z_1-z_2} below: 
\begin{equation}
    \max_{i \in \{0, 1, \ldots, T-2\}} \left |\mathrm{e}^{\mathrm{i}(\theta_i-\theta_{i+1})} -  \mathrm{e}^{\mathrm{i}(\widehat{\theta}_i - \widehat{\theta}_{i+1})} \right| \leq \sqrt{2 - 2 \sqrt{1 - \frac{\varepsilon'^2}{\varepsilon^2}}}\label{eq:element-phase-error-bound}.
\end{equation}

Substituting this in~\eqref{eq:error-in-phases-j}, we get
\begin{equation}
    r_i \left |\mathrm{e}^{\mathrm{i}\theta_i} -  \mathrm{e}^{\mathrm{i}\widehat{\theta}_i} \right| \leq r_i \left(\sqrt{2 - 2 \sqrt{1 - \frac{\varepsilon'^2}{\varepsilon^2}}} + + \left |\mathrm{e}^{\mathrm{i}\theta_{i-1}} -  \mathrm{e}^{\mathrm{i}\widehat{\theta}_{i-1}} \right|\right). 
\end{equation}
Observe that the second term above represents the phase error in the previous estimate, that is, estimate for the $(i-1)$-th element of $f$. 
With that, we can recursively bound this term using the same technique employed for bounding the $i$-th element's phase error:
\begin{align}
    r_i \left |\mathrm{e}^{\mathrm{i}\theta_i} -  \mathrm{e}^{\mathrm{i}\widehat{\theta}_i} \right| & \leq r_i \left(\sqrt{2 - 2 \sqrt{1 - \frac{\varepsilon'^2}{\varepsilon^2}}} + \left |\frac{\mathrm{e}^{\mathrm{i}(\theta_{i-1}-\theta_{i-2})}}{\mathrm{e}^{-\mathrm{i}\theta_{i-2}}} -  \frac{\mathrm{e}^{\mathrm{i}(\widehat{\theta}_{i-1} - \widehat{\theta}_{i-2})}}{\mathrm{e}^{-\mathrm{i}\widehat{\theta}_{i-2}}} \right|\right)\\
    & \leq r_i \left(\sqrt{2 - 2 \sqrt{1 - \frac{\varepsilon'^2}{\varepsilon^2}}} + \left |\frac{\mathrm{e}^{\mathrm{i}(\theta_{i-1}-\theta_{i-2})}}{\mathrm{e}^{-\mathrm{i}\theta_{i-2}}} -  \frac{\mathrm{e}^{\mathrm{i}(\widehat{\theta}_{i-1} - \widehat{\theta}_{i-2})}}{\mathrm{e}^{-\mathrm{i}\theta_{i-2}}} \right| + \left |\frac{\mathrm{e}^{\mathrm{i}(\widehat{\theta}_{i-1} - \widehat{\theta}_{i-2})}}{\mathrm{e}^{-\mathrm{i}\theta_{i-2}}} -  \frac{\mathrm{e}^{\mathrm{i}(\widehat{\theta}_{i-1} - \widehat{\theta}_{i-2})}}{\mathrm{e}^{-\mathrm{i}\widehat{\theta}_{i-2}}} \right|\right)\\
    & \leq r_i \left(\sqrt{2 - 2 \sqrt{1 - \frac{\varepsilon'^2}{\varepsilon^2}}} + \left |\mathrm{e}^{\mathrm{i}(\theta_{i-1}-\theta_{i-2})} -  \mathrm{e}^{\mathrm{i}(\widehat{\theta}_{i-1} - \widehat{\theta}_{i-2})} \right| + \left |\mathrm{e}^{\mathrm{i}\theta_{i-2}} -  \mathrm{e}^{\mathrm{i}\widehat{\theta}_{i-2}} \right|\right).
\end{align}
The middle term can be bounded from above using~\eqref{eq:element-phase-error-bound}:
\begin{align}
    r_i \left |\mathrm{e}^{\mathrm{i}\theta_i} -  \mathrm{e}^{\mathrm{i}\widehat{\theta}_i} \right| \leq r_i \left(2 \sqrt{2 - 2 \sqrt{1 - \frac{\varepsilon'^2}{\varepsilon^2}}} + \left |\mathrm{e}^{\mathrm{i}\theta_{i-2}} -  \mathrm{e}^{\mathrm{i}\widehat{\theta}_{i-2}} \right|\right).
\end{align}
Recursively repeating this bounding technique $(i-2)$ times and squaring on both sides, we finally get
\begin{equation}
r_i^2 \left |\mathrm{e}^{\mathrm{i}\theta_i} -  \mathrm{e}^{\mathrm{i}\widehat{\theta}_i} \right|^2 \leq i^2 r_i^2 \left(2 - 2 \sqrt{1 - \frac{\varepsilon'^2}{\varepsilon^2}}\right).
\end{equation}

Now, using the above bound and the bound in~\eqref{eq:bound-on-magnitudes}, we can finally bound~\eqref{eq:Eucl-dist-bound} from above:
\begin{align}
    \left \Vert f - \widehat{f}\,\right \Vert^2_2 & \leq 2 \sum_{i=1}^{T-1} \left( \left |r_i  - \widehat{r}_i\right|^2 + r_i^2 \left |\mathrm{e}^{\mathrm{i}\theta_i} -  \mathrm{e}^{\mathrm{i}\widehat{\theta}_i} \right|^2 \right)\\
    & \leq \sum_{i=1}^{T-1}\left( \frac{2\varepsilon^2}{r_i ^2} + 2i^2 r_i^2 \left(2 - 2 \sqrt{1 - \frac{\varepsilon'^2}{\varepsilon^2}}\right) \right)\\
    & \leq  \frac{2T\varepsilon^2}{\gamma} + \frac{2T^3 - 3T^2 + T}{3}\left(2 - 2 \sqrt{1 - \frac{\varepsilon'^2}{\varepsilon^2}}\right)\label{eq:error-bound-K-2},
\end{align}
where the final inequality holds due to the following facts:
\begin{enumerate}
    \item $0 < r_i \leq 1$, for all $i \in \{0, 1, \ldots, T-1\}$,
    \item $\gamma = \min_i |f_i|^2 = \min_i r_i^2$.
\end{enumerate}
This completes the reconstruction error bound analysis.

\subsection{Number of Circuit Repetitions}

Recall from the statement of Theorem~\ref{thm:stability-base-case} that each diagonal element $Z_{ii}$ is estimated using $N_{\mathrm{diag}}$ repetitions of the circuit presented in Section~\ref{sec:measurement-diag}, and each real and imaginary part of each first off-diagonal element $Z_{i,i+1}$ is estimated using $N_{\mathrm{off}}$ repetitions of the corresponding circuit presented in Section~\ref{sec:measurement-off-diag}. 
Each repetition yields a Bernoulli sample, so the empirical averages $\widehat Z_{ii}$ and $\widehat Z_{i,i+1}$ concentrate around their true values $Z_{ii}$ and $Z_{i,i+1}$, respectively.

We start from the reconstruction error bound~\eqref{eq:error-bound-K-2}:
\begin{equation}\label{eq:reconstruction-error}
\left\Vert  f - \widehat{f}\, \right\Vert^2_2
\leq \frac{2T\varepsilon^2}{\gamma}
+ \frac{2T^3 - 3T^2 + T}{3}
  \left(2 - 2\sqrt{1 - \frac{\varepsilon'^2}{\varepsilon^2}}\right).
\end{equation}
To ensure the total reconstruction error is below the target threshold $\eta$, we allocate half of the total error budget to each term in~\eqref{eq:reconstruction-error}:
\begin{align}
\frac{2T\varepsilon^2}{\gamma} &\leq \frac{\eta^2}{2}, \label{eq:error-allocation1}\\
\frac{2T^3 - 3T^2 + T}{3}
    \left(2 - 2\sqrt{1 - \frac{\varepsilon'^2}{\varepsilon^2}}\right)
&\leq \frac{\eta^2}{2}. \label{eq:error-allocation2}
\end{align}

From~\eqref{eq:error-allocation1}, we immediately obtain
\begin{equation}\label{eq:eps-bound}
\varepsilon^2 \leq \frac{\gamma \eta^2}{4T}.
\end{equation}

For~\eqref{eq:error-allocation2}, we use the standard bound
\begin{equation}
2 - 2\sqrt{1 - x} \le 2x, \qquad \text{for all } x \in [0,1],
\end{equation}
which gives a sufficient condition:
\begin{equation}\label{eq:epsp-bound}
\frac{2T^3 - 3T^2 + T}{3} \cdot \frac{2 \varepsilon'^2}{\varepsilon^2} \le \frac{\eta^2}{2}.
\end{equation}
Rearranging and substituting the bound on $\varepsilon^2$ from~\eqref{eq:eps-bound}, we obtain
\begin{equation}\label{eq:epsp-final}
\varepsilon'^2 \le \frac{3 \gamma \eta^4}{16T(2T^3 - 3T^2 + T)}.
\end{equation}

We now invoke the max-element concentration lemma
(Lemma~\ref{lem:max-element}) with
\begin{equation}
    D = \{0,\dots,T-1\},
    \qquad
    O = \{(i,i+1): i = 0,\dots,T-2\},
\end{equation}
so that $|D| = T$ and $|O| = T-1$. For the diagonals,
Lemma~\ref{lem:max-element} states that if
\begin{equation}\label{eq:Ndiag-from-lemma}
    N_{\mathrm{diag}} \ge \frac{2}{\varepsilon^2} \log \left(\frac{4T}{\delta}\right),
\end{equation}
then with probability at least $1 - \delta/2$, we have
\begin{equation}
    \max_{i \in D} |\widehat Z_{ii} - Z_{ii}| \le \varepsilon.
\end{equation}
Combining~\eqref{eq:Ndiag-from-lemma} with the choice~\eqref{eq:eps-bound} gives
\begin{equation}\label{eq:Ndiag}
N_{\mathrm{diag}} \ge \frac{8T}{\gamma \eta^2} \log\!\left(\frac{4T}{\delta}\right).
\end{equation}

For the first off-diagonal entries, Lemma~\ref{lem:max-element} gives that if
\begin{equation}\label{eq:Noff-from-lemma}
    N_{\mathrm{off}}\ge \frac{2}{\varepsilon'^2} \log \left(\frac{8(T-1)}{\delta}\right),
\end{equation}
then with probability at least $1 - \delta/2$, we have
\begin{equation}
    \max_{(i, i+1) \in O} |\widehat Z_{i,i+1} - Z_{i,i+1}| \le \varepsilon'
\end{equation}
Combining~\eqref{eq:Noff-from-lemma} with~\eqref{eq:epsp-bound} yields
\begin{equation}\label{eq:Noff}
    N_{\mathrm{off}} \ge \frac{32T(2T^3 - 3T^2 + T)}{3\gamma \eta^4}
    \log\!\left(\frac{8(T-1)}{\delta}\right).
\end{equation}

Combining~\eqref{eq:Ndiag} and~\eqref{eq:Noff} gives the stated bounds in~\eqref{eq:circuit-repetitions-bound}.  
Under these sampling conditions, both terms in~\eqref{eq:reconstruction-error} are at most $\eta^2/2$, ensuring
\begin{equation}
\left\| f - \widehat{f}\,\right\|_2 \le \eta
\end{equation}
with probability at least $1 - \delta$. This completes the proof for Theorem~\ref{thm:stability-base-case}.

\section{Proof for Lemma~\ref{lemma:injectivity-kernel-equivalence}}
\label{app:proof-for-lemma-injectivity-kernel-equivalence}

In this section, we prove that the measurement map $\mathcal{A}$ is injective on the signal set $\mathcal{S}_K$ if and only if $\mathcal{D}(L, Z) \cap \ker(\mathcal{A}) = \{0\}$ for all $Z = ff^{\dagger}$ generated by signals $f \in \mathcal{S}_K$.

($\Rightarrow$) Suppose $\mathcal{A}$ is injective on $\mathcal{S}_K$. Let $\Delta \in \mathcal{D}(L, Z) \cap \ker(\mathcal{A})$. 
By the definition of the descent cone $\mathcal{D}(L, Z)$ in~\eqref{def:descent-cone}, there exists $\tau > 0$ such that $Y = Z + \tau \Delta \succeq 0$. 
Since $\Delta \in \ker(\mathcal{A})$, we have $\mathcal{A}(\Delta) = 0$, which implies $\mathcal{A}(Y) = \mathcal{A}(Z) + \tau \mathcal{A}(\Delta) = \mathcal{A}(Z)$. 
Thus, $Y$ is a positive semidefinite matrix that produces the exact same measurements as the true lifted matrix $Z$. 
By the exact recovery guarantee (Theorem~\ref{thm:exact-recovery-quantum-phaselift-estimator}), the solution to the noiseless problem is unique. 
Therefore, $Y = Z$, which implies $\tau \Delta = 0$. Since $\tau > 0$, we must have $\Delta = 0$.

($\Leftarrow$) Suppose $\mathcal{D}(L, Z) \cap \ker(\mathcal{A}) = \{0\}$ for all $Z = ff^{\dagger}$ with $f \in \mathcal{S}_K$. 
Assume for the sake of contradiction that $\mathcal{A}$ is not injective on $\mathcal{S}_K$. 
Then there exist distinct signals $f, g \in \mathcal{S}_K$ such that their lifted matrices $Z = f f^\dagger$ and $Y = g g^\dagger$ satisfy $Z \neq Y$ but $\mathcal{A}(Z) = \mathcal{A}(Y)$. 
Consider the direction $\Delta = Y - Z$. 
First, we have $Z + 1 \cdot \Delta = Y \succeq 0$, so $\Delta$ is feasible with respect to the first condition in the definition of the cone $\mathcal{D}(L, Z)$. 
Second, we have $L(Z + \Delta) = \|\mathcal{A}(Y) - b\|_2^2$, and in the noiseless case ($b=\mathcal{A}(Z)$), we have $L(Z+\Delta) = 0 = L(Z)$. 
Thus $\Delta$ satisfies the second condition of the cone $\mathcal{D}(L, Z)$. 
Thus, $\Delta \in \mathcal{D}(L, Z)$. 
Furthermore, $\mathcal{A}(\Delta) = \mathcal{A}(Y) - \mathcal{A}(Z) = 0$, so $\Delta \in \ker(\mathcal{A})$. 
This implies $\Delta$ is a non-zero element in the intersection, which contradicts the hypothesis. Thus, $\mathcal{A}$ must be injective.

This concludes the proof.

\section{Proof for Theorem~\ref{thm:stability-proof-least-squares-estimator}}\label{app:stability-proof-least-squares-estimator}

Let $\Delta \coloneqq X^* - Z$. Consider the following:
\begin{equation}
 \left \| \mathcal{A}(\Delta)  \right \|_2  =  \left \| \mathcal{A}(X^*) - \mathcal{A}(Z) \right \|_2. 
\end{equation}
Using the identity $b = \mathcal{A}(Z) + \varepsilon$, we get
\begin{align}
 \left \| \mathcal{A}(\Delta)  \right \|_2  & =  \left \| \mathcal{A}(X^*) - b + \varepsilon \right \|_2\\
 & \leq \left \| \mathcal{A}(X^*) - b\right \|_2 + \left \| \varepsilon \right \|_2\\
 & \leq \left \| \mathcal{A}(Z) - b\right \|_2 + \left \| \varepsilon \right \|_2 \\
 & = \left \| \varepsilon \right \|_2 + \left \| \varepsilon \right \|_2\\
 & = 2\left \| \varepsilon \right \|_2. 
\end{align}
The first inequality follows directly from the triangle inequality, the second inequality follows from the fact that $\left \| \mathcal{A}(X^*) - b\right \|_2 \leq \left \| \mathcal{A}(Z) - b\right \|_2$ holds since $X^*$ is the minimizer and $Z$ is a feasible solution of the least-squares problem~\eqref{eq:least-square-problem}, and finally, the second equality again uses the identity $b = \mathcal{A}(Z) + \varepsilon$. 
Now finally using the definition of the minimum conic singular value $\sigma_{\mathrm{min}}(\mathcal{A}, \mathcal{D}(L, Z))$ in~\eqref{eq:mini-conic-singular-value}, we have $\left \| \mathcal{A}(\Delta) \right \|_2 \geq \sigma_{\mathrm{min}}(\mathcal{A}, \mathcal{D}(L, Z)) \left \| \Delta \right \|_2$ for any feasible $\Delta$. 
For our case, we have $\Delta \coloneqq X^* - Z$ which is feasible since $X^* - Z \in \mathcal{D}(L, Z)$ and therefore, using this fact and using the above inequality, we have
\begin{equation}\label{eq:minimizer-errror-bound}
    \left \|X^* - Z \right\|_2 \leq \frac{2\left \|\varepsilon\right \|_2}{\sigma_{\mathrm{min}}(\mathcal{A}, \mathcal{D}(L, Z))}.
\end{equation}

This inequality provides a bound on the error between the minimizer $X^*$ and the true lifted matrix $Z$. 
Using this, we next bound the error between the signal estimate $\widehat{f}$ and the true signal $f$. 
The least-squares estimator (Algorithm~\ref{algo:least-squares}) extracts $\widehat{f}$ from the principal eigenpair of $X^*$ as $\widehat{f} = \sqrt{\lambda_{\max}} v_{\max}$, up to a global phase. 
To this end, we apply the Davis-Kahan $\sin \Theta$ theorem~\cite{Davis1970}. 
Let $\theta$ be the angle between the principal eigenvector of $Z$, that is, $f/\|f\|_2$ and that of $X^*$, that is, $v_{\max}$. 
Since the spectral gap of $Z$ is $\|f\|_2^2$, we have
\begin{equation}\label{eq:angle-bound}
    |\sin \theta| \leq \frac{2 \left \|X^* - Z \right\|_2}{\left \|f\right \|_2^2}. 
\end{equation}

With that, consider the following:
\begin{align}
    \min_{\phi \in [0, 2\pi)} \left \| \widehat{f} - \mathrm{e}^{\mathrm{i} \phi} f \right \|_2 & = \min_{\phi \in [0, 2\pi)} \left \| \sqrt{\lambda_{\max}} v_{\max} - \mathrm{e}^{\mathrm{i} \phi} \left \| f\right \|_2 \frac{f}{\left \| f\right \|_2} \right \|_2\\
    & = \min_{\phi \in [0, 2\pi)} \left \| \sqrt{\lambda_{\max}} v_{\max} - \left \| f\right \|_2 v_{\max} +  \left \| f\right \|_2 v_{\max} - \mathrm{e}^{\mathrm{i} \phi} \left \| f\right \|_2 \frac{f}{\left \| f\right \|_2} \right \|_2\\
    & \leq \min_{\phi \in [0, 2\pi)} \left \| \sqrt{\lambda_{\max}} v_{\max} - \left \| f\right \|_2 v_{\max} \right \|_2 +  \left \| \left \| f\right \|_2 v_{\max} - \mathrm{e}^{\mathrm{i} \phi} \left \| f\right \|_2 \frac{f}{\left \| f\right \|_2} \right \|_2\\
    & =  \left \| \sqrt{\lambda_{\max}} v_{\max} - \left \| f\right \|_2 v_{\max} \right \|_2 +  \min_{\phi \in [0, 2\pi)} \left \| \left \| f\right \|_2 v_{\max} - \mathrm{e}^{\mathrm{i} \phi} \left \| f\right \|_2 \frac{f}{\left \| f\right \|_2} \right \|_2\\
    & \leq  \left \| v_{\max} \right \|_2 \left | \sqrt{\lambda_{\max}}  - \left \| f\right \|_2  \right | +  \left \| f\right \|_2  \min_{\phi \in [0, 2\pi)}\left \|  v_{\max} - \mathrm{e}^{\mathrm{i} \phi}  \frac{f}{\left \| f\right \|_2} \right \|_2\\
    & = \left | \sqrt{\lambda_{\max}}  - \left \| f\right \|_2  \right | +  \left \| f\right \|_2 \min_{\phi \in [0, 2\pi)} \left \|  v_{\max} - \mathrm{e}^{\mathrm{i} \phi}  \frac{f}{\left \| f\right \|_2} \right \|_2\label{eq:two-terms-error-bound},
\end{align}
where the first inequality follows from the triangle inequality and the fourth equality follows due to the fact that $v_{\max}$ has norm 1.

Using Weyl's equality $\left | \lambda_{\max}  - \left \| f\right \|_2^2  \right | \leq \left\| X^* - Z\right\|_2$, we bound the first term (eigenvalue perturbation) of the above inequality:
\begin{equation}\label{eq:eigenvalue-perturbation}
\left | \sqrt{\lambda_{\max}}  - \left \| f\right \|_2  \right | = \frac{\left | \lambda_{\max}  - \left \| f\right \|_2^2  \right |}{\left | \sqrt{\lambda_{\max}}  + \left \| f\right \|_2  \right |} \leq \frac{\left\| X^* - Z\right\|_2}{ \left \| f\right \|_2}.
\end{equation}

Now, for bounding the second term (eigenvector perturbation), we use the following standard identity:
\begin{equation}
    \min_{\phi \in [0, 2\pi)} \left \|  v_{\max} - \mathrm{e}^{\mathrm{i} \phi}  \frac{f}{\left \| f\right \|_2} \right \|_2 = 2 \sin(\theta/2).
\end{equation}
Using the inequality $2 \sin(\theta/2) \leq \sqrt{2}\sin(\theta)$, which holds for all $\theta \in [0, \pi/2]$, we get
\begin{equation}
    \min_{\phi \in [0, 2\pi)} \left \|  v_{\max} - \mathrm{e}^{\mathrm{i} \phi}  \frac{f}{\left \| f\right \|_2} \right \|_2 \leq \sqrt{2}\sin(\theta).
\end{equation}
Now using the $\sin\theta$-bound, given in~\eqref{eq:angle-bound}, we get
\begin{equation}
    \min_{\phi \in [0, 2\pi)} \left \|  v_{\max} - \mathrm{e}^{\mathrm{i} \phi}  \frac{f}{\left \| f\right \|_2} \right \|_2 \leq  \frac{2\sqrt{2} \left \|X^* - Z \right\|_2}{\left \|f\right \|_2^2}.
\end{equation}

Finally, substituting the eigenvalue perturbation bound~\eqref{eq:eigenvalue-perturbation}, the above eigenvector perturbation bound, and the least-squares minimizer error bound~\eqref{eq:minimizer-errror-bound}, in~\eqref{eq:two-terms-error-bound}, we get
\begin{equation}
    \min_{\phi \in [0, 2\pi)} \left \| \widehat{f} - \mathrm{e}^{\mathrm{i} \phi} f \right \|_2 \leq \frac{\left\| X^* - Z\right\|_2}{ \left \| f\right \|_2} + \left \| f\right \|_2 \left ( \frac{2\sqrt{2} \left \|X^* - Z \right\|_2}{\left \|f\right \|_2^2} \right ) \leq \frac{2(2\sqrt{2} + 1)}{\sigma_{\mathrm{min}}(\mathcal{A}, \mathcal{D}(L, Z)) } \cdot \frac{\left \|\varepsilon\right \|_2}{\left \| f\right \|_2}.
\end{equation}
This concludes the proof.

\section{Proof for Theorem~\ref{thm:stability-spectral-estimation}}
\label{app:spectral-estimation-proof}

Here, we prove the stability of spectral estimation from time-series estimation.

For convenience, we introduce the discrete angular frequencies
\begin{equation}
    \omega_k \coloneqq E_k \Delta, \qquad \text{for all } k = 1, \ldots, r,
\end{equation}
so that~\eqref{eq:discrete-signal-energy} becomes
\begin{equation}\label{eq:signal-complex-form}
    f_n = \sum_{k=1}^r \left | c_k\right|^2 \mathrm{e}^{-\mathrm{i}\omega_k n}, \qquad \text{for all }n = 0, \ldots, T-1
\end{equation}
 The Nyquist--Shannon sampling rate $\Delta \leq \pi/E_{\max}$ ensures that the mapping $\omega \leftrightarrow E_k$ is bijective. 
Furthermore, we consider the most general scenario, where the $\omega_k$ need not coincide with the discrete grid $\{2\pi \ell / T\}_{\ell=0}^{T-1}$. 
The case in which some frequencies do lie on these grid points naturally appears as a special case within this general scenario.

Now let $D \in \mathbb{C}^{T \times T}$ denote the unitary discrete Fourier transform (DFT) matrix with entries
\begin{equation}
    D_{\ell n} = \frac{1}{\sqrt{T}} \mathrm{e}^{\mathrm{i}2\pi\ell n/T},\qquad \text{where } \ell, n =  0, 1, \ldots, T-1.
\end{equation}
With this, we define the DFT of the true signal $f$ as follows:
\begin{equation}
    F[\ell] \coloneqq (Df)[\ell] = \frac{1}{\sqrt{T}}\sum_{n=0}^{T-1} f_n \mathrm{e}^{\mathrm{i}2\pi\ell n/T}, \qquad \ell =  0, 1, \ldots, T-1.\label{eq:dft-signal}
\end{equation}
Similarly, we define it for the signal estimate $\widehat f$:
\begin{equation}
    \widehat F[\ell] \coloneqq (D\widehat f)[\ell] = \frac{1}{\sqrt{T}}\sum_{n=0}^{T-1} \widehat f_n \mathrm{e}^{\mathrm{i}2\pi\ell n/T}, \qquad \ell =  0, 1, \ldots, T-1\label{eq:dft-signal-estimate}.
\end{equation}
Since DFT is unitary and the Euclidean distance is invariant under the action of a unitary, we have
\begin{equation}
    \left \Vert \widehat F - F \right \Vert_2 = \left \Vert D \widehat f - D f \right \Vert_2 =  \left \Vert \widehat f - f \right \Vert_2 \leq \varepsilon\label{eq:error-DFT-domain}.
\end{equation}

Now substituting~\eqref{eq:signal-complex-form} in~\eqref{eq:dft-signal}, we get the following for all $\ell =  0, 1, \ldots, T-1$ :
\begin{align}
     F[\ell] & = \frac{1}{\sqrt{T}}\sum_{n=0}^{T-1} \left(\sum_{k=1}^r \left | c_k\right|^2 \mathrm{e}^{-\mathrm{i}\omega_k n}\right) \mathrm{e}^{\mathrm{i}2\pi\ell n/T}\\
     & = \sum_{k=1}^r \left | c_k\right|^2 \left (\frac{1}{\sqrt{T}}\sum_{n=0}^{T-1}   \mathrm{e}^{-\mathrm{i}(\omega_k n - 2\pi\ell n/T)}\right)\\
     & = \sum_{k=1}^r \left | c_k\right|^2 D_T(\omega_k - 2\pi \ell / T)\label{eq:dft-with-Dirichlet},
\end{align}
where 
\begin{equation}\label{eq:Dirichlet-kernel}
    D_T(x) \coloneqq \frac{1}{\sqrt{T}} \mathrm{e}^{-\mathrm{i} x (T-1)/2} \frac{\sin(Tx/2)}{\sin(x/2)}
\end{equation}
is the Dirichlet kernel.

Before moving forward, we state the following three properties of the Dirichlet kernel that we will later use in our analysis:
\begin{enumerate}
    \item \textbf{Value at the origin.} Taking the limit $x\to 0$ in~\eqref{eq:Dirichlet-kernel} and using
    $\lim_{x\to 0} \frac{\sin(Tx/2)}{\sin(x/2)} = T$ gives
    \begin{equation}
        |D_T(0)| = \sqrt{T}.
        \label{eq:dirichlet-at-zero}
    \end{equation}
    \item \textbf{Uniform upper bound for $x \in [-\pi, \pi]$.} For all $x \in  [-\pi, \pi]$, we have
    \begin{equation}
        \left | D_T(x)\right| = \frac{1}{\sqrt{T}}  \frac{\left | \sin(Tx/2) \right|}{\left | \sin(x/2)\right |} = \frac{1}{\sqrt{T}}  \frac{ \sin(T|x|/2) }{ \sin(|x|/2)} \leq \frac{1}{\sqrt{T} \sin(|x|/2)}\label{eq:uniform-upper-bound-dirichlet-1},
    \end{equation}
    where the second equality simply follows due to the fact that the Dirichlet kernel is symmetric around zero, that is, $D_T(x) = \overline{D_T(-x)}$, and the inequality holds due to $|\sin y | \leq 1$ for all $y \in \mathbb{R}$.
    
    Now since $|x| \leq \pi$, we have $|x|/2 \leq \pi/2$, and on the range $[0, \pi/2]$, the following inequality holds:
    \begin{equation}
        \sin y \geq \frac{2}{\pi} y.
    \end{equation}
    Taking $y = |x|/2$, we obtain
    \begin{equation}
        \sin(|x|/2) \geq \frac{2}{\pi} \cdot \frac{|x|}{2} = \frac{|x|}{\pi}.
    \end{equation}

    Pluggin this in~\eqref{eq:uniform-upper-bound-dirichlet-1}, we get
    \begin{equation}
        \left | D_T(x)\right| \leq  \frac{\pi }{\sqrt{T} |x| }\label{eq:upper-bound-all-x}
    \end{equation}
    for all $|x| \leq \pi$. This means that the magnitude of the Dirichlet kernel decays away from zero like $\sfrac{\pi }{\sqrt{T} |x| }$.

    \item \textbf{Uniform lower bound for $x \in [-\pi/T, \pi/T]$ (region near zero):} Recall from above that
    \begin{equation}
        \left | D_T(x)\right| = \frac{1}{\sqrt{T}}  \frac{ \sin(T|x|/2) }{ \sin(|x|/2)}\label{eq:lower-bound-dirichlet-1}.
    \end{equation}
    For all $|x| \leq \pi/T$, we have $0 \leq T|x|/2 \leq \pi/2$ and $0 \leq |x|/2 \leq \pi/2$. Therefore, we recall the following two inequalities that hold for all $y \in [0, \pi/2]$:
    \begin{align}
        \frac{2}{\pi} y \leq \sin y \leq y.
    \end{align}
    Taking $y = |x|/2$ and using these inequalities in~\eqref{eq:lower-bound-dirichlet-1}, we get
    \begin{equation}
        \left | D_T(x)\right| \geq \frac{2}{\pi} \sqrt{T}\label{eq:lower-bound-dirichlet-2}.
    \end{equation}
    This implies that the Dirichlet kernel maintains a large magnitude for all $x$ sufficiently close to zero.
\end{enumerate}

Now let $\gamma_\omega \in (0, \pi]$ denote the minimum nonzero gap between any two frequencies of $H$:
\begin{equation}\label{eq:freq-gap}
    \gamma_\omega \coloneqq \min_{j\neq k} | \omega_j - \omega_k|~~(\operatorname{mod} 2\pi).
\end{equation}
Note that $\gamma_\omega$ lies in the range $(0, \pi]$ because the angular frequencies $\omega_j$ lie in $(0, 2\pi]$ and we do modulo $2\pi$. 
Furthermore, $\gamma_\omega$ is related to the energy gap $\gamma$ as $\gamma_\omega = \gamma \Delta$. 
We will use this relation at the end of the analysis to get the final conditions in terms of $\gamma$.
Moreover, let $\ell_k$ denote the nearest grid point to $\omega_k$
\begin{equation}
    \ell_k \coloneqq \operatorname{argmin}_{\ell \in \{0, \ldots T-1\}} \left | \omega_k - 2\pi\ell/T\right|.
\end{equation}
By this definition, we have
\begin{equation}\label{eq:nearest-point-distance}
    \left | \omega_k -  2\pi\ell_k/T\right| \leq \frac{\pi}{T}
\end{equation}
since the distance between two consecutive grid points is $2\pi/T$.

With this in place, we now focus on the amplitudes $F[\ell]$ at those grid points $\ell$ that are closest to the true frequencies. 
Our goal is to derive a lower bound on these amplitudes in terms of $T$, showing that as $T$ increases, the corresponding peaks become stronger and more distinguishable. 
This ensures that the grid points near the true frequencies can be reliably identified as the frequency estimates. 
Evaluating~\eqref{eq:dft-with-Dirichlet} at $\ell = \ell_k$, where, as mentioned above, is the nearest grid point to $\omega_k$, we obtain
\begin{equation}
    |F[\ell_k]| = \Big | | c_k |^2 D_T(\omega_k - 2\pi \ell_k / T) + \sum_{j \neq k} | c_j |^2 D_T(\omega_j - 2\pi \ell_k / T) \Big|
\end{equation}
Using the reverse triangle inequality, we get
\begin{align}
    |F[\ell_k]| & \geq \Big | | c_k |^2 D_T(\omega_k - 2\pi \ell_k / T) \Big|  -  \Big |\sum_{j \neq k} | c_j |^2 D_T(\omega_j - 2\pi \ell_k / T) \Big|\\
    & \geq | c_k |^2 | D_T(\omega_k - 2\pi \ell_k / T) |  -  \sum_{j \neq k} | c_j |^2 | D_T(\omega_j - 2\pi \ell_k / T) |\label{eq:eq:two-term-for-points-near-freqs}.
\end{align}

For bounding the first term from below, we use the fact that the distance between $\omega_k$ and its nearest grid point $2\pi \ell_k/T$ is bounded from above by $\pi/T$ and use the uniform lower bound of the Dirichlet kernel for all $x \in [-\pi/T, \pi/T]$ given in~\eqref{eq:lower-bound-dirichlet-2}, to get
\begin{equation}
    | c_k |^2 | D_T(\omega_k - 2\pi \ell_k / T) | \geq  | c_k |^2 \cdot \frac{2}{\pi} \sqrt{T} =  \frac{2| c_k |^2}{\pi}  \sqrt{T}\label{eq:first-term-bound-for-points-near-freqs}.
\end{equation}

We next bound the second term from above. For this, consider the following:
\begin{equation}
    \omega_j - 2\pi \ell_k/T
= (\omega_j - \omega_k) + (\omega_k - 2\pi \ell_k/T).
\end{equation}
Taking the absolute value on both sides and using the reverse triangle inequality, we get
\begin{equation}
    |\omega_j - 2\pi \ell_k/T| \geq |\omega_j - \omega_k| - |\omega_k - 2\pi \ell_k/T|.
\end{equation}
Using the definition of the minimum frequency gap $\gamma_\omega$ (given in~\eqref{eq:freq-gap}) and~\eqref{eq:nearest-point-distance}, we get
\begin{equation}
    |\omega_j - 2\pi \ell_k/T| \geq \gamma_\omega - \pi/T.
\end{equation}
Now, if we assume that $\gamma_\omega/2 \geq \pi/T$, which means that the minimum separation between any two frequencies is large enough so that no two of them share the same nearest grid point, then we have
\begin{equation}
    |\omega_j - 2\pi \ell_k/T| \geq \frac{\gamma_\omega}{2}.
\end{equation}
Setting $x = \omega_j - 2\pi \ell_k/T $ and using the uniform upper bound~\eqref{eq:upper-bound-all-x}, we get
\begin{equation}
    \left | D_T(\omega_j - 2\pi \ell_k/T) \right| \leq \frac{\pi}{\sqrt{T} \cdot \gamma_\omega/2} = \frac{2 \pi}{ \gamma_\omega}\cdot \frac{1}{\sqrt{T}}\label{eq:second-term-bound-for-points-near-freqs}.
\end{equation}
With this, we get an upper bound on the second term as follows:
\begin{equation}
    \sum_{j \neq k} | c_j |^2 \left | D_T(\omega_j - 2\pi \ell_k / T) \right | \leq  \sum_{j \neq k} | c_j |^2 \left ( \frac{2 \pi}{ \gamma_\omega}\cdot \frac{1}{\sqrt{T}} \right) \leq \frac{2 \pi}{ \gamma_\omega}\cdot \frac{1}{\sqrt{T}}, 
\end{equation}
where the last inequality follows because $\sum_{j \neq k} | c_j |^2 \leq 1$.

Using the bound above for the second term and the bound~\eqref{eq:first-term-bound-for-points-near-freqs} for the first term in~\eqref{eq:eq:two-term-for-points-near-freqs}, we obtain
\begin{equation}\label{eq:true-grid-points-bound}
    |F[\ell_k]| \geq \frac{2| c_k |^2}{\pi}  \sqrt{T} - \frac{2 \pi}{ \gamma_\omega}\cdot \frac{1}{\sqrt{T}}
\end{equation}
This gives a consistent "signal minus leakage" bound.

Next, we analyze the amplitudes $F[\ell]$ at grid points $\ell$ that are sufficiently far from all true frequencies. 
Our goal is to show that, at such points, the contributions from all Dirichlet kernels are small, leading to negligible overall amplitude. 
In particular, we derive an upper bound on $|F[\ell]|$ as a function of $T$ and show that it decreases with increasing $T$. 
This result ensures that the grid points distant from the true frequencies produce only weak responses, allowing the true spectral peaks to be clearly distinguished from background leakage or numerical noise.

Let $\ell$ be a grid point that is at least $\gamma_\omega/2$ away from every $\omega_j$, that is, 
\begin{equation}
    \left|\omega_j - 2\pi \ell/T\right| \ge \frac{\gamma_\omega}{2},
    \qquad \text{for all } j=1,\dots,r.
    \label{eq:offgrid-far-bin}
\end{equation}
Then applying~\eqref{eq:upper-bound-all-x} to each summand of~\eqref{eq:dft-with-Dirichlet}, we get
\begin{align}
    \left | F[\ell]\right| & = \left | \sum_{j=1}^r \left | c_j\right|^2 D_T(\omega_j - 2\pi \ell / T) \right| \leq  \sum_{j=1}^r \left | c_j\right|^2 \left | D_T(\omega_j - 2\pi \ell / T) \right| \\
    & \leq \sum_{j=1}^r | c_j |^2 \left ( \frac{2 \pi}{ \gamma_\omega}\cdot \frac{1}{\sqrt{T}} \right)= \frac{2 \pi}{ \gamma_\omega}\cdot \frac{1}{\sqrt{T}}\label{eq:empty-grid-points-bound},
\end{align}
since $\sum_j |c_j|^2 = 1$. This bound clearly shows that increasing $T$ suppresses the amplitudes of the unwanted grid points.

So far, our analysis has focused on the DFT of the true signal $f$. 
However, in practice, we only have access to an estimated signal $\widehat f$ that is within an error $\varepsilon$ of $f$. 
We now turn our attention to the DFT of this estimated signal and examine how its amplitudes differ from those of the true DFT, specifically quantifying how the estimation error $\varepsilon$ propagates into the frequency domain.

From~\eqref{eq:error-DFT-domain}, we have the following for every $\ell$:
\begin{equation}
    \left | \widehat F[\ell] - F[\ell] \right| \leq \varepsilon\label{eq:dft-error-epsilon}.
\end{equation}
For the grid points $\ell$ that are close to some frequency, we combine~\eqref{eq:true-grid-points-bound} with this to get
\begin{equation}
    \left | \widehat F[\ell]  \right| \geq \frac{2| c_k |^2}{\pi}  \sqrt{T} - \frac{2 \pi}{ \gamma_\omega}\cdot \frac{1}{\sqrt{T}} - \varepsilon. 
\end{equation}
Similarly, for the grid points $\ell$ that are not close to any frequencies, we combine~\eqref{eq:empty-grid-points-bound} with~\eqref{eq:dft-error-epsilon} to get
\begin{equation}
    \left | \widehat F[\ell]  \right| \leq \frac{2 \pi}{ \gamma_\omega}\cdot \frac{1}{\sqrt{T}} + \varepsilon. 
\end{equation}
Hence, the peaks of the desired grid points are distinguishable from the ones that are not as soon as the following holds:
\begin{equation}
    \frac{2 \min_k | c_k |^2}{\pi}  \sqrt{T} - \frac{2 \pi}{ \gamma_\omega}\cdot \frac{1}{\sqrt{T}} - \varepsilon > \frac{2 \pi}{ \gamma_\omega}\cdot \frac{1}{\sqrt{T}} + \varepsilon.
\end{equation}
This gives us our detectability condition if there is a signal estimation error $\varepsilon$:
\begin{equation}
    \frac{2 \min_k | c_k |^2}{\pi}  \sqrt{T} > \frac{4 \pi}{\gamma_\omega}\cdot \frac{1}{\sqrt{T}} + 2\varepsilon.
\end{equation}
In terms of the energy gap $\gamma$, we have
\begin{equation}
    \frac{2 \min_k | c_k |^2}{\pi}  \sqrt{T} > \frac{4 \pi}{\gamma \Delta}\cdot \frac{1}{\sqrt{T}} + 2\varepsilon.
\end{equation}

With this setup, we now transition from the angular frequency domain to the eigenvalue domain, since our final goal is to estimate the eigenvalues with a target precision of $\delta$. 
The grid spacing in the angular frequency domain is $2\pi/T$, and since $\omega_k = E_k \Delta$, the induced grid spacing in the eigenvalue domain is
\begin{equation}
    \frac{2\pi}{T \Delta}.
\end{equation}
Thus, if we want eigenvalue precision $\delta > 0$, it suffices to choose $T$ so that
\begin{equation}
    \frac{2\pi}{T \Delta} \le \delta
    \qquad \Longleftrightarrow \qquad
    T \ge \frac{2\pi}{\delta \Delta}.
    \label{eq:offgrid-eigenvalue-grid}
\end{equation}
Then each true $\omega_k$ is within $\pi/T$ of some grid point $2 \pi \ell_k/T$, and we can define the coarse estimate
\begin{equation}
\widehat{E}_k \coloneqq \frac{2\pi \ell_k}{T \Delta},
\end{equation}
which obeys $|\widehat{E}_k - E_k| \le \delta$.
This concludes the proof.

\section{Additional Numerical Simulations}\label{app:extranumerics}

\subsection{Recovery Performance as a Function of the Total Number of Shots}

In this subsection, we continue the study of how the reconstruction error scales with the total number of measurement shots, using the second and third datasets explained in Section \ref{sec:numerical-simulations}. 
Recall that these datasets consist of signals generated using a $2 \times 3$ Fermi-Hubbard model and a $4 \times 3$ transverse-field Ising model. 

\begin{figure}[!ht]
\begin{subfigure}[h]{\textwidth}
   \centering
\includegraphics[width=\linewidth]{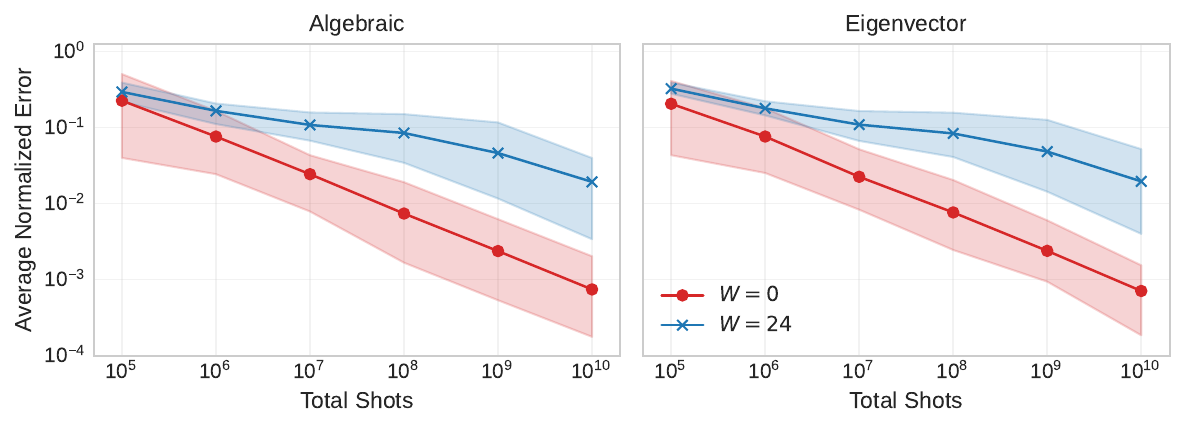}
\caption{$K=2$ }
\end{subfigure}
\begin{subfigure}[h]{\textwidth}
\includegraphics[width=\linewidth]{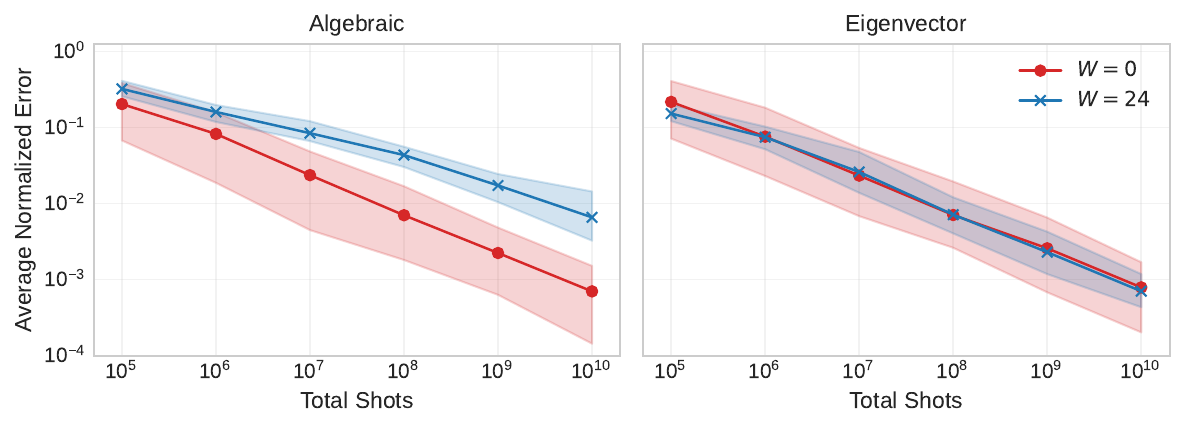}
\caption{$K=W+2$}
\end{subfigure}
\caption{ \textbf{$2\times 3$ Fermi-Hubbard Model.} For signals with a Fock state as input, $T=50$ and randomly chosen Hamiltonian parameters. The signals are grouped into two sets of $10$ random signals labeled by $W =0$ or $W=24$. This grouping is done according to whether they contain either $0$ or $24$ consecutive entries below a noise threshold of $0.1$.}
    \label{fig:Ksig-vs-Kmeas-realistic-FH23}
\end{figure}

Figure \ref{fig:Ksig-vs-Kmeas-realistic-FH23} illustrates this behavior for the $2 \times 3$ Fermi–Hubbard model with a simple Fock state as input and signal length $T = 50$. 
Again, we focus on two representative families of signals,
characterized by intrinsic widths $W = 0$ and $W = 24$, respectively.
We generate random signals for each dataset by randomly sampling the parameters of each Hamiltonian from the intervals previously specified. We sort them into groups labeled by $W$ according to whether they contain a run of $W$ consecutive entries whose absolute value is smaller than $\chi = 0.1$. The majority of these signals have $W=0$ at $\chi=0.1$. We generated random signals until we had $10$ random signals with $W=24$ to compare against $10$ random signals with $W=0$.
The estimator is run on $5$ noisy versions of every signal in each $W$ family at a given shot budget. We compute the normalized Euclidean reconstruction error for all of these instances. For each family, we compute the mean, maximum and minimum of these errors, with the mean reconstruction error shown in a solid line and the minimum and maximum indicated by the shaded region in Figure \ref{fig:Ksig-vs-Kmeas-realistic-FH23} and Figure \ref{fig:Ksig-vs-Kmeas-realistic-ising}.
Figure \ref{fig:Ksig-vs-Kmeas-realistic-FH23}(a) corresponds to a fixed measurement bandwidth $K = 2$ and in Figure \ref{fig:Ksig-vs-Kmeas-realistic-FH23}(b), we choose the measurement bandwidth adaptively as $K = W + 2$.

\begin{figure}[!ht]
\begin{subfigure}[h]{\textwidth}
\includegraphics[width=\linewidth]{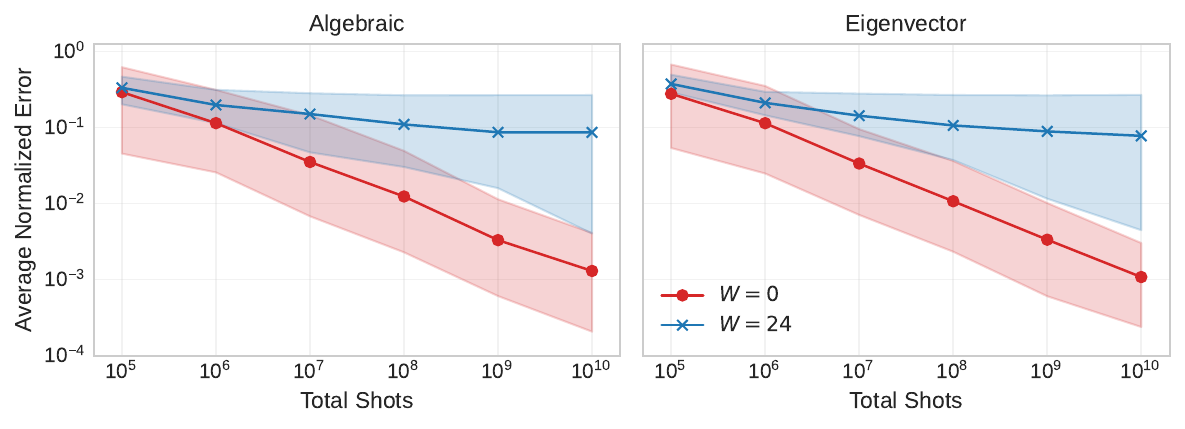}
\caption{$K=2$ }
\end{subfigure}
\begin{subfigure}[h]{\textwidth}
   \centering
\includegraphics[width=\linewidth]{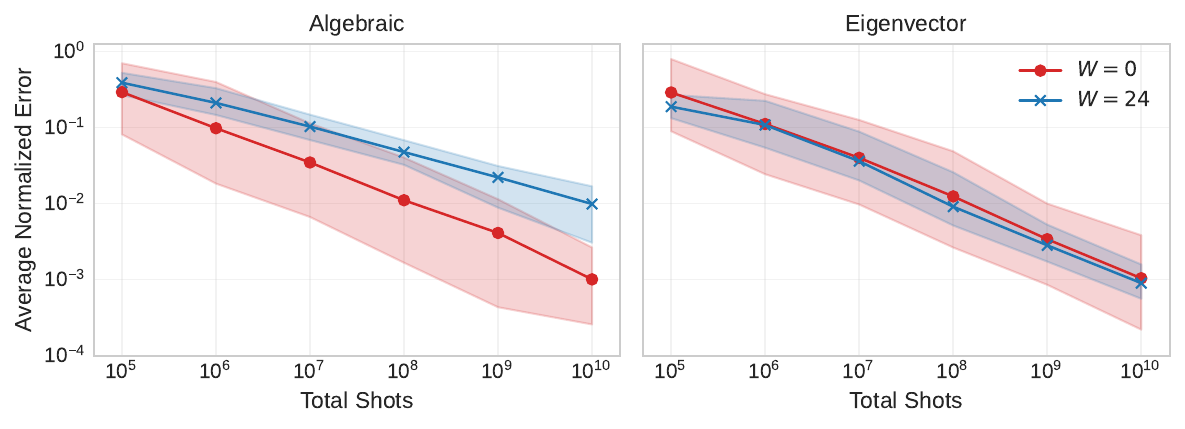}
\caption{$K=W+2$}
\end{subfigure}
\caption{ \textbf{$4\times 3$ Transverse-Field Ising Model.} For signals with the uniform superposition of computational basis states as input, $T=50$ and randomly chosen Hamiltonian parameters. The signals are grouped into two sets of $10$ random signals labeled by $W =0$ or $W=24$. This grouping is done according to whether they contain either $0$ or $24$ consecutive entries below a noise threshold of $0.1$. }
    \label{fig:Ksig-vs-Kmeas-realistic-ising}
\end{figure}

In Figure \ref{fig:Ksig-vs-Kmeas-realistic-ising} we perform the same analysis for a $4 \times 3$ transverse-field Ising model with an input state consisting of the uniform superposition over all computational basis states and signal length $T = 50$. 
Figure \ref{fig:Ksig-vs-Kmeas-realistic-ising}(a) corresponds to a fixed measurement bandwidth $K = 2$ and in Figure \ref{fig:Ksig-vs-Kmeas-realistic-ising}(b) we choose the measurement bandwidth adaptively as $K = W + 2$. 

For both Figure \ref{fig:Ksig-vs-Kmeas-realistic-ising}(a) and Figure \ref{fig:Ksig-vs-Kmeas-realistic-FH23}(a) signals with $W = 0$ satisfy the identifiability condition $K \geq W + 1$, and their reconstruction error decreases
steadily as the total number of shots increases. 
In contrast, signals with $W = 24$ violate this condition, and their reconstruction error again exhibits a pronounced plateau. 
This is further evidence that increasing the number of shots cannot compensate for insufficient band information. 

For both Figure \ref{fig:Ksig-vs-Kmeas-realistic-ising}(b) and Figure \ref{fig:Ksig-vs-Kmeas-realistic-FH23}(b), we choose $K= W+2$ and guarantee that the identifiability condition is satisfied for both signal families from the outset, where we start from $10^5$ shots. 
In this regime, the reconstruction error decreases with the total number of shots for all signals, confirming that once $K$ is chosen appropriately relative to $W$ at a particular high noise threshold $\chi$, any additional measurements directly translate into improved reconstruction accuracy.

Observe that the reconstruction of signals in both cases is similar to or even slightly better than the one we saw previously in Figure~\ref{fig:Ksig-vs-Kmeas-realistic-FH22}, probably due to the interplay of input state and complexity of the Hamiltonian leading to a simpler time-series.

\begin{figure}[!ht]

   \centering
\includegraphics[width=\linewidth]{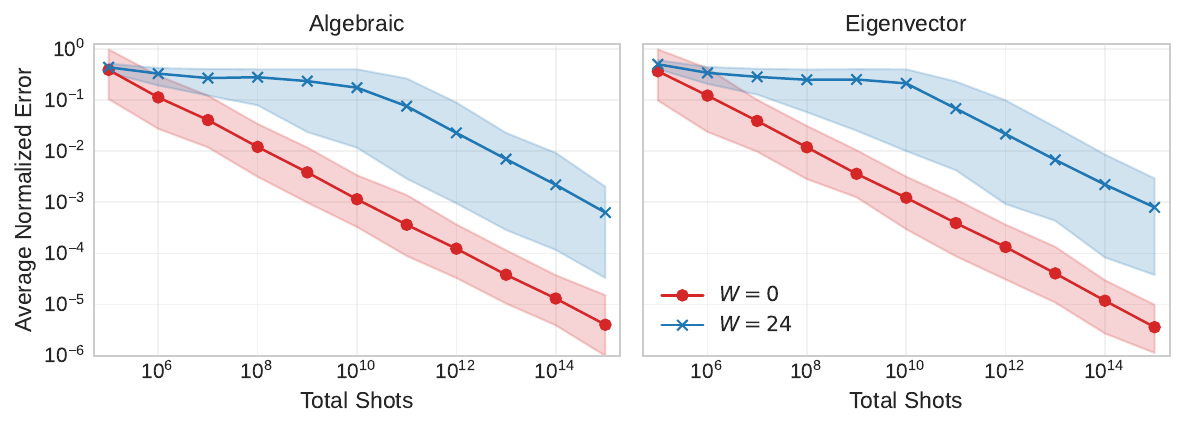}
\caption{ \textbf{$2\times 2$ Fermi-Hubbard Model.} For signals with a coherent Gibbs state with $\beta=0.5$ as input, $T=50$ and randomly chosen Hamiltonian parameters. The signals are grouped into two sets of each with one example signal and labeled by $W =0$ and $W=24$. This grouping is done according to whether they contain either $0$ or $24$ consecutive entries below a noise threshold of $0.1$. Here we set $K=2$ for both signals.}
    \label{fig:flat-to-decay}
\end{figure}

\subsection{From Failed Recovery to Error Decaying with the Number of Shots.}

In our previous analysis (Figure~\ref{fig:Ksig-vs-Kmeas-realistic-FH22}(b),~\ref{fig:Ksig-vs-Kmeas-realistic-FH23}(b), and~\ref{fig:Ksig-vs-Kmeas-realistic-ising}(b)), we showed that for a group of signals with some intrinsic width $W$ at $\chi=0.1$, it is sufficient to choose $K=W+2$ in order to ensure that reconstruction error does not plateau for that entire group of signals. 
Now in Figure \ref{fig:flat-to-decay} we randomly pick a single signal (instead of averaging over a set) from the $2 \times 2$ Fermi-Hubbard dataset with $W=0$ and one signal with $W=24$ and importantly fix $K=2$ for both.
We see that we can increase the number of shots and eventually overcome the plateau for the signal with $W=24$ at $\chi =0.1$. 
This is due to the fact that the initial label $W$ assigned to the signal relates to the choice threshold noise $\chi =0.1$. 
As we increase the number of shots, the shot noise decreases, which would make the number of continuous points under the shot noise value, which we label as $W^*$, lower than the initially assigned value of $W$. 
Once we reach the situation where $K\geq W^*+1$ we shift to an efficient recovery regime where the error starts decreasing as the total number of shots increase, moving away from a plateau regime. 
To show this effect, it was crucial to choose one signal from each group, as the point at which the decay sets in depends on the structure of the individual signal. This is show in Figure \ref{fig:flat-to-decay} where we choose two example signals from the first data set, one with $W=0$ at $\chi=0.1$ and one with $W=24$ at $\chi =0.1$. We fix $K=2$ and increase the number of shots until the reconstruction error decays for both signals.The estimator is run on $150$ noisy versions of each signal in each $W$ family at a given shot budget. We compute the normalized Euclidean reconstruction error for all of these instances. We compute the mean, maximum, and minimum of these errors, with the mean reconstruction error shown in a solid line and the minimum and maximum indicated by the shaded region in Figure \ref{fig:flat-to-decay}.

\subsection{Fixed Shots Budget Analysis}

In this subsection, we provide additional numerics on the trade-off 
for varying the bandwidth $K$ used for reconstruction at a fixed total number of shots.
In Figure \ref{fig:fix-shots-budget-analysis-example-signal}, we show an example signal from the $4 \times 3$ transverse-field Ising model dataset. 
We use our block-by-block algebraic and eigenvector estimators to reconstruct the underlying signal, and again fix the total number of shots to be approximately one million, which are then distributed uniformly across the $K$-band. 
In Figure \ref{fig:fix-shots-budget-analysis-example-signal}(a) we set $K=W+1=2$ (as $W=1$) and in Figure \ref{fig:fix-shots-budget-analysis-example-signal}(b) we set $K=W+4=5$ (as $W=1$). 
We observe good qualitative agreement between the true and reconstructed signal in all plots. 
In both cases, we see that the eigenvector estimator outperforms the algebraic estimator, particularly in the second case.

\begin{figure}[!ht]
\begin{subfigure}[h]{\textwidth}
   \centering
\includegraphics[width=\linewidth]{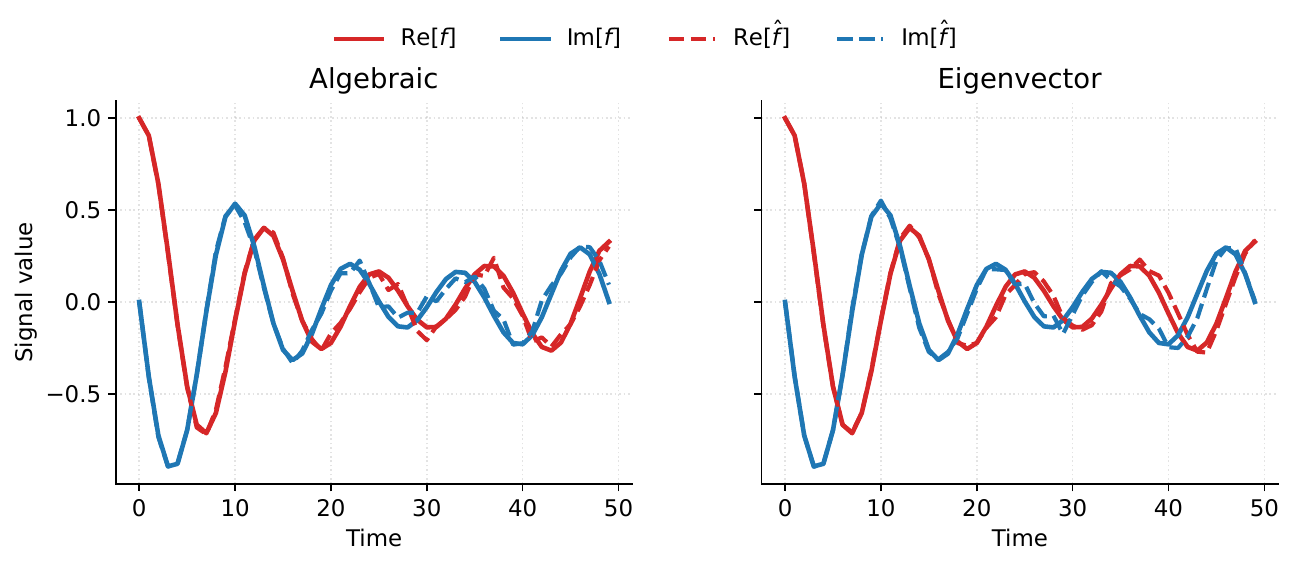}
\caption{Case with $W=1$ and $K = W+1 = 2$}
\end{subfigure}
\begin{subfigure}[h]{\textwidth}
\includegraphics[width=\linewidth]{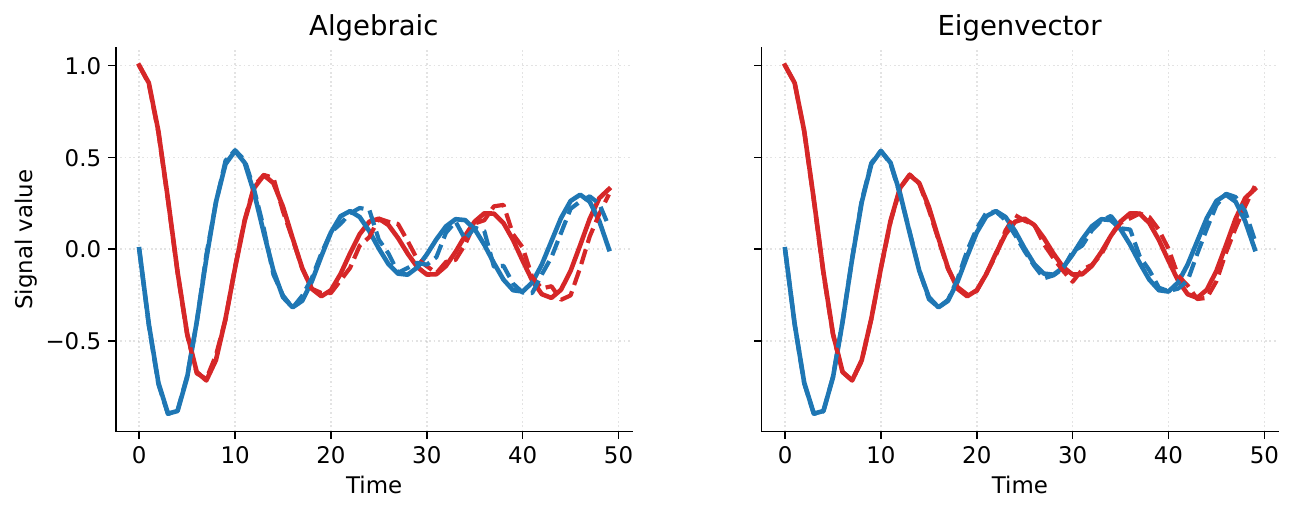}
\caption{Case with $W=1$ and $K = W+4 = 5$}
\end{subfigure}
\caption{\textbf{Fixed Shots Budget Analysis (Example Signal).} Recovery of an example signal corresponding to the $4 \times 3$ transverse-field Ising model using our block-by-block algebraic and eigenvector estimators, with the total number of measurement shots fixed to approximately one million and distributed uniformly across the measured entries of the $K$-band of the lifted matrix $Z$.}
    \label{fig:fix-shots-budget-analysis-example-signal}
\end{figure}

\clearpage
\phantomsection
\addcontentsline{toc}{section}{References}
\bibliographystyle{alpha}
\bibliography{Ref}

\end{document}